\documentclass[12pt]{article}
\usepackage{amsmath,amsthm}
\usepackage{amssymb}
\usepackage{eufrak}
\usepackage{pb-diagram,lamsarrow,pb-lams}

\newcommand{\ar}{\renewcommand{\arraystretch}{1}} 
\DeclareMathAlphabet{\bb}{U}{msb}{m}{n} \gdef\C{\bb C} \gdef\dZ{\bb
Z}   \gdef\dS{\bb S} \gdef\R{\bb R}
\gdef\K{\bb K} \gdef\BH{\bb H} \gdef\F{\bb F} 

\DeclareMathOperator{\End}{End} \DeclareMathOperator{\spin}{{\bf
Spin}}

\DeclareMathOperator{\Ker}{Ker} 
  
\DeclareMathOperator{\Sym}{Sym} 
 
 \DeclareMathOperator{\Mat}{Mat}
 
\DeclareMathOperator{\Tr}{Tr} \DeclareMathOperator{\SL}{SL}
\DeclareMathOperator{\SO}{SO}\DeclareMathOperator{\SU}{SU}

\newcommand{\bcirc}{\raisebox{0.5mm}{$\scriptstyle\bigcirc$}}

\newcommand{\re}{\mbox{\rm Re}\,}
\newcommand{\im}{\mbox{\rm Im}\,}

\newcommand{\bi}{{\bf i}}
\newcommand{\cA}{\mathcal{A}}

\newcommand{\bsH}{{\boldsymbol{\sf H}}}
\newcommand{\bsZ}{{\boldsymbol{\sf Z}}}

\newcommand{\bZ}{{\bf Z}}

\newcommand{\cl}{C\kern -0.2em \ell}

\newcommand{\p}{\prime}
\newcommand{\e}{\mbox{\bf e}}

\newtheorem{thm}{Theorem}
\begin{document}
\title{Spinor Structure and Modulo 8 Periodicity}
\author{V.~V. Varlamov\thanks{Siberian State Industrial University,
Kirova 42, Novokuznetsk 654007, Russia, e-mail:
vadim.varlamov@mail.ru}}
\date{}
\maketitle
\begin{abstract}
Spinor structure is understood as a totality of tensor products of biquaternion algebras, and the each tensor product is associated with an irreducible representation of the Lorentz group. A so-defined algebraic structure allows one to apply modulo 8 periodicity of Clifford algebras on the system of real and quaternionic representations of the Lorentz group. It is shown that modulo 8 periodic action of the Brauer-Wall group generates modulo 2 periodic relations on the system of representations, and all the totality of representations under this action forms a self-similar fractal structure. Some relations between spinors, twistors and qubits are discussed in the context of quantum information and decoherence theory.
\end{abstract}
\section{Introduction}
A geometric description of space-time continuum via the representation of time by an imaginary coordinate of a four-dimensional pseudo-Euclidean space, which was first given by Minkowski, can be considered as an original point for subsequent programme of \emph{geometrization} of physics initiated in the works of Poincar\'{e}, Einstein, Klein (theories of special and general relativity), Weyl, Kaluza, Eddington (multidimensional generalizations of relativistic theories, unification of gravitation and electromagnetism), Wheeler, Connes and other (geometrodynamics, noncommutative geometry and so on). On the other hand, along with the geometry algebraic methods, revealing the unity of mathematics, were penetrated into physics, first of all, methods of group theory (Lorentz and Poincar\'{e} groups, Erlangen programme, group-theoretical description of quantum mechanics and so on). In such a way, the programme of geometrization of physics involves a programme of \emph{algebraization} of physics. In this programme Clifford algebras and theory of hypercomplex structures as a whole play an essential role. A wide application of Clifford algebras and spinors in physics began with the famous Dirac's work on electron theory \cite{Dir28}, where well-known $\gamma$-matrices form a basis of the complex Clifford algebra $\C_4$. As is known, a fundamental notion of \textit{antimatter} follows directly from the Dirac equation which presents by itself a \textit{spinor equation}. Later on, Laport and Uhlenbeck \cite{LU31}, using van der Waerden 2-spinor formalism \cite{Wae29}, wrote Maxwell equations in a \textit{spinor form}. It allows one to consider all the relativistic wave equations on an equal footing, from the one algebraic (group-theoretical) viewpoint \cite{Var022}. Clifford algebras form the foundation of twistor approach. In accordance with Penrose twistor programme \cite{Pen77,PM72}, space-time continuum is a derivative construction with respect to \textit{underlying spinor structure}. Spinor structure contains in itself pre-images of all basic properties of classical space-time, such as dimension, signature, metrics and many other. In parallel with twistor approach, decoherence theory \cite{JZKGKS} claims that in the background of reality we have a \textit{nonlocal quantum substrate} (quantum domain), and all visible world (classical domain) arises from quantum domain in the result of decoherence process. In this context Clifford algebras should be understood as a mathematical tool working on the level of quantum domain.

The present paper is a continuation of the previous work \cite{Var12}. In this paper we study applications of modulo 8 periodicity of Clifford algebras to particle representations of the group $\spin_+(1,3)$ in more details. The underlying spinor structure is understood here as an algebraic structure (tensor products of Clifford algebras) associated with the each finite-dimensional representation of $\spin_+(1,3)$. In accordance with the Wigner interpretation of elementary particles \cite{Wig39}, a so-defined spinor structure contains pre-images of all basic properties of elementary particles, such as spin, mass \cite{Var1402}, charge (a pseudoautomorphism $\cA\rightarrow\overline{\cA}$, \cite{Var04,Var1401}), space inversion $P$, time reversal $T$ and their combination $PT$ (fundamental automorphisms $\cA\rightarrow\cA^\star$, $\cA\rightarrow\widetilde{\cA}$, $\cA\rightarrow\widetilde{\cA^\star}$ of Clifford algebras \cite{Var01,Var11}). It is shown that modulo 8 periodicity of underlying spinor structure generates modulo 2 periodic relations on the system of representations (real and quaternionic) of the group $\spin_+(1,3)$). In consequence of the action of the Brauer-Wall group $BW_{\R}$ we have a self-similar fractal structure on the system of representations of $\spin_+(1,3)$, where the each period of this fractal structure is generated by the cycle of $BW_{\R}$. It is shown also that modulo 8 action of $BW_{\R}$ induces modulo 4 periodic relations on the idempotents groups of the Clifford algebras. Some relations between spinors, twistors and qubits are discussed with respect to theory of quantum information.
\section{Periodicity of Clifford algebras}
As is known, for the Clifford algebra $\cl_{p,q}$ over
the field $\F=\R$ there are isomorphisms
$\cl_{p,q}\simeq\End_{\K}(I_{p,q})\simeq\Mat_{2^m}(\K)$, where
$m=(p+q)/2$, $I_{p,q}=\cl_{p,q}f$ is a minimal left ideal of
$\cl_{p,q}$, and $\K=f\cl_{p,q}f$ is a division ring of $\cl_{p,q}$.
A primitive idempotent of the algebra $\cl_{p,q}$ has the form
\[
f=\frac{1}{2}(1\pm\e_{\alpha_1})\frac{1}{2}(1\pm\e_{\alpha_2})\cdots\frac{1}{2}
(1\pm\e_{\alpha_k}),
\]
where $\e_{\alpha_1},\e_{\alpha_2},\ldots,\e_{\alpha_k}$ are
commuting elements with square 1 of the canonical basis of
$\cl_{p,q}$ generating a group of order $2^k$, that is,
$(\e_{\alpha_1},\e_{\alpha_2},\ldots,\e_{\alpha_k})\simeq\left(\dZ_2\right)^{\otimes
k}$, where $\left(\dZ_2\right)^{\otimes
k}=\dZ_2\otimes\dZ_2\otimes\cdots\otimes\dZ_2$ ($k$ times) is an
Abelian group. The values of $k$ are defined by the formula
$k=q-r_{q-p}$, where $r_i$ are the Radon-Hurwitz numbers
\cite{Rad22,Hur23}, values of which form a cycle of the period 8:
$r_{i+8}=r_i+4$. The values of all $r_i$ are
\begin{center}
\begin{tabular}{lcccccccc}
$i$  & 0 & 1 & 2 & 3 & 4 & 5 & 6 & 7\\ \hline $r_i$& 0 & 1 & 2 & 2 &
3 & 3 & 3 & 3
\end{tabular}.
\end{center}\begin{sloppypar}\noindent
In terms of finite groups we have here \emph{an idempotent group}
$T_{p,q}(f)\simeq\left(\dZ_2\right)^{\otimes (k+1)}$ of the order
$2^{k+1}=2^{1+q-r_{q-p}}$ \cite{AF02}.\end{sloppypar}

For example, let us consider a minimal left ideal of the de Sitter algebra $\cl_{4,1}$ associated with the space $\R^{4,1}$. This ideal has the form
\begin{equation}\label{IdealS}
I_{4,1}=\cl_{4,1}f_{4,1}=\cl_{4,1}\frac{1}{2}(1+\e_0)\frac{1}{2}
(1+i\e_{12}).
\end{equation}
In its turn, for the space-time algebra $\cl_{1,3}$ we have
\[
I_{1,3}=\cl_{1,3}f_{13}=\cl_{1,3}\frac{1}{2}(1+\e_0).
\]
Further, for the Dirac algebra there are isomorphisms
$\cl_{4,1}\simeq\C_4=\C\otimes\cl_{1,3}\simeq\Mat_2(\C_2)$, $\cl^+_{4,1}\simeq\cl_{1,3}\simeq\Mat^{\BH}_2(\cl_{1,1})$. Using an identity $\cl_{1,3}f_{13}=\cl^+_{1,3}f_{13}$ \cite{FRO90a},
we obtain for the minimal left ideal of $\cl_{4,1}$ the following expression:
\begin{multline}\label{2e1}
I_{4,1}=\cl_{4,1}f_{41}=(\C\otimes\cl_{1,3})f_{41}\simeq\cl^+_{4,1}f_{41}
\simeq\cl_{1,3}f_{41}=\\
\cl_{1,3}f_{13}\frac{1}{2}(1+i\e_{12})=\cl^+_{1,3}f_{13}\frac{1}{2}
(1+i\e_{12}).
\end{multline}
Let $\Phi\in\cl_{4,1}\simeq\Mat_4(\C)$ be a Dirac spinor and let $\phi\in \cl^+_{1,3}\simeq\cl_{3,0}$ be a Dirac-Hestenes spinor. Then from (\ref{2e1}) we have the relation between spinors
$\Phi$ and $\phi$:
\begin{equation}\label{2e2}
\Phi=\phi\frac{1}{2}(1+\e_0)\frac{1}{2}(1+\bi\e_{12}).
\end{equation}
Since $\phi\in\cl^+_{1,3}\simeq\cl_{3,0}$, the Dirac-Hestenes spinor can be represented by a biquaternion number
\begin{equation}\label{2e3}
\phi=a^0+a^{01}\e_{01}+a^{02}\e_{02}+a^{03}\e_{03}+
a^{12}\e_{12}+a^{13}\e_{13}+a^{23}\e_{23}+a^{0123}\e_{0123}.
\end{equation}
It should be noted that Hestenes programme of reinterpretation of quantum mechanics within the field of real numbers gives rise many interesting applications of Dirac-Hestenes fields in geometry and general theory of relativity
\cite{HS84,RF90,RSVL,Keller,Var99a}.

Over the field $\F=\R$ there are eight different types of Clifford algebras $\cl_{p,q}$ with the following division ring structure.\\[0.3cm]
{\bf I}. Central simple algebras.
\begin{description}
\item[1.] Two types $p-q\equiv 0,2\pmod{8}$ with a division ring
$\K\simeq\R$.
\item[2.] Two types $p-q\equiv 3,7\pmod{8}$ with a division ring
$\K\simeq\C$.
\item[3.] Two types $p-q\equiv 4,6\pmod{8}$ with a division ring
$\K\simeq\BH$.
\end{description}
{\bf II}. Semi-simple algebras.
\begin{description}
\item[4.] The type $p-q\equiv 1\pmod{8}$ with a double division ring
$\K\simeq\R\oplus\R$.
\item[5.] The type $p-q\equiv 5\pmod{8}$ with a double quaternionic
division ring $\K\simeq\BH\oplus\BH$.
\end{description}

The \textit{spinorial chessboard} \cite{BT88} (see Fig.\,1) is the set of 64 real algebras
\[
\{\cl_{p,q},\;|\;0\leq\,p,q\,\leq 7\},
\]
where it is understood that $\cl_{0,0}\simeq\R$.
\begin{figure}[ht]
\unitlength=1mm
\begin{center}
\begin{picture}(100,100)
\put(0,0){\circle*{2.5}}
\put(0.5,-4){$\cl_{0,0}$}
\put(1.25,0){\line(1,0){7.5}}
\put(0,1.25){\line(0,1){7.5}}
\put(10,0){\circle{2.5}}
\put(10.5,-4){$\cl_{1,0}$}
\put(11.25,0){\line(1,0){7.5}}
\put(10,1.25){\line(0,1){7.5}}
\put(20,0){\circle*{2.5}}
\put(20.5,-4){$\cl_{2,0}$}
\put(21.25,0){\line(1,0){7.5}}
\put(20,1.25){\line(0,1){7.5}}
\put(30,0){\circle{2.5}}
\put(30.5,-4){$\cl_{3,0}$}
\put(31.25,0){\line(1,0){7.5}}
\put(30,1.25){\line(0,1){7.5}}
\put(40,0){\circle*{2.5}}
\put(40.5,-4){$\cl_{4,0}$}
\put(41.25,0){\line(1,0){7.5}}
\put(40,1.25){\line(0,1){7.5}}
\put(50,0){\circle{2.5}}
\put(50.5,-4){$\cl_{5,0}$}
\put(51.25,0){\line(1,0){7.5}}
\put(50,1.25){\line(0,1){7.5}}
\put(60,0){\circle*{2.5}}
\put(60.5,-4){$\cl_{6,0}$}
\put(61.25,0){\line(1,0){7.5}}
\put(60,1.25){\line(0,1){7.5}}
\put(70,0){\circle{2.5}}
\put(70.5,-4){$\cl_{7,0}$}
\put(70,1.25){\line(0,1){7.5}}
\put(0,10){\circle{2.5}}
\put(0.5,6){$\cl_{0,1}$}
\put(1.25,10){\line(1,0){7.5}}
\put(0,11.25){\line(0,1){7.5}}
\put(10,10){\circle*{2.5}}
\put(10.5,6){$\cl_{1,1}$}
\put(11.25,10){\line(1,0){7.5}}
\put(10,11.25){\line(0,1){7.5}}
\put(20,10){\circle{2.5}}
\put(20.5,6){$\cl_{2,1}$}
\put(21.25,10){\line(1,0){7.5}}
\put(20,11.25){\line(0,1){7.5}}
\put(30,10){\circle*{2.5}}
\put(30.5,6){$\cl_{3,1}$}
\put(31.25,10){\line(1,0){7.5}}
\put(30,11.25){\line(0,1){7.5}}
\put(40,10){\circle{2.5}}
\put(40.5,6){$\cl_{4,1}$}
\put(41.25,10){\line(1,0){7.5}}
\put(40,11.25){\line(0,1){7.5}}
\put(50,10){\circle*{2.5}}
\put(50.5,6){$\cl_{5,1}$}
\put(51.25,10){\line(1,0){7.5}}
\put(50,11.25){\line(0,1){7.5}}
\put(60,10){\circle{2.5}}
\put(60.5,6){$\cl_{6,1}$}
\put(61.25,10){\line(1,0){7.5}}
\put(60,11.25){\line(0,1){7.5}}
\put(70,10){\circle*{2.5}}
\put(70.5,6){$\cl_{7,1}$}
\put(70,11.25){\line(0,1){7.5}}
\put(0,20){\circle*{2.5}}
\put(0.5,16){$\cl_{0,2}$}
\put(1.25,20){\line(1,0){7.5}}
\put(0,21.25){\line(0,1){7.5}}
\put(10,20){\circle{2.5}}
\put(10.5,16){$\cl_{1,2}$}
\put(11.25,20){\line(1,0){7.5}}
\put(10,21.25){\line(0,1){7.5}}
\put(20,20){\circle*{2.5}}
\put(20.5,16){$\cl_{2,2}$}
\put(21.25,20){\line(1,0){7.5}}
\put(20,21.25){\line(0,1){7.5}}
\put(30,20){\circle{2.5}}
\put(30.5,16){$\cl_{3,2}$}
\put(31.25,20){\line(1,0){7.5}}
\put(30,21.25){\line(0,1){7.5}}
\put(40,20){\circle*{2.5}}
\put(40.5,16){$\cl_{4,2}$}
\put(41.25,20){\line(1,0){7.5}}
\put(40,21.25){\line(0,1){7.5}}
\put(50,20){\circle{2.5}}
\put(50.5,16){$\cl_{5,2}$}
\put(51.25,20){\line(1,0){7.5}}
\put(50,21.25){\line(0,1){7.5}}
\put(60,20){\circle*{2.5}}
\put(60.5,16){$\cl_{6,2}$}
\put(61.25,20){\line(1,0){7.5}}
\put(60,21.25){\line(0,1){7.5}}
\put(70,20){\circle{2.5}}
\put(70.5,16){$\cl_{7,2}$}
\put(70,21.25){\line(0,1){7.5}}
\put(0,30){\circle{2.5}}
\put(0.5,26){$\cl_{0,3}$}
\put(1.25,30){\line(1,0){7.5}}
\put(0,31.25){\line(0,1){7.5}}
\put(10,30){\circle*{2.5}}
\put(10.5,26){$\cl_{1,3}$}
\put(11.25,30){\line(1,0){7.5}}
\put(10,31.25){\line(0,1){7.5}}
\put(20,30){\circle{2.5}}
\put(20.5,26){$\cl_{2,3}$}
\put(21.25,30){\line(1,0){7.5}}
\put(20,31.25){\line(0,1){7.5}}
\put(30,30){\circle*{2.5}}
\put(30.5,26){$\cl_{3,3}$}
\put(31.25,30){\line(1,0){7.5}}
\put(30,31.25){\line(0,1){7.5}}
\put(40,30){\circle{2.5}}
\put(40.5,26){$\cl_{4,3}$}
\put(41.25,30){\line(1,0){7.5}}
\put(40,31.25){\line(0,1){7.5}}
\put(50,30){\circle*{2.5}}
\put(50.5,26){$\cl_{5,3}$}
\put(51.25,30){\line(1,0){7.5}}
\put(50,31.25){\line(0,1){7.5}}
\put(60,30){\circle{2.5}}
\put(60.5,26){$\cl_{6,3}$}
\put(61.25,30){\line(1,0){7.5}}
\put(60,31.25){\line(0,1){7.5}}
\put(70,30){\circle*{2.5}}
\put(70.5,26){$\cl_{7,3}$}
\put(70,31.25){\line(0,1){7.5}}
\put(0,40){\circle*{2.5}}
\put(0.5,36){$\cl_{0,4}$}
\put(1.25,40){\line(1,0){7.5}}
\put(0,41.25){\line(0,1){7.5}}
\put(10,40){\circle{2.5}}
\put(10.5,36){$\cl_{1,4}$}
\put(11.25,40){\line(1,0){7.5}}
\put(10,41.25){\line(0,1){7.5}}
\put(20,40){\circle*{2.5}}
\put(20.5,36){$\cl_{2,4}$}
\put(21.25,40){\line(1,0){7.5}}
\put(20,41.25){\line(0,1){7.5}}
\put(30,40){\circle{2.5}}
\put(30.5,36){$\cl_{3,4}$}
\put(31.25,40){\line(1,0){7.5}}
\put(30,41.25){\line(0,1){7.5}}
\put(40,40){\circle*{2.5}}
\put(40.5,36){$\cl_{4,4}$}
\put(41.25,40){\line(1,0){7.5}}
\put(40,41.25){\line(0,1){7.5}}
\put(50,40){\circle{2.5}}
\put(50.5,36){$\cl_{5,4}$}
\put(51.25,40){\line(1,0){7.5}}
\put(50,41.25){\line(0,1){7.5}}
\put(60,40){\circle*{2.5}}
\put(60.5,36){$\cl_{6,4}$}
\put(61.25,40){\line(1,0){7.5}}
\put(60,41.25){\line(0,1){7.5}}
\put(70,40){\circle{2.5}}
\put(70.5,36){$\cl_{7,4}$}
\put(70,41.25){\line(0,1){7.5}}
\put(0,50){\circle{2.5}}
\put(0.5,46){$\cl_{0,5}$}
\put(1.25,50){\line(1,0){7.5}}
\put(0,51.25){\line(0,1){7.5}}
\put(10,50){\circle*{2.5}}
\put(10.5,46){$\cl_{1,5}$}
\put(11.25,50){\line(1,0){7.5}}
\put(10,51.25){\line(0,1){7.5}}
\put(20,50){\circle{2.5}}
\put(20.5,46){$\cl_{2,5}$}
\put(21.25,50){\line(1,0){7.5}}
\put(20,51.25){\line(0,1){7.5}}
\put(30,50){\circle*{2.5}}
\put(30.5,46){$\cl_{3,5}$}
\put(31.25,50){\line(1,0){7.5}}
\put(30,51.25){\line(0,1){7.5}}
\put(40,50){\circle{2.5}}
\put(40.5,46){$\cl_{4,5}$}
\put(41.25,50){\line(1,0){7.5}}
\put(40,51.25){\line(0,1){7.5}}
\put(50,50){\circle*{2.5}}
\put(50.5,46){$\cl_{5,5}$}
\put(51.25,50){\line(1,0){7.5}}
\put(50,51.25){\line(0,1){7.5}}
\put(60,50){\circle{2.5}}
\put(60.5,46){$\cl_{6,5}$}
\put(61.25,50){\line(1,0){7.5}}
\put(60,51.25){\line(0,1){7.5}}
\put(70,50){\circle*{2.5}}
\put(70.5,46){$\cl_{7,5}$}
\put(70,51.25){\line(0,1){7.5}}
\put(0,60){\circle*{2.5}}
\put(0.5,56){$\cl_{0,6}$}
\put(1.25,60){\line(1,0){7.5}}
\put(0,61.25){\line(0,1){7.5}}
\put(10,60){\circle{2.5}}
\put(10.5,56){$\cl_{1,6}$}
\put(11.25,60){\line(1,0){7.5}}
\put(10,61.25){\line(0,1){7.5}}
\put(20,60){\circle*{2.5}}
\put(20.5,56){$\cl_{2,6}$}
\put(21.25,60){\line(1,0){7.5}}
\put(20,61.25){\line(0,1){7.5}}
\put(30,60){\circle{2.5}}
\put(30.5,56){$\cl_{3,6}$}
\put(31.25,60){\line(1,0){7.5}}
\put(30,61.25){\line(0,1){7.5}}
\put(40,60){\circle*{2.5}}
\put(40.5,56){$\cl_{4,6}$}
\put(41.25,60){\line(1,0){7.5}}
\put(40,61.25){\line(0,1){7.5}}
\put(50,60){\circle{2.5}}
\put(50.5,56){$\cl_{5,6}$}
\put(51.25,60){\line(1,0){7.5}}
\put(50,61.25){\line(0,1){7.5}}
\put(60,60){\circle*{2.5}}
\put(60.5,56){$\cl_{6,6}$}
\put(61.25,60){\line(1,0){7.5}}
\put(60,61.25){\line(0,1){7.5}}
\put(70,60){\circle{2.5}}
\put(70.5,56){$\cl_{7,6}$}
\put(70,61.25){\line(0,1){7.5}}
\put(0,70){\circle{2.5}}
\put(0.5,66){$\cl_{0,7}$}
\put(1.25,70){\line(1,0){7.5}}
\put(10,70){\circle*{2.5}}
\put(10.5,66){$\cl_{1,7}$}
\put(11.25,70){\line(1,0){7.5}}
\put(20,70){\circle{2.5}}
\put(20.5,66){$\cl_{2,7}$}
\put(21.25,70){\line(1,0){7.5}}
\put(30,70){\circle*{2.5}}
\put(30.5,66){$\cl_{3,7}$}
\put(31.25,70){\line(1,0){7.5}}
\put(40,70){\circle{2.5}}
\put(40.5,66){$\cl_{4,7}$}
\put(41.25,70){\line(1,0){7.5}}
\put(50,70){\circle*{2.5}}
\put(50.5,66){$\cl_{5,7}$}
\put(51.25,70){\line(1,0){7.5}}
\put(60,70){\circle{2.5}}
\put(60.5,66){$\cl_{6,7}$}
\put(61.25,70){\line(1,0){7.5}}
\put(70,70){\circle*{2.5}}
\put(70.5,66){$\cl_{7,7}$}
\put(80,80){\circle*{2.5}}
\put(80.5,76){$\cl_{8,8}$}
\put(15,23){$\bullet$}
\put(16,24){\vector(1,0){38}}
\put(54,24){\vector(1,0){36.5}}
\put(90,23){$\bullet$}
\put(91,24){\vector(0,1){30}}
\put(91,53){\vector(0,1){3}}
\put(91,56){\vector(0,1){42}}
\put(90,98){$\bullet$}
\end{picture}
\end{center}
\vspace{0.3cm}
\begin{center}\begin{minipage}{32pc}{\small {\bf Fig.\,1:} \textbf{The Spinorial Chessboard}. Even- and odd-dimensional Clifford algebras $\cl_{p,q}$, $0\leq p,q\leq 7$, occupy, respectively, black and white circles (squares of the board). Every real Clifford algebra can be reached from one on the board with rook's moves to the right and upward.}\end{minipage}\end{center}
\end{figure}

The algebra $\cl$ is naturally $\dZ_2$-graded. Let
$\cl^+$ (correspondingly $\cl^-$) be a set consisting of all even
(correspondingly odd) elements of the algebra $\cl$. The set $\cl^+$
is a subalgebra of $\cl$. It is obvious that $\cl=\cl^+\oplus\cl^-$,
and also $\cl^+\cl^+ \subset\cl^+,\,\cl^+\cl^-\subset\cl^-,\,
\cl^-\cl^+\subset\cl^-,\,\cl^-\cl^-\subset \cl^+$. A degree $\deg a$
of the even (correspondingly odd) element $a\in\cl$ is equal to 0
(correspondingly 1). Let $\mathfrak{A}$ and $\mathfrak{B}$ be the
two associative $\dZ_2$-graded algebras over the field $\F$; then a
multiplication of homogeneous elements
$\mathfrak{a}^\prime\in\mathfrak{A}$ and
$\mathfrak{b}\in\mathfrak{B}$ in a graded tensor product
$\mathfrak{A}\hat{\otimes}\mathfrak{B}$ is defined as follows:
$(\mathfrak{a}\otimes \mathfrak{b})(\mathfrak{a}^\prime \otimes
\mathfrak{b}^\prime)=(-1)^{\deg\mathfrak{b}\deg\mathfrak{a}^\prime}
\mathfrak{a}\mathfrak{a}^\prime\otimes\mathfrak{b}\mathfrak{b}^\prime$.
\begin{thm}[{\rm Chevalley \cite{Che55}}]
Let $V$ and $V^\prime$ are vector spaces over the field $\F$ and let
$Q$ and $Q^\prime$ are quadratic forms for $V$ and $V^\prime$. Then
a Clifford algebra $\cl(V\oplus V^\prime,Q\oplus Q^\prime)$ is
naturally isomorphic to
$\cl(V,Q)\hat{\otimes}\cl(V^\prime,Q^\prime)$.
\end{thm}
Let $\cl(V,Q)$ be the Clifford algebra over the field $\F=\R$, where
$V$ is a vector space endowed with quadratic form
$Q=x^2_1+\ldots+x^2_p-\ldots-x^2_{p+q}$. If $p+q$ is even and
$\omega^2=1$, then $\cl(V,Q)$ is called {\it positive} and
correspondingly {\it negative} if $\omega^2=-1$, that is,
$\cl_{p,q}>0$ if $p-q\equiv 0,4\pmod{8}$ and $\cl_{p,q}<0$ if
$p-q\equiv 2,6\pmod{8}$.
\begin{thm}[{\rm Karoubi \cite[Prop.~3.16]{Kar79}}]
1) If $\cl(V,Q)>0$ and $\dim V$ is even, then
\[
\cl(V\oplus V^{\p},Q\oplus
Q^{\p})\simeq\cl(V,Q)\otimes\cl(V^{\p},Q^{\p}).
\]
2) If $\cl(V,Q)<0$ and $\dim V$ is even, then
\[
\cl(V\oplus V^{\p},Q\oplus
Q^{\p})\simeq\cl(V,Q)\otimes\cl(V^{\p},-Q^{\p}).
\]
\end{thm}
\begin{sloppypar}\noindent
Over the field $\F=\C$ the Clifford algebra is always positive (if
$\omega^2=\e^2_{12\ldots p+q}=-1$, then we can suppose $\omega=
i\e_{12\ldots p+q}$). Thus, using Karoubi Theorem, we find that
\end{sloppypar}
\begin{equation}\label{TenProd}
\underbrace{\C_2\otimes\cdots\otimes\C_2}_{m\,\text{times}}\simeq\C_{2m}.
\end{equation}
Therefore, the tensor product in (\ref{TenProd}) is isomorphic to
$\C_{2m}$. For example, there are two different factorizations
$\cl_{1,1}\otimes\cl_{0,2}$ and $\cl_{1,1}\otimes\cl_{2,0}$ for the
spacetime algebra $\cl_{1,3}$ and Majorana algebra $\cl_{3,1}$.

The real Clifford algebra $\cl_{p,q}$ is central
simple if $p-q\not\equiv 1,5\pmod{8}$. The graded tensor product of
the two graded central simple algebras\index{algebra!graded central
simple} is also graded central simple \cite[Theorem 2]{Wal64}. It is
known that for the Clifford algebra with odd dimensionality, the
isomorphisms are as follows: $\cl^+_{p,q+1}\simeq\cl_{p,q}$ and
$\cl^+_{p+1,q}\simeq\cl_{q,p}$ \cite{Rash,Port}. Thus,
$\cl^+_{p,q+1}$ and $\cl^+_{p+1,q}$ are central simple algebras.
Further, in accordance with Chevalley Theorem for the graded tensor
product there is an isomorphism
$\cl_{p,q}\hat{\otimes}\cl_{p^{\p},q^{\p}}\simeq
\cl_{p+p^{\p},q+q^{\p}}$. Two algebras $\cl_{p,q}$ and
$\cl_{p^{\p},q^{\p}}$ are said to be of the same class if
$p+q^{\p}\equiv p^{\p}+q\pmod{8}$. The graded central simple
Clifford algebras over the field $\F=\R$ form eight similarity
classes, which, as it is easy to see, coincide with the eight types
of the algebras $\cl_{p,q}$. The set of these 8 types (classes)
forms a Brauer-Wall group $BW_{\R}$ \cite{Wal64,Lou81}. It is
obvious that an action of
$BW_{\R}$ has a cyclic structure, which
is formally equivalent to the action of cyclic
group $\dZ_8$. The cyclic structure of $BW_{\R}$ may be represented on the Budinich-Trautman diagram
(spinorial clock) \cite{BT88} (Fig.\,2) by means of a
transition $\cl^+_{p,q}\stackrel{h}{\longrightarrow}\cl_{p,q}$ (the
round on the diagram is realized by an hour-hand). At this point,
the type of the algebra is defined on the diagram by an equality
$q-p=h+8r$, where $h\in\{1,\ldots,8\}$, $r\in\dZ$. It is obvious that a group structure over $\cl_{p,q}$, defined by
$BW_{\R}$, is related with the Atiyah-Bott-Shapiro periodicity
\cite{AtBSh}.

\begin{figure}[ht]
\[
\unitlength=0.5mm
\begin{picture}(100.00,110.00)

\put(97,67){$\C$}\put(105,64){$p-q\equiv 7\!\!\!\!\pmod{8}$}
\put(80,80){1} \put(75,93.3){$\cdot$} \put(75.5,93){$\cdot$}
\put(76,92.7){$\cdot$} \put(76.5,92.4){$\cdot$}
\put(77,92.08){$\cdot$} \put(77.5,91.76){$\cdot$}
\put(78,91.42){$\cdot$} \put(78.5,91.08){$\cdot$}
\put(79,90.73){$\cdot$} \put(79.5,90.37){$\cdot$}
\put(80,90.0){$\cdot$} \put(80.5,89.62){$\cdot$}
\put(81,89.23){$\cdot$} \put(81.5,88.83){$\cdot$}
\put(82,88.42){$\cdot$} \put(82.5,87.99){$\cdot$}
\put(83,87.56){$\cdot$} \put(83.5,87.12){$\cdot$}
\put(84,86.66){$\cdot$} \put(84.5,86.19){$\cdot$}
\put(85,85.70){$\cdot$} \put(85.5,85.21){$\cdot$}
\put(86,84.69){$\cdot$} \put(86.5,84.17){$\cdot$}
\put(87,83.63){$\cdot$} \put(87.5,83.07){$\cdot$}
\put(88,82.49){$\cdot$} \put(88.5,81.9){$\cdot$}
\put(89,81.29){$\cdot$} \put(89.5,80.65){$\cdot$}
\put(90,80){$\cdot$} \put(90.5,79.32){$\cdot$}
\put(91,78.62){$\cdot$} \put(91.5,77.89){$\cdot$}
\put(92,77.13){$\cdot$} \put(92.5,76.34){$\cdot$}
\put(93,75.51){$\cdot$} \put(93.5,74.65){$\cdot$}
\put(94,73.74){$\cdot$} \put(94.5,72.79){$\cdot$}
\put(96.5,73.74){\vector(1,-2){1}}
\put(80,20){3} \put(97,31){$\BH$}\put(105,28){$p-q\equiv
6\!\!\!\!\pmod{8}$} \put(75,6.7){$\cdot$} \put(75.5,7){$\cdot$}
\put(76,7.29){$\cdot$} \put(76.5,7.6){$\cdot$}
\put(77,7.91){$\cdot$} \put(77.5,8.24){$\cdot$}
\put(78,8.57){$\cdot$} \put(78.5,8.91){$\cdot$}
\put(79,9.27){$\cdot$} \put(79.5,9.63){$\cdot$} \put(80,10){$\cdot$}
\put(80.5,10.38){$\cdot$} \put(81,10.77){$\cdot$}
\put(81.5,11.17){$\cdot$} \put(82,11.58){$\cdot$}
\put(82.5,12.00){$\cdot$} \put(83,12.44){$\cdot$}
\put(83.5,12.88){$\cdot$} \put(84,13.34){$\cdot$}
\put(84.5,13.8){$\cdot$} \put(85,14.29){$\cdot$}
\put(85.5,14.79){$\cdot$} \put(86,15.3){$\cdot$}
\put(86.5,15.82){$\cdot$} \put(87,16.37){$\cdot$}
\put(87.5,16.92){$\cdot$} \put(88,17.5){$\cdot$}
\put(88.5,18.09){$\cdot$} \put(89,18.71){$\cdot$}
\put(89.5,19.34){$\cdot$} \put(90,20){$\cdot$}
\put(90.5,20.68){$\cdot$} \put(91,21.38){$\cdot$}
\put(91.5,22.11){$\cdot$} \put(92,22.87){$\cdot$}
\put(92.5,23.66){$\cdot$} \put(93,24.48){$\cdot$}
\put(93.5,25.34){$\cdot$} \put(94,26.25){$\cdot$}
\put(94.5,27.20){$\cdot$} \put(95,28.20){$\cdot$}
\put(20,80){7} \put(25,93.3){$\cdot$} \put(24.5,93){$\cdot$}
\put(24,92.7){$\cdot$} \put(23.5,92.49){$\cdot$}
\put(23,92.08){$\cdot$} \put(22.5,91.75){$\cdot$}
\put(22,91.42){$\cdot$} \put(21.5,91.08){$\cdot$}
\put(21,90.73){$\cdot$} \put(20.5,90.37){$\cdot$}
\put(20,90){$\cdot$} \put(19.5,89.62){$\cdot$}
\put(19,89.23){$\cdot$} \put(18.5,88.83){$\cdot$}
\put(18,88.42){$\cdot$} \put(17.5,87.99){$\cdot$}
\put(17,87.56){$\cdot$} \put(16.5,87.12){$\cdot$}
\put(16,86.66){$\cdot$} \put(15.5,86.19){$\cdot$}
\put(15,85.70){$\cdot$} \put(14.5,85.21){$\cdot$}
\put(14,84.69){$\cdot$} \put(13.5,84.17){$\cdot$}
\put(13,83.63){$\cdot$} \put(12.5,83.07){$\cdot$}
\put(12,82.49){$\cdot$} \put(11.5,81.9){$\cdot$}
\put(11,81.29){$\cdot$} \put(10.5,80.65){$\cdot$}
\put(10,80){$\cdot$} \put(9.5,79.32){$\cdot$} \put(9,78.62){$\cdot$}
\put(8.5,77.89){$\cdot$} \put(8,77.13){$\cdot$}
\put(7.5,76.34){$\cdot$} \put(7,75.51){$\cdot$}
\put(6.5,74.65){$\cdot$} \put(6,73.79){$\cdot$}
\put(5.5,72.79){$\cdot$} \put(5,71.79){$\cdot$}
\put(20,20){5} \put(25,6.7){$\cdot$} \put(24.5,7){$\cdot$}
\put(24,7.29){$\cdot$} \put(23.5,7.6){$\cdot$}
\put(23,7.91){$\cdot$} \put(22.5,8.24){$\cdot$}
\put(22,8.57){$\cdot$} \put(21.5,8.91){$\cdot$}
\put(21,9.27){$\cdot$} \put(20.5,9.63){$\cdot$} \put(20,10){$\cdot$}
\put(19.5,10.38){$\cdot$} \put(19,10.77){$\cdot$}
\put(18.5,11.17){$\cdot$} \put(18,11.58){$\cdot$}
\put(17.5,12){$\cdot$} \put(17,12.44){$\cdot$}
\put(16.5,12.88){$\cdot$} \put(16,13.34){$\cdot$}
\put(15.5,13.8){$\cdot$} \put(15,14.29){$\cdot$}
\put(14.5,14.79){$\cdot$} \put(14,15.3){$\cdot$}
\put(13.5,15.82){$\cdot$} \put(13,16.37){$\cdot$}
\put(12.5,16.92){$\cdot$} \put(12,17.5){$\cdot$}
\put(11.5,18.09){$\cdot$} \put(11,18.71){$\cdot$}
\put(10.5,19.34){$\cdot$} \put(10,20){$\cdot$}
\put(9.5,20.68){$\cdot$} \put(9,21.38){$\cdot$}
\put(8.5,22.11){$\cdot$} \put(8,22.87){$\cdot$}
\put(7.5,23.66){$\cdot$} \put(7,24.48){$\cdot$}
\put(6.5,25.34){$\cdot$} \put(6,26.25){$\cdot$}
\put(5.5,27.20){$\cdot$} \put(5,28.20){$\cdot$}
\put(13,97){$\R\oplus\R$}\put(-55,105){$p-q\equiv
1\!\!\!\!\pmod{8}$} \put(50,93){8} \put(50,100){$\cdot$}
\put(49.5,99.99){$\cdot$} \put(49,99.98){$\cdot$}
\put(48.5,99.97){$\cdot$} \put(48,99.96){$\cdot$}
\put(47.5,99.94){$\cdot$} \put(47,99.91){$\cdot$}
\put(46.5,99.86){$\cdot$} \put(46,99.84){$\cdot$}
\put(45.5,99.8){$\cdot$} \put(45,99.75){$\cdot$}
\put(44.5,99.7){$\cdot$} \put(44,99.64){$\cdot$}
\put(43.5,99.57){$\cdot$} \put(43,99.51){$\cdot$}
\put(42.5,99.43){$\cdot$} \put(42,99.35){$\cdot$}
\put(41.5,99.27){$\cdot$} \put(41,99.18){$\cdot$}
\put(40.5,99.09){$\cdot$} \put(40,98.99){$\cdot$}
\put(39.5,98.88){$\cdot$} \put(39,98.77){$\cdot$}
\put(38.5,98.66){$\cdot$} \put(38,98.54){$\cdot$}
\put(37.5,98.41){$\cdot$} \put(37,98.28){$\cdot$}
\put(50.5,99.99){$\cdot$} \put(51,99.98){$\cdot$}
\put(51.5,99.97){$\cdot$} \put(52,99.96){$\cdot$}
\put(52.5,99.94){$\cdot$} \put(53,99.91){$\cdot$}
\put(53.5,99.86){$\cdot$} \put(54,99.84){$\cdot$}
\put(54.5,99.8){$\cdot$} \put(55,99.75){$\cdot$}
\put(55.5,99.7){$\cdot$} \put(56,99.64){$\cdot$}
\put(56.5,99.57){$\cdot$} \put(57,99.51){$\cdot$}
\put(57.5,99.43){$\cdot$} \put(58,99.35){$\cdot$}
\put(58.5,99.27){$\cdot$} \put(59,99.18){$\cdot$}
\put(59.5,99.09){$\cdot$} \put(60,98.99){$\cdot$}
\put(60.5,98.88){$\cdot$} \put(61,98.77){$\cdot$}
\put(61.5,98.66){$\cdot$} \put(62,98.54){$\cdot$}
\put(62.5,98.41){$\cdot$} \put(63,98.28){$\cdot$}
\put(68,97){$\R$}\put(73,105){$p-q\equiv 0\!\!\!\!\pmod{8}$}
\put(50,7){4} \put(68,2){$\BH\oplus\BH$}\put(90,-4){$p-q\equiv
5\!\!\!\!\pmod{8}$} \put(50,0){$\cdot$} \put(50.5,0){$\cdot$}
\put(51,0.01){$\cdot$} \put(51.5,0.02){$\cdot$}
\put(52,0.04){$\cdot$} \put(52.5,0.06){$\cdot$}
\put(53,0.09){$\cdot$} \put(53.5,0.12){$\cdot$}
\put(54,0.16){$\cdot$} \put(54.5,0.2){$\cdot$}
\put(55,0.25){$\cdot$} \put(55.5,0.3){$\cdot$}
\put(56,0.36){$\cdot$} \put(56.5,0.42){$\cdot$}
\put(57,0.49){$\cdot$} \put(57.5,0.56){$\cdot$}
\put(58,0.64){$\cdot$} \put(58.5,0.73){$\cdot$}
\put(59,0.82){$\cdot$} \put(59.5,0.91){$\cdot$}
\put(60,1.01){$\cdot$} \put(60.5,1.11){$\cdot$}
\put(61,1.22){$\cdot$} \put(61.5,1.34){$\cdot$}
\put(62,1.46){$\cdot$} \put(62.5,1.59){$\cdot$}
\put(63,1.72){$\cdot$} \put(49.5,0){$\cdot$} \put(49,0.01){$\cdot$}
\put(48.5,0.02){$\cdot$} \put(48,0.04){$\cdot$}
\put(47.5,0.06){$\cdot$} \put(47,0.09){$\cdot$}
\put(46.5,0.12){$\cdot$} \put(46,0.16){$\cdot$}
\put(45.5,0.2){$\cdot$} \put(45,0.25){$\cdot$}
\put(44.5,0.3){$\cdot$} \put(44,0.36){$\cdot$}
\put(43.5,0.42){$\cdot$} \put(43,0.49){$\cdot$}
\put(42.5,0.56){$\cdot$} \put(42,0.64){$\cdot$}
\put(41.5,0.73){$\cdot$} \put(41,0.82){$\cdot$}
\put(40.5,0.91){$\cdot$} \put(40,1.01){$\cdot$}
\put(39.5,1.11){$\cdot$} \put(39,1.22){$\cdot$}
\put(38.5,1.34){$\cdot$} \put(38,1.46){$\cdot$}
\put(37.5,1.59){$\cdot$} \put(37,1.72){$\cdot$}
\put(28,3){$\BH$}\put(-40,-4){$p-q\equiv 4\!\!\!\!\pmod{8}$}
\put(93,50){2} \put(98.28,63){$\cdot$} \put(98.41,62.5){$\cdot$}
\put(98.54,62){$\cdot$} \put(98.66,61.5){$\cdot$}
\put(98.77,61){$\cdot$} \put(98.88,60.5){$\cdot$}
\put(98.99,60){$\cdot$} \put(99.09,59.5){$\cdot$}
\put(99.18,59){$\cdot$} \put(99.27,58.5){$\cdot$}
\put(99.35,58){$\cdot$} \put(99.43,57.5){$\cdot$}
\put(99.51,57){$\cdot$} \put(99.57,56.5){$\cdot$}
\put(99.64,56){$\cdot$} \put(99.7,55.5){$\cdot$}
\put(99.75,55){$\cdot$} \put(99.8,54.5){$\cdot$}
\put(99.84,54){$\cdot$} \put(99.86,53.5){$\cdot$}
\put(99.91,53){$\cdot$} \put(99.94,52.5){$\cdot$}
\put(99.96,52){$\cdot$} \put(99.97,51.5){$\cdot$}
\put(99.98,51){$\cdot$} \put(99.99,50.5){$\cdot$}
\put(100,50){$\cdot$} \put(98.28,37){$\cdot$}
\put(98.41,37.5){$\cdot$} \put(98.54,38){$\cdot$}
\put(98.66,38.5){$\cdot$} \put(98.77,39){$\cdot$}
\put(98.88,39.5){$\cdot$} \put(98.99,40){$\cdot$}
\put(99.09,40.5){$\cdot$} \put(99.18,41){$\cdot$}
\put(99.27,41.5){$\cdot$} \put(99.35,42){$\cdot$}
\put(99.43,42.5){$\cdot$} \put(99.51,43){$\cdot$}
\put(99.57,43.5){$\cdot$} \put(99.64,44){$\cdot$}
\put(99.7,44.5){$\cdot$} \put(99.75,45){$\cdot$}
\put(99.8,45.5){$\cdot$} \put(99.84,46){$\cdot$}
\put(99.86,46.5){$\cdot$} \put(99.91,47){$\cdot$}
\put(99.94,47.5){$\cdot$} \put(99.96,48){$\cdot$}
\put(99.97,48.5){$\cdot$} \put(99.98,49){$\cdot$}
\put(99.99,49.5){$\cdot$}
\put(7,50){6} \put(1,32){$\C$}\put(-67,29){$p-q\equiv
3\!\!\!\!\pmod{8}$} \put(1.72,63){$\cdot$} \put(1.59,62.5){$\cdot$}
\put(1.46,62){$\cdot$} \put(1.34,61.5){$\cdot$}
\put(1.22,61){$\cdot$} \put(1.11,60.5){$\cdot$}
\put(1.01,60){$\cdot$} \put(0.99,59.5){$\cdot$}
\put(0.82,59){$\cdot$} \put(0.73,58.5){$\cdot$}
\put(0.64,58){$\cdot$} \put(0.56,57.5){$\cdot$}
\put(0.49,57){$\cdot$} \put(0.42,56.5){$\cdot$}
\put(0.36,56){$\cdot$} \put(0.3,55.5){$\cdot$}
\put(0.25,55){$\cdot$} \put(0.2,54.5){$\cdot$}
\put(0.16,54){$\cdot$} \put(0.12,53.5){$\cdot$}
\put(0.09,53){$\cdot$} \put(0.06,52.5){$\cdot$}
\put(0.04,52){$\cdot$} \put(0.02,51.5){$\cdot$}
\put(0.01,51){$\cdot$} \put(0,50.5){$\cdot$} \put(0,50){$\cdot$}
\put(1.72,37){$\cdot$} \put(1.59,37.5){$\cdot$}
\put(1.46,38){$\cdot$} \put(1.34,38.5){$\cdot$}
\put(1.22,39){$\cdot$} \put(1.11,39.5){$\cdot$}
\put(1.01,40){$\cdot$} \put(0.99,40.5){$\cdot$}
\put(0.82,41){$\cdot$} \put(0.73,41.5){$\cdot$}
\put(0.64,42){$\cdot$} \put(0.56,42.5){$\cdot$}
\put(0.49,43){$\cdot$} \put(0.42,43.5){$\cdot$}
\put(0.36,44){$\cdot$} \put(0.3,44.5){$\cdot$}
\put(0.25,45){$\cdot$} \put(0.2,45.5){$\cdot$}
\put(0.16,46){$\cdot$} \put(0.12,46.5){$\cdot$}
\put(0.09,47){$\cdot$} \put(0.06,47.5){$\cdot$}
\put(0.04,48){$\cdot$} \put(0.02,48.5){$\cdot$}
\put(0.01,49){$\cdot$} \put(0,49.5){$\cdot$}
\put(0.5,67){$\R$}\put(-67,75){$p-q\equiv 2\!\!\!\!\pmod{8}$}
\end{picture}
\]
\vspace{2ex}
\begin{center}
\begin{minipage}{25pc}{\small
{\bf Fig.\,2:} \textbf{The Spinorial Clock}.
The Budinich-Trautman diagram for the Brauer-Wall group
$BW_{\R}\simeq\dZ_8$.}
\end{minipage}
\end{center}
\end{figure}
\medskip

On the other hand, graded central simple Clifford algebras over the field
$\F=\R$ form a \textbf{\emph{graded Brauer group}}
$G(\cl_{p,q},\gamma,\bcirc)$ \cite{Wal64,Var12}, a cyclic
structure of which is described by the Brauer-Wall
group $BW_{\R}\simeq\dZ_8$ \cite{Lou81}. Therefore, a cyclic structure of $G(\cl_{p,q},\gamma,\bcirc)\sim BW_{\R}$ is
defined by a transition
$\cl^+_{p,q}\overset{h}{\longrightarrow}\cl_{p,q}$, where the type
of $\cl_{p,q}$ is defined by the formula $q-p=h+8r$, here
$h\in\{1,\ldots,8\}$, $r\in\dZ$ \cite{BT88}. Let us consider in detail several action cycles of $BW_{\R}\simeq\dZ_8$. In virtue of an isomorphism $\cl^+_{0,1}\simeq\cl_{0,0}$ a transition $\cl^+_{0,1}\overset{1}{\longrightarrow}\cl_{0,1}$ leads to a transition $\cl_{0,0}\overset{1}{\longrightarrow}\cl_{0,1}$, that is, $\R\overset{1}{\longrightarrow}\C$. At this point, $h=1$ and $r=0$ (an original point of the first cycle). Further, in virtue of $\cl^+_{0,2}\simeq\cl_{0,1}$ a transition $\cl^+_{0,2}\overset{2}{\longrightarrow}\cl_{0,2}$ induces a transition $\cl_{0,1}\overset{2}{\longrightarrow}\cl_{0,2}$ ($\C\overset{2}{\longrightarrow}\BH$), at this point, $h=2$ and $r=0$. The following transition $\cl^+_{0,3}\overset{3}{\longrightarrow}\cl_{0,3}$ ($\BH\overset{3}{\longrightarrow}\BH\oplus\BH$) in virtue of $\cl^+_{0,3}\simeq\cl_{0,2}$ leads to a transition $\cl_{0,2}\overset{3}{\longrightarrow}\cl_{0,3}$. At this transition we have $h=3$ and $r=0$. In virtue of an isomorphism $\cl^+_{0,4}\simeq\cl_{0,3}$ a transition $\cl^+_{0,4}\overset{4}{\longrightarrow}\cl_{0,4}$ ($\BH\oplus\BH\overset{4}{\longrightarrow}\BH$) induces $\cl_{0,3}\overset{4}{\longrightarrow}\cl_{0,4}$, at this point, $h=4$ and $r=0$. Further, in virtue of $\cl^+_{0,5}\simeq\cl_{0,4}$ a transition $\cl^+_{0,5}\overset{5}{\longrightarrow}\cl_{0,5}$ ($\BH\overset{5}{\longrightarrow}\C$) induces $\cl_{0,4}\overset{5}{\longrightarrow}\cl_{0,5}$. At this transition we have $h=5$ and $r=0$. The following transition $\cl^+_{0,6}\overset{6}{\longrightarrow}\cl_{0,6}$ ($\C\overset{6}{\longrightarrow}\R$) in virtue of $\cl^+_{0,6}\simeq\cl_{0,5}$ induces $\cl_{0,5}\overset{6}{\longrightarrow}\cl_{0,6}$, here $h=6$ and $r=0$. In its turn, the transition $\cl^+_{0,7}\overset{7}{\longrightarrow}\cl_{0,7}$ ($\R\overset{7}{\longrightarrow}\R\oplus\R$) in virtue of $\cl^+_{0,7}\simeq\cl_{0,6}$ induces $\cl_{0,6}\overset{7}{\longrightarrow}\cl_{0,7}$. At this transition we have $h=7$ and $r=0$. Finally, a transition $\cl^+_{0,8}\overset{8}{\longrightarrow}\cl_{0,8}$ ($\R\oplus\R\overset{8}{\longrightarrow}\R$) finishes the first cycle ($h=8$, $r=0$) and in virtue of $\cl^+_{0,8}\simeq\cl_{0,7}$ induces the following transition $\cl_{0,7}\overset{8}{\longrightarrow}\cl_{0,8}$. The full round of the first cycle is shown on the Fig.\,3. The first cycle generates the first eight squares ($\cl_{0,q}$, $q=0,\ldots,7$) of the spinorial chessboard (see Fig.\,1). The following eight squares ($\cl_{1,q}$, $q=0,\ldots,7$) are generated by the first cycle also (via the rule $\cl_{1,q}\simeq\cl_{1,0}\otimes\cl_{0,q}$, $q=0,\ldots,7$) and so on ($\cl_{2,q}\simeq\cl_{2,0}\otimes\cl_{0,q}$, $\cl_{3,q}\simeq\cl_{3,0}\otimes\cl_{0,q}$, $\ldots$, $\cl_{6,q}\simeq\cl_{6,0}\otimes\cl_{0,q}$, $q=0,\ldots,7$). In like manner all the squares of the spinorial chessboard, shown on the Fig.\,1, are filled.
\begin{figure}[ht]
\[
\unitlength=0.4mm
\begin{picture}(100.00,110.00)
\put(97,67){$\cl_{0,1}$}\put(108,78){$\C$}
\put(80,80){1} \put(75,93.3){$\cdot$} \put(75.5,93){$\cdot$}
\put(76,92.7){$\cdot$} \put(76.5,92.4){$\cdot$}
\put(77,92.08){$\cdot$} \put(77.5,91.76){$\cdot$}
\put(78,91.42){$\cdot$} \put(78.5,91.08){$\cdot$}
\put(79,90.73){$\cdot$} \put(79.5,90.37){$\cdot$}
\put(80,90.0){$\cdot$} \put(80.5,89.62){$\cdot$}
\put(81,89.23){$\cdot$} \put(81.5,88.83){$\cdot$}
\put(82,88.42){$\cdot$} \put(82.5,87.99){$\cdot$}
\put(83,87.56){$\cdot$} \put(83.5,87.12){$\cdot$}
\put(84,86.66){$\cdot$} \put(84.5,86.19){$\cdot$}
\put(85,85.70){$\cdot$} \put(85.5,85.21){$\cdot$}
\put(86,84.69){$\cdot$} \put(86.5,84.17){$\cdot$}
\put(87,83.63){$\cdot$} \put(87.5,83.07){$\cdot$}
\put(88,82.49){$\cdot$} \put(88.5,81.9){$\cdot$}
\put(89,81.29){$\cdot$} \put(89.5,80.65){$\cdot$}
\put(90,80){$\cdot$} \put(90.5,79.32){$\cdot$}
\put(91,78.62){$\cdot$} \put(91.5,77.89){$\cdot$}
\put(92,77.13){$\cdot$} \put(92.5,76.34){$\cdot$}
\put(93,75.51){$\cdot$} \put(93.5,74.65){$\cdot$}
\put(94,73.74){$\cdot$} \put(94.5,72.79){$\cdot$}
\put(96.5,73.74){\vector(1,-2){1}}
\put(80,20){3} \put(97,31){$\cl_{0,2}$}\put(108,42){$\BH$} \put(75,6.7){$\cdot$} \put(75.5,7){$\cdot$}
\put(76,7.29){$\cdot$} \put(76.5,7.6){$\cdot$}
\put(77,7.91){$\cdot$} \put(77.5,8.24){$\cdot$}
\put(78,8.57){$\cdot$} \put(78.5,8.91){$\cdot$}
\put(79,9.27){$\cdot$} \put(79.5,9.63){$\cdot$} \put(80,10){$\cdot$}
\put(80.5,10.38){$\cdot$} \put(81,10.77){$\cdot$}
\put(81.5,11.17){$\cdot$} \put(82,11.58){$\cdot$}
\put(82.5,12.00){$\cdot$} \put(83,12.44){$\cdot$}
\put(83.5,12.88){$\cdot$} \put(84,13.34){$\cdot$}
\put(84.5,13.8){$\cdot$} \put(85,14.29){$\cdot$}
\put(85.5,14.79){$\cdot$} \put(86,15.3){$\cdot$}
\put(86.5,15.82){$\cdot$} \put(87,16.37){$\cdot$}
\put(87.5,16.92){$\cdot$} \put(88,17.5){$\cdot$}
\put(88.5,18.09){$\cdot$} \put(89,18.71){$\cdot$}
\put(89.5,19.34){$\cdot$} \put(90,20){$\cdot$}
\put(90.5,20.68){$\cdot$} \put(91,21.38){$\cdot$}
\put(91.5,22.11){$\cdot$} \put(92,22.87){$\cdot$}
\put(92.5,23.66){$\cdot$} \put(93,24.48){$\cdot$}
\put(93.5,25.34){$\cdot$} \put(94,26.25){$\cdot$}
\put(94.5,27.20){$\cdot$} \put(95,28.20){$\cdot$}
\put(20,80){7} \put(25,93.3){$\cdot$} \put(24.5,93){$\cdot$}
\put(24,92.7){$\cdot$} \put(23.5,92.49){$\cdot$}
\put(23,92.08){$\cdot$} \put(22.5,91.75){$\cdot$}
\put(22,91.42){$\cdot$} \put(21.5,91.08){$\cdot$}
\put(21,90.73){$\cdot$} \put(20.5,90.37){$\cdot$}
\put(20,90){$\cdot$} \put(19.5,89.62){$\cdot$}
\put(19,89.23){$\cdot$} \put(18.5,88.83){$\cdot$}
\put(18,88.42){$\cdot$} \put(17.5,87.99){$\cdot$}
\put(17,87.56){$\cdot$} \put(16.5,87.12){$\cdot$}
\put(16,86.66){$\cdot$} \put(15.5,86.19){$\cdot$}
\put(15,85.70){$\cdot$} \put(14.5,85.21){$\cdot$}
\put(14,84.69){$\cdot$} \put(13.5,84.17){$\cdot$}
\put(13,83.63){$\cdot$} \put(12.5,83.07){$\cdot$}
\put(12,82.49){$\cdot$} \put(11.5,81.9){$\cdot$}
\put(11,81.29){$\cdot$} \put(10.5,80.65){$\cdot$}
\put(10,80){$\cdot$} \put(9.5,79.32){$\cdot$} \put(9,78.62){$\cdot$}
\put(8.5,77.89){$\cdot$} \put(8,77.13){$\cdot$}
\put(7.5,76.34){$\cdot$} \put(7,75.51){$\cdot$}
\put(6.5,74.65){$\cdot$} \put(6,73.79){$\cdot$}
\put(5.5,72.79){$\cdot$} \put(5,71.79){$\cdot$}
\put(20,20){5} \put(25,6.7){$\cdot$} \put(24.5,7){$\cdot$}
\put(24,7.29){$\cdot$} \put(23.5,7.6){$\cdot$}
\put(23,7.91){$\cdot$} \put(22.5,8.24){$\cdot$}
\put(22,8.57){$\cdot$} \put(21.5,8.91){$\cdot$}
\put(21,9.27){$\cdot$} \put(20.5,9.63){$\cdot$} \put(20,10){$\cdot$}
\put(19.5,10.38){$\cdot$} \put(19,10.77){$\cdot$}
\put(18.5,11.17){$\cdot$} \put(18,11.58){$\cdot$}
\put(17.5,12){$\cdot$} \put(17,12.44){$\cdot$}
\put(16.5,12.88){$\cdot$} \put(16,13.34){$\cdot$}
\put(15.5,13.8){$\cdot$} \put(15,14.29){$\cdot$}
\put(14.5,14.79){$\cdot$} \put(14,15.3){$\cdot$}
\put(13.5,15.82){$\cdot$} \put(13,16.37){$\cdot$}
\put(12.5,16.92){$\cdot$} \put(12,17.5){$\cdot$}
\put(11.5,18.09){$\cdot$} \put(11,18.71){$\cdot$}
\put(10.5,19.34){$\cdot$} \put(10,20){$\cdot$}
\put(9.5,20.68){$\cdot$} \put(9,21.38){$\cdot$}
\put(8.5,22.11){$\cdot$} \put(8,22.87){$\cdot$}
\put(7.5,23.66){$\cdot$} \put(7,24.48){$\cdot$}
\put(6.5,25.34){$\cdot$} \put(6,26.25){$\cdot$}
\put(5.5,27.20){$\cdot$} \put(5,28.20){$\cdot$}
\put(19,98){$\cl_{0,7}$}\put(-10,105){$\R\oplus\R$} \put(50,93){8} \put(50,100){$\cdot$}
\put(49.5,99.99){$\cdot$} \put(49,99.98){$\cdot$}
\put(48.5,99.97){$\cdot$} \put(48,99.96){$\cdot$}
\put(47.5,99.94){$\cdot$} \put(47,99.91){$\cdot$}
\put(46.5,99.86){$\cdot$} \put(46,99.84){$\cdot$}
\put(45.5,99.8){$\cdot$} \put(45,99.75){$\cdot$}
\put(44.5,99.7){$\cdot$} \put(44,99.64){$\cdot$}
\put(43.5,99.57){$\cdot$} \put(43,99.51){$\cdot$}
\put(42.5,99.43){$\cdot$} \put(42,99.35){$\cdot$}
\put(41.5,99.27){$\cdot$} \put(41,99.18){$\cdot$}
\put(40.5,99.09){$\cdot$} \put(40,98.99){$\cdot$}
\put(39.5,98.88){$\cdot$} \put(39,98.77){$\cdot$}
\put(38.5,98.66){$\cdot$} \put(38,98.54){$\cdot$}
\put(37.5,98.41){$\cdot$} \put(37,98.28){$\cdot$}
\put(50.5,99.99){$\cdot$} \put(51,99.98){$\cdot$}
\put(51.5,99.97){$\cdot$} \put(52,99.96){$\cdot$}
\put(52.5,99.94){$\cdot$} \put(53,99.91){$\cdot$}
\put(53.5,99.86){$\cdot$} \put(54,99.84){$\cdot$}
\put(54.5,99.8){$\cdot$} \put(55,99.75){$\cdot$}
\put(55.5,99.7){$\cdot$} \put(56,99.64){$\cdot$}
\put(56.5,99.57){$\cdot$} \put(57,99.51){$\cdot$}
\put(57.5,99.43){$\cdot$} \put(58,99.35){$\cdot$}
\put(58.5,99.27){$\cdot$} \put(59,99.18){$\cdot$}
\put(59.5,99.09){$\cdot$} \put(60,98.99){$\cdot$}
\put(60.5,98.88){$\cdot$} \put(61,98.77){$\cdot$}
\put(61.5,98.66){$\cdot$} \put(62,98.54){$\cdot$}
\put(62.5,98.41){$\cdot$} \put(63,98.28){$\cdot$}
\put(67,98){$\cl_{0,0}$}\put(83,105){$\R$}
\put(50,7){4} \put(67,1){$\cl_{0,3}$}\put(90,-4){$\BH\oplus\BH$} \put(50,0){$\cdot$} \put(50.5,0){$\cdot$}
\put(51,0.01){$\cdot$} \put(51.5,0.02){$\cdot$}
\put(52,0.04){$\cdot$} \put(52.5,0.06){$\cdot$}
\put(53,0.09){$\cdot$} \put(53.5,0.12){$\cdot$}
\put(54,0.16){$\cdot$} \put(54.5,0.2){$\cdot$}
\put(55,0.25){$\cdot$} \put(55.5,0.3){$\cdot$}
\put(56,0.36){$\cdot$} \put(56.5,0.42){$\cdot$}
\put(57,0.49){$\cdot$} \put(57.5,0.56){$\cdot$}
\put(58,0.64){$\cdot$} \put(58.5,0.73){$\cdot$}
\put(59,0.82){$\cdot$} \put(59.5,0.91){$\cdot$}
\put(60,1.01){$\cdot$} \put(60.5,1.11){$\cdot$}
\put(61,1.22){$\cdot$} \put(61.5,1.34){$\cdot$}
\put(62,1.46){$\cdot$} \put(62.5,1.59){$\cdot$}
\put(63,1.72){$\cdot$} \put(49.5,0){$\cdot$} \put(49,0.01){$\cdot$}
\put(48.5,0.02){$\cdot$} \put(48,0.04){$\cdot$}
\put(47.5,0.06){$\cdot$} \put(47,0.09){$\cdot$}
\put(46.5,0.12){$\cdot$} \put(46,0.16){$\cdot$}
\put(45.5,0.2){$\cdot$} \put(45,0.25){$\cdot$}
\put(44.5,0.3){$\cdot$} \put(44,0.36){$\cdot$}
\put(43.5,0.42){$\cdot$} \put(43,0.49){$\cdot$}
\put(42.5,0.56){$\cdot$} \put(42,0.64){$\cdot$}
\put(41.5,0.73){$\cdot$} \put(41,0.82){$\cdot$}
\put(40.5,0.91){$\cdot$} \put(40,1.01){$\cdot$}
\put(39.5,1.11){$\cdot$} \put(39,1.22){$\cdot$}
\put(38.5,1.34){$\cdot$} \put(38,1.46){$\cdot$}
\put(37.5,1.59){$\cdot$} \put(37,1.72){$\cdot$}
\put(21,1){$\cl_{0,4}$}\put(10,-4){$\BH$}
\put(93,50){2} \put(98.28,63){$\cdot$} \put(98.41,62.5){$\cdot$}
\put(98.54,62){$\cdot$} \put(98.66,61.5){$\cdot$}
\put(98.77,61){$\cdot$} \put(98.88,60.5){$\cdot$}
\put(98.99,60){$\cdot$} \put(99.09,59.5){$\cdot$}
\put(99.18,59){$\cdot$} \put(99.27,58.5){$\cdot$}
\put(99.35,58){$\cdot$} \put(99.43,57.5){$\cdot$}
\put(99.51,57){$\cdot$} \put(99.57,56.5){$\cdot$}
\put(99.64,56){$\cdot$} \put(99.7,55.5){$\cdot$}
\put(99.75,55){$\cdot$} \put(99.8,54.5){$\cdot$}
\put(99.84,54){$\cdot$} \put(99.86,53.5){$\cdot$}
\put(99.91,53){$\cdot$} \put(99.94,52.5){$\cdot$}
\put(99.96,52){$\cdot$} \put(99.97,51.5){$\cdot$}
\put(99.98,51){$\cdot$} \put(99.99,50.5){$\cdot$}
\put(100,50){$\cdot$} \put(98.28,37){$\cdot$}
\put(98.41,37.5){$\cdot$} \put(98.54,38){$\cdot$}
\put(98.66,38.5){$\cdot$} \put(98.77,39){$\cdot$}
\put(98.88,39.5){$\cdot$} \put(98.99,40){$\cdot$}
\put(99.09,40.5){$\cdot$} \put(99.18,41){$\cdot$}
\put(99.27,41.5){$\cdot$} \put(99.35,42){$\cdot$}
\put(99.43,42.5){$\cdot$} \put(99.51,43){$\cdot$}
\put(99.57,43.5){$\cdot$} \put(99.64,44){$\cdot$}
\put(99.7,44.5){$\cdot$} \put(99.75,45){$\cdot$}
\put(99.8,45.5){$\cdot$} \put(99.84,46){$\cdot$}
\put(99.86,46.5){$\cdot$} \put(99.91,47){$\cdot$}
\put(99.94,47.5){$\cdot$} \put(99.96,48){$\cdot$}
\put(99.97,48.5){$\cdot$} \put(99.98,49){$\cdot$}
\put(99.99,49.5){$\cdot$}
\put(7,50){6} \put(-5,32){$\cl_{0,5}$}\put(-19,29){$\C$} \put(1.72,63){$\cdot$} \put(1.59,62.5){$\cdot$}
\put(1.46,62){$\cdot$} \put(1.34,61.5){$\cdot$}
\put(1.22,61){$\cdot$} \put(1.11,60.5){$\cdot$}
\put(1.01,60){$\cdot$} \put(0.99,59.5){$\cdot$}
\put(0.82,59){$\cdot$} \put(0.73,58.5){$\cdot$}
\put(0.64,58){$\cdot$} \put(0.56,57.5){$\cdot$}
\put(0.49,57){$\cdot$} \put(0.42,56.5){$\cdot$}
\put(0.36,56){$\cdot$} \put(0.3,55.5){$\cdot$}
\put(0.25,55){$\cdot$} \put(0.2,54.5){$\cdot$}
\put(0.16,54){$\cdot$} \put(0.12,53.5){$\cdot$}
\put(0.09,53){$\cdot$} \put(0.06,52.5){$\cdot$}
\put(0.04,52){$\cdot$} \put(0.02,51.5){$\cdot$}
\put(0.01,51){$\cdot$} \put(0,50.5){$\cdot$} \put(0,50){$\cdot$}
\put(1.72,37){$\cdot$} \put(1.59,37.5){$\cdot$}
\put(1.46,38){$\cdot$} \put(1.34,38.5){$\cdot$}
\put(1.22,39){$\cdot$} \put(1.11,39.5){$\cdot$}
\put(1.01,40){$\cdot$} \put(0.99,40.5){$\cdot$}
\put(0.82,41){$\cdot$} \put(0.73,41.5){$\cdot$}
\put(0.64,42){$\cdot$} \put(0.56,42.5){$\cdot$}
\put(0.49,43){$\cdot$} \put(0.42,43.5){$\cdot$}
\put(0.36,44){$\cdot$} \put(0.3,44.5){$\cdot$}
\put(0.25,45){$\cdot$} \put(0.2,45.5){$\cdot$}
\put(0.16,46){$\cdot$} \put(0.12,46.5){$\cdot$}
\put(0.09,47){$\cdot$} \put(0.06,47.5){$\cdot$}
\put(0.04,48){$\cdot$} \put(0.02,48.5){$\cdot$}
\put(0.01,49){$\cdot$} \put(0,49.5){$\cdot$}
\put(-5,67){$\cl_{0,6}$}\put(-19,75){$\R$}
\end{picture}
\]
\vspace{1ex}
\begin{center}
{\small
{\bf Fig.\,3:} The first cycle of
$BW_{\R}\simeq\dZ_8$.}
\end{center}
\end{figure}

The second cycle ($h=1$, $r=1$) is started with a transition $\cl^+_{0,9}\overset{1}{\longrightarrow}\cl_{0,9}$ ($\cl_{0,8}\overset{1}{\longrightarrow}\cl_{0,9}$) and so on:\\
1) $\cl^+_{0,9}\overset{1}{\longrightarrow}\cl_{0,9}$ ($\cl_{0,8}\overset{1}{\longrightarrow}\cl_{0,9}$), $h=1$, $r=1$, $\R\overset{1}{\longrightarrow}\C$;\\
2) $\cl^+_{0,10}\overset{2}{\longrightarrow}\cl_{0,10}$ ($\cl_{0,9}\overset{2}{\longrightarrow}\cl_{0,10}$), $h=2$, $r=1$, $\C\overset{2}{\longrightarrow}\BH$;\\
3) $\cl^+_{0,11}\overset{3}{\longrightarrow}\cl_{0,11}$ ($\cl_{0,10}\overset{3}{\longrightarrow}\cl_{0,11}$), $h=3$, $r=1$, $\BH\overset{3}{\longrightarrow}\BH\oplus\BH$;\\
4) $\cl^+_{0,12}\overset{4}{\longrightarrow}\cl_{0,12}$ ($\cl_{0,11}\overset{4}{\longrightarrow}\cl_{0,12}$), $h=4$, $r=1$, $\BH\oplus\BH\overset{4}{\longrightarrow}\BH$;\\
5) $\cl^+_{0,13}\overset{5}{\longrightarrow}\cl_{0,13}$ ($\cl_{0,12}\overset{5}{\longrightarrow}\cl_{0,13}$), $h=5$, $r=1$, $\BH\overset{5}{\longrightarrow}\C$;\\
6) $\cl^+_{0,14}\overset{6}{\longrightarrow}\cl_{0,14}$ ($\cl_{0,13}\overset{6}{\longrightarrow}\cl_{0,14}$), $h=6$, $r=1$, $\C\overset{6}{\longrightarrow}\R$;\\
7) $\cl^+_{0,15}\overset{7}{\longrightarrow}\cl_{0,15}$ ($\cl_{0,14}\overset{7}{\longrightarrow}\cl_{0,15}$), $h=7$, $r=1$, $\R\overset{7}{\longrightarrow}\R\oplus\R$;\\
8) $\cl^+_{0,16}\overset{8}{\longrightarrow}\cl_{0,16}$ ($\cl_{0,15}\overset{8}{\longrightarrow}\cl_{0,16}$), $h=8$, $r=1$, $\R\oplus\R\overset{8}{\longrightarrow}\R$.\\
The full round of the second cycle is shown on the Fig.\,4.
\begin{figure}[ht]
\[
\unitlength=0.4mm
\begin{picture}(100.00,110.00)
\put(97,67){$\cl_{0,9}$}\put(108,78){$\C$}
\put(80,80){1} \put(75,93.3){$\cdot$} \put(75.5,93){$\cdot$}
\put(76,92.7){$\cdot$} \put(76.5,92.4){$\cdot$}
\put(77,92.08){$\cdot$} \put(77.5,91.76){$\cdot$}
\put(78,91.42){$\cdot$} \put(78.5,91.08){$\cdot$}
\put(79,90.73){$\cdot$} \put(79.5,90.37){$\cdot$}
\put(80,90.0){$\cdot$} \put(80.5,89.62){$\cdot$}
\put(81,89.23){$\cdot$} \put(81.5,88.83){$\cdot$}
\put(82,88.42){$\cdot$} \put(82.5,87.99){$\cdot$}
\put(83,87.56){$\cdot$} \put(83.5,87.12){$\cdot$}
\put(84,86.66){$\cdot$} \put(84.5,86.19){$\cdot$}
\put(85,85.70){$\cdot$} \put(85.5,85.21){$\cdot$}
\put(86,84.69){$\cdot$} \put(86.5,84.17){$\cdot$}
\put(87,83.63){$\cdot$} \put(87.5,83.07){$\cdot$}
\put(88,82.49){$\cdot$} \put(88.5,81.9){$\cdot$}
\put(89,81.29){$\cdot$} \put(89.5,80.65){$\cdot$}
\put(90,80){$\cdot$} \put(90.5,79.32){$\cdot$}
\put(91,78.62){$\cdot$} \put(91.5,77.89){$\cdot$}
\put(92,77.13){$\cdot$} \put(92.5,76.34){$\cdot$}
\put(93,75.51){$\cdot$} \put(93.5,74.65){$\cdot$}
\put(94,73.74){$\cdot$} \put(94.5,72.79){$\cdot$}
\put(96.5,73.74){\vector(1,-2){1}}
\put(80,20){3} \put(97,31){$\cl_{0,10}$}\put(108,42){$\BH$} \put(75,6.7){$\cdot$} \put(75.5,7){$\cdot$}
\put(76,7.29){$\cdot$} \put(76.5,7.6){$\cdot$}
\put(77,7.91){$\cdot$} \put(77.5,8.24){$\cdot$}
\put(78,8.57){$\cdot$} \put(78.5,8.91){$\cdot$}
\put(79,9.27){$\cdot$} \put(79.5,9.63){$\cdot$} \put(80,10){$\cdot$}
\put(80.5,10.38){$\cdot$} \put(81,10.77){$\cdot$}
\put(81.5,11.17){$\cdot$} \put(82,11.58){$\cdot$}
\put(82.5,12.00){$\cdot$} \put(83,12.44){$\cdot$}
\put(83.5,12.88){$\cdot$} \put(84,13.34){$\cdot$}
\put(84.5,13.8){$\cdot$} \put(85,14.29){$\cdot$}
\put(85.5,14.79){$\cdot$} \put(86,15.3){$\cdot$}
\put(86.5,15.82){$\cdot$} \put(87,16.37){$\cdot$}
\put(87.5,16.92){$\cdot$} \put(88,17.5){$\cdot$}
\put(88.5,18.09){$\cdot$} \put(89,18.71){$\cdot$}
\put(89.5,19.34){$\cdot$} \put(90,20){$\cdot$}
\put(90.5,20.68){$\cdot$} \put(91,21.38){$\cdot$}
\put(91.5,22.11){$\cdot$} \put(92,22.87){$\cdot$}
\put(92.5,23.66){$\cdot$} \put(93,24.48){$\cdot$}
\put(93.5,25.34){$\cdot$} \put(94,26.25){$\cdot$}
\put(94.5,27.20){$\cdot$} \put(95,28.20){$\cdot$}
\put(20,80){7} \put(25,93.3){$\cdot$} \put(24.5,93){$\cdot$}
\put(24,92.7){$\cdot$} \put(23.5,92.49){$\cdot$}
\put(23,92.08){$\cdot$} \put(22.5,91.75){$\cdot$}
\put(22,91.42){$\cdot$} \put(21.5,91.08){$\cdot$}
\put(21,90.73){$\cdot$} \put(20.5,90.37){$\cdot$}
\put(20,90){$\cdot$} \put(19.5,89.62){$\cdot$}
\put(19,89.23){$\cdot$} \put(18.5,88.83){$\cdot$}
\put(18,88.42){$\cdot$} \put(17.5,87.99){$\cdot$}
\put(17,87.56){$\cdot$} \put(16.5,87.12){$\cdot$}
\put(16,86.66){$\cdot$} \put(15.5,86.19){$\cdot$}
\put(15,85.70){$\cdot$} \put(14.5,85.21){$\cdot$}
\put(14,84.69){$\cdot$} \put(13.5,84.17){$\cdot$}
\put(13,83.63){$\cdot$} \put(12.5,83.07){$\cdot$}
\put(12,82.49){$\cdot$} \put(11.5,81.9){$\cdot$}
\put(11,81.29){$\cdot$} \put(10.5,80.65){$\cdot$}
\put(10,80){$\cdot$} \put(9.5,79.32){$\cdot$} \put(9,78.62){$\cdot$}
\put(8.5,77.89){$\cdot$} \put(8,77.13){$\cdot$}
\put(7.5,76.34){$\cdot$} \put(7,75.51){$\cdot$}
\put(6.5,74.65){$\cdot$} \put(6,73.79){$\cdot$}
\put(5.5,72.79){$\cdot$} \put(5,71.79){$\cdot$}
\put(20,20){5} \put(25,6.7){$\cdot$} \put(24.5,7){$\cdot$}
\put(24,7.29){$\cdot$} \put(23.5,7.6){$\cdot$}
\put(23,7.91){$\cdot$} \put(22.5,8.24){$\cdot$}
\put(22,8.57){$\cdot$} \put(21.5,8.91){$\cdot$}
\put(21,9.27){$\cdot$} \put(20.5,9.63){$\cdot$} \put(20,10){$\cdot$}
\put(19.5,10.38){$\cdot$} \put(19,10.77){$\cdot$}
\put(18.5,11.17){$\cdot$} \put(18,11.58){$\cdot$}
\put(17.5,12){$\cdot$} \put(17,12.44){$\cdot$}
\put(16.5,12.88){$\cdot$} \put(16,13.34){$\cdot$}
\put(15.5,13.8){$\cdot$} \put(15,14.29){$\cdot$}
\put(14.5,14.79){$\cdot$} \put(14,15.3){$\cdot$}
\put(13.5,15.82){$\cdot$} \put(13,16.37){$\cdot$}
\put(12.5,16.92){$\cdot$} \put(12,17.5){$\cdot$}
\put(11.5,18.09){$\cdot$} \put(11,18.71){$\cdot$}
\put(10.5,19.34){$\cdot$} \put(10,20){$\cdot$}
\put(9.5,20.68){$\cdot$} \put(9,21.38){$\cdot$}
\put(8.5,22.11){$\cdot$} \put(8,22.87){$\cdot$}
\put(7.5,23.66){$\cdot$} \put(7,24.48){$\cdot$}
\put(6.5,25.34){$\cdot$} \put(6,26.25){$\cdot$}
\put(5.5,27.20){$\cdot$} \put(5,28.20){$\cdot$}
\put(17,98){$\cl_{0,15}$}\put(-10,105){$\R\oplus\R$} \put(50,93){8} \put(50,100){$\cdot$}
\put(49.5,99.99){$\cdot$} \put(49,99.98){$\cdot$}
\put(48.5,99.97){$\cdot$} \put(48,99.96){$\cdot$}
\put(47.5,99.94){$\cdot$} \put(47,99.91){$\cdot$}
\put(46.5,99.86){$\cdot$} \put(46,99.84){$\cdot$}
\put(45.5,99.8){$\cdot$} \put(45,99.75){$\cdot$}
\put(44.5,99.7){$\cdot$} \put(44,99.64){$\cdot$}
\put(43.5,99.57){$\cdot$} \put(43,99.51){$\cdot$}
\put(42.5,99.43){$\cdot$} \put(42,99.35){$\cdot$}
\put(41.5,99.27){$\cdot$} \put(41,99.18){$\cdot$}
\put(40.5,99.09){$\cdot$} \put(40,98.99){$\cdot$}
\put(39.5,98.88){$\cdot$} \put(39,98.77){$\cdot$}
\put(38.5,98.66){$\cdot$} \put(38,98.54){$\cdot$}
\put(37.5,98.41){$\cdot$} \put(37,98.28){$\cdot$}
\put(50.5,99.99){$\cdot$} \put(51,99.98){$\cdot$}
\put(51.5,99.97){$\cdot$} \put(52,99.96){$\cdot$}
\put(52.5,99.94){$\cdot$} \put(53,99.91){$\cdot$}
\put(53.5,99.86){$\cdot$} \put(54,99.84){$\cdot$}
\put(54.5,99.8){$\cdot$} \put(55,99.75){$\cdot$}
\put(55.5,99.7){$\cdot$} \put(56,99.64){$\cdot$}
\put(56.5,99.57){$\cdot$} \put(57,99.51){$\cdot$}
\put(57.5,99.43){$\cdot$} \put(58,99.35){$\cdot$}
\put(58.5,99.27){$\cdot$} \put(59,99.18){$\cdot$}
\put(59.5,99.09){$\cdot$} \put(60,98.99){$\cdot$}
\put(60.5,98.88){$\cdot$} \put(61,98.77){$\cdot$}
\put(61.5,98.66){$\cdot$} \put(62,98.54){$\cdot$}
\put(62.5,98.41){$\cdot$} \put(63,98.28){$\cdot$}
\put(67,98){$\cl_{0,8}$}\put(83,105){$\R$}
\put(50,7){4} \put(67,1){$\cl_{0,11}$}\put(90,-4){$\BH\oplus\BH$} \put(50,0){$\cdot$} \put(50.5,0){$\cdot$}
\put(51,0.01){$\cdot$} \put(51.5,0.02){$\cdot$}
\put(52,0.04){$\cdot$} \put(52.5,0.06){$\cdot$}
\put(53,0.09){$\cdot$} \put(53.5,0.12){$\cdot$}
\put(54,0.16){$\cdot$} \put(54.5,0.2){$\cdot$}
\put(55,0.25){$\cdot$} \put(55.5,0.3){$\cdot$}
\put(56,0.36){$\cdot$} \put(56.5,0.42){$\cdot$}
\put(57,0.49){$\cdot$} \put(57.5,0.56){$\cdot$}
\put(58,0.64){$\cdot$} \put(58.5,0.73){$\cdot$}
\put(59,0.82){$\cdot$} \put(59.5,0.91){$\cdot$}
\put(60,1.01){$\cdot$} \put(60.5,1.11){$\cdot$}
\put(61,1.22){$\cdot$} \put(61.5,1.34){$\cdot$}
\put(62,1.46){$\cdot$} \put(62.5,1.59){$\cdot$}
\put(63,1.72){$\cdot$} \put(49.5,0){$\cdot$} \put(49,0.01){$\cdot$}
\put(48.5,0.02){$\cdot$} \put(48,0.04){$\cdot$}
\put(47.5,0.06){$\cdot$} \put(47,0.09){$\cdot$}
\put(46.5,0.12){$\cdot$} \put(46,0.16){$\cdot$}
\put(45.5,0.2){$\cdot$} \put(45,0.25){$\cdot$}
\put(44.5,0.3){$\cdot$} \put(44,0.36){$\cdot$}
\put(43.5,0.42){$\cdot$} \put(43,0.49){$\cdot$}
\put(42.5,0.56){$\cdot$} \put(42,0.64){$\cdot$}
\put(41.5,0.73){$\cdot$} \put(41,0.82){$\cdot$}
\put(40.5,0.91){$\cdot$} \put(40,1.01){$\cdot$}
\put(39.5,1.11){$\cdot$} \put(39,1.22){$\cdot$}
\put(38.5,1.34){$\cdot$} \put(38,1.46){$\cdot$}
\put(37.5,1.59){$\cdot$} \put(37,1.72){$\cdot$}
\put(18,1){$\cl_{0,12}$}\put(10,-4){$\BH$}
\put(93,50){2} \put(98.28,63){$\cdot$} \put(98.41,62.5){$\cdot$}
\put(98.54,62){$\cdot$} \put(98.66,61.5){$\cdot$}
\put(98.77,61){$\cdot$} \put(98.88,60.5){$\cdot$}
\put(98.99,60){$\cdot$} \put(99.09,59.5){$\cdot$}
\put(99.18,59){$\cdot$} \put(99.27,58.5){$\cdot$}
\put(99.35,58){$\cdot$} \put(99.43,57.5){$\cdot$}
\put(99.51,57){$\cdot$} \put(99.57,56.5){$\cdot$}
\put(99.64,56){$\cdot$} \put(99.7,55.5){$\cdot$}
\put(99.75,55){$\cdot$} \put(99.8,54.5){$\cdot$}
\put(99.84,54){$\cdot$} \put(99.86,53.5){$\cdot$}
\put(99.91,53){$\cdot$} \put(99.94,52.5){$\cdot$}
\put(99.96,52){$\cdot$} \put(99.97,51.5){$\cdot$}
\put(99.98,51){$\cdot$} \put(99.99,50.5){$\cdot$}
\put(100,50){$\cdot$} \put(98.28,37){$\cdot$}
\put(98.41,37.5){$\cdot$} \put(98.54,38){$\cdot$}
\put(98.66,38.5){$\cdot$} \put(98.77,39){$\cdot$}
\put(98.88,39.5){$\cdot$} \put(98.99,40){$\cdot$}
\put(99.09,40.5){$\cdot$} \put(99.18,41){$\cdot$}
\put(99.27,41.5){$\cdot$} \put(99.35,42){$\cdot$}
\put(99.43,42.5){$\cdot$} \put(99.51,43){$\cdot$}
\put(99.57,43.5){$\cdot$} \put(99.64,44){$\cdot$}
\put(99.7,44.5){$\cdot$} \put(99.75,45){$\cdot$}
\put(99.8,45.5){$\cdot$} \put(99.84,46){$\cdot$}
\put(99.86,46.5){$\cdot$} \put(99.91,47){$\cdot$}
\put(99.94,47.5){$\cdot$} \put(99.96,48){$\cdot$}
\put(99.97,48.5){$\cdot$} \put(99.98,49){$\cdot$}
\put(99.99,49.5){$\cdot$}
\put(7,50){6} \put(-5,32){$\cl_{0,13}$}\put(-19,29){$\C$} \put(1.72,63){$\cdot$} \put(1.59,62.5){$\cdot$}
\put(1.46,62){$\cdot$} \put(1.34,61.5){$\cdot$}
\put(1.22,61){$\cdot$} \put(1.11,60.5){$\cdot$}
\put(1.01,60){$\cdot$} \put(0.99,59.5){$\cdot$}
\put(0.82,59){$\cdot$} \put(0.73,58.5){$\cdot$}
\put(0.64,58){$\cdot$} \put(0.56,57.5){$\cdot$}
\put(0.49,57){$\cdot$} \put(0.42,56.5){$\cdot$}
\put(0.36,56){$\cdot$} \put(0.3,55.5){$\cdot$}
\put(0.25,55){$\cdot$} \put(0.2,54.5){$\cdot$}
\put(0.16,54){$\cdot$} \put(0.12,53.5){$\cdot$}
\put(0.09,53){$\cdot$} \put(0.06,52.5){$\cdot$}
\put(0.04,52){$\cdot$} \put(0.02,51.5){$\cdot$}
\put(0.01,51){$\cdot$} \put(0,50.5){$\cdot$} \put(0,50){$\cdot$}
\put(1.72,37){$\cdot$} \put(1.59,37.5){$\cdot$}
\put(1.46,38){$\cdot$} \put(1.34,38.5){$\cdot$}
\put(1.22,39){$\cdot$} \put(1.11,39.5){$\cdot$}
\put(1.01,40){$\cdot$} \put(0.99,40.5){$\cdot$}
\put(0.82,41){$\cdot$} \put(0.73,41.5){$\cdot$}
\put(0.64,42){$\cdot$} \put(0.56,42.5){$\cdot$}
\put(0.49,43){$\cdot$} \put(0.42,43.5){$\cdot$}
\put(0.36,44){$\cdot$} \put(0.3,44.5){$\cdot$}
\put(0.25,45){$\cdot$} \put(0.2,45.5){$\cdot$}
\put(0.16,46){$\cdot$} \put(0.12,46.5){$\cdot$}
\put(0.09,47){$\cdot$} \put(0.06,47.5){$\cdot$}
\put(0.04,48){$\cdot$} \put(0.02,48.5){$\cdot$}
\put(0.01,49){$\cdot$} \put(0,49.5){$\cdot$}
\put(-5,67){$\cl_{0,14}$}\put(-19,75){$\R$}
\end{picture}
\]
\vspace{1ex}
\begin{center}
{\small
{\bf Fig.\,4:} The second cycle of
$BW_{\R}\simeq\dZ_8$.}
\end{center}
\end{figure}

Further, the eighth cycle ($r=7$) finishes the construction of eight squares of the new spinorial chessboard (fractal self-similar algebraic structure of the second order):\\
1) $\cl^+_{0,57}\overset{1}{\longrightarrow}\cl_{0,57}$ ($\cl_{0,56}\overset{1}{\longrightarrow}\cl_{0,57}$), $h=1$, $r=7$, $\R\overset{1}{\longrightarrow}\C$;\\
2) $\cl^+_{0,58}\overset{2}{\longrightarrow}\cl_{0,58}$ ($\cl_{0,57}\overset{2}{\longrightarrow}\cl_{0,58}$), $h=2$, $r=7$, $\C\overset{2}{\longrightarrow}\BH$;\\
3) $\cl^+_{0,59}\overset{3}{\longrightarrow}\cl_{0,59}$ ($\cl_{0,58}\overset{3}{\longrightarrow}\cl_{0,59}$), $h=3$, $r=7$, $\BH\overset{3}{\longrightarrow}\BH\oplus\BH$;\\
4) $\cl^+_{0,60}\overset{4}{\longrightarrow}\cl_{0,60}$ ($\cl_{0,59}\overset{4}{\longrightarrow}\cl_{0,60}$), $h=4$, $r=7$, $\BH\oplus\BH\overset{4}{\longrightarrow}\BH$;\\
5) $\cl^+_{0,61}\overset{5}{\longrightarrow}\cl_{0,61}$ ($\cl_{0,60}\overset{5}{\longrightarrow}\cl_{0,61}$), $h=5$, $r=7$, $\BH\overset{5}{\longrightarrow}\C$;\\
6) $\cl^+_{0,62}\overset{6}{\longrightarrow}\cl_{0,62}$ ($\cl_{0,61}\overset{6}{\longrightarrow}\cl_{0,62}$), $h=6$, $r=7$, $\C\overset{6}{\longrightarrow}\R$;\\
7) $\cl^+_{0,63}\overset{7}{\longrightarrow}\cl_{0,63}$ ($\cl_{0,62}\overset{7}{\longrightarrow}\cl_{0,63}$), $h=7$, $r=7$, $\R\overset{7}{\longrightarrow}\R\oplus\R$;\\
8) $\cl^+_{0,64}\overset{8}{\longrightarrow}\cl_{0,64}$ ($\cl_{0,63}\overset{8}{\longrightarrow}\cl_{0,64}$), $h=8$, $r=7$, $\R\oplus\R\overset{8}{\longrightarrow}\R$.\\
It is obvious that a process, which leads us to the fractal self-similar algebraic structure of the second order (Fig.\,5), can be continued to any order (to infinity). Therefore, we come here to a fractal self-similar structure which is analogous to a Sierpi\'{n}ski carpet \cite{Mandel}. The fractal dimension (Besicovitch-Hausdorff dimension) of this structure is equal to $D=\ln 63/\ln 8\approx 1,9924$.
\begin{figure}[ht]
\unitlength=1mm
\begin{center}
\begin{picture}(100,100)
\put(0,0){\circle*{5}}
\put(3,-4){$0$}
\put(2.5,0){\line(1,0){5}}
\put(0,2.5){\line(0,1){5}}
\put(10,0){\circle{5}}
\put(13,-4){$1$}
\put(12.5,0){\line(1,0){5}}
\put(10,2.5){\line(0,1){5}}
\put(20,0){\circle*{5}}
\put(23,-4){$2$}
\put(22.5,0){\line(1,0){5}}
\put(20,2.5){\line(0,1){5}}
\put(30,0){\circle{5}}
\put(33,-4){$3$}
\put(32.5,0){\line(1,0){5}}
\put(30,2.5){\line(0,1){5}}
\put(40,0){\circle*{5}}
\put(43,-4){$4$}
\put(42.5,0){\line(1,0){5}}
\put(40,2.5){\line(0,1){5}}
\put(50,0){\circle{5}}
\put(53,-4){$5$}
\put(52.5,0){\line(1,0){5}}
\put(50,2.5){\line(0,1){5}}
\put(60,0){\circle*{5}}
\put(63,-4){$6$}
\put(62.5,0){\line(1,0){5}}
\put(60,2.5){\line(0,1){5}}
\put(70,0){\circle{5}}
\put(73,-4){$7$}
\put(70,2.5){\line(0,1){5}}
\put(0,10){\circle{5}}
\put(-5,6){$1$}
\put(2.5,10){\line(1,0){5}}
\put(0,12.5){\line(0,1){5}}
\put(10,10){\circle*{5}}
\put(12.5,10){\line(1,0){5}}
\put(10,12.5){\line(0,1){5}}
\put(20,10){\circle{5}}
\put(22.5,10){\line(1,0){5}}
\put(20,12.5){\line(0,1){5}}
\put(30,10){\circle*{5}}
\put(32.5,10){\line(1,0){5}}
\put(30,12.5){\line(0,1){5}}
\put(40,10){\circle{5}}
\put(42.5,10){\line(1,0){5}}
\put(40,12.5){\line(0,1){5}}
\put(50,10){\circle*{5}}
\put(52.5,10){\line(1,0){5}}
\put(50,12.5){\line(0,1){5}}
\put(60,10){\circle{5}}
\put(62.5,10){\line(1,0){5}}
\put(60,12.5){\line(0,1){5}}
\put(70,10){\circle*{5}}
\put(70,12.5){\line(0,1){5}}
\put(0,20){\circle*{5}}
\put(-5,16){$2$}
\put(2.5,20){\line(1,0){5}}
\put(0,22.5){\line(0,1){5}}
\put(10,20){\circle{5}}
\put(12.5,20){\line(1,0){5}}
\put(10,22.5){\line(0,1){5}}
\put(20,20){\circle*{5}}
\put(22.5,20){\line(1,0){5}}
\put(20,22.5){\line(0,1){5}}
\put(30,20){\circle{5}}
\put(32.5,20){\line(1,0){5}}
\put(30,22.5){\line(0,1){5}}
\put(40,20){\circle*{5}}
\put(42.5,20){\line(1,0){5}}
\put(40,22.5){\line(0,1){5}}
\put(50,20){\circle{5}}
\put(52.5,20){\line(1,0){5}}
\put(50,22.5){\line(0,1){5}}
\put(60,20){\circle*{5}}
\put(62.5,20){\line(1,0){5}}
\put(60,22.5){\line(0,1){5}}
\put(70,20){\circle{5}}
\put(70,22.5){\line(0,1){5}}
\put(0,30){\circle{5}}
\put(-5,26){$3$}
\put(2.5,30){\line(1,0){5}}
\put(0,32.5){\line(0,1){5}}
\put(10,30){\circle*{5}}
\put(12.5,30){\line(1,0){5}}
\put(10,32.5){\line(0,1){5}}
\put(20,30){\circle{5}}
\put(22.5,30){\line(1,0){5}}
\put(20,32.5){\line(0,1){5}}
\put(30,30){\circle*{5}}
\put(32.5,30){\line(1,0){5}}
\put(30,32.5){\line(0,1){5}}
\put(40,30){\circle{5}}
\put(42.5,30){\line(1,0){5}}
\put(40,32.5){\line(0,1){5}}
\put(50,30){\circle*{5}}
\put(52.5,30){\line(1,0){5}}
\put(50,32.5){\line(0,1){5}}
\put(60,30){\circle{5}}
\put(62.5,30){\line(1,0){5}}
\put(60,32.5){\line(0,1){5}}
\put(70,30){\circle*{5}}
\put(70,32.5){\line(0,1){5}}
\put(0,40){\circle*{5}}
\put(-5,36){$4$}
\put(2.5,40){\line(1,0){5}}
\put(0,42.5){\line(0,1){5}}
\put(10,40){\circle{5}}
\put(12.5,40){\line(1,0){5}}
\put(10,42.5){\line(0,1){5}}
\put(20,40){\circle*{5}}
\put(22.5,40){\line(1,0){5}}
\put(20,42.5){\line(0,1){5}}
\put(30,40){\circle{5}}
\put(32.5,40){\line(1,0){5}}
\put(30,42.5){\line(0,1){5}}
\put(40,40){\circle*{5}}
\put(42.5,40){\line(1,0){5}}
\put(40,42.5){\line(0,1){5}}
\put(50,40){\circle{5}}
\put(52.5,40){\line(1,0){5}}
\put(50,42.5){\line(0,1){5}}
\put(60,40){\circle*{5}}
\put(62.5,40){\line(1,0){5}}
\put(60,42.5){\line(0,1){5}}
\put(70,40){\circle{5}}
\put(70,42.5){\line(0,1){5}}
\put(0,50){\circle{5}}
\put(-5,46){$5$}
\put(2.5,50){\line(1,0){5}}
\put(0,52.5){\line(0,1){5}}
\put(10,50){\circle*{5}}
\put(12.5,50){\line(1,0){5}}
\put(10,52.5){\line(0,1){5}}
\put(20,50){\circle{5}}
\put(22.5,50){\line(1,0){5}}
\put(20,52.5){\line(0,1){5}}
\put(30,50){\circle*{5}}
\put(32.5,50){\line(1,0){5}}
\put(30,52.5){\line(0,1){5}}
\put(40,50){\circle{5}}
\put(42.5,50){\line(1,0){5}}
\put(40,52.5){\line(0,1){5}}
\put(50,50){\circle*{5}}
\put(52.5,50){\line(1,0){5}}
\put(50,52.5){\line(0,1){5}}
\put(60,50){\circle{5}}
\put(62.5,50){\line(1,0){5}}
\put(60,52.5){\line(0,1){5}}
\put(70,50){\circle*{5}}
\put(70,52.5){\line(0,1){5}}
\put(0,60){\circle*{5}}
\put(-5,56){$6$}
\put(2.5,60){\line(1,0){5}}
\put(0,62.5){\line(0,1){5}}
\put(10,60){\circle{5}}
\put(12.5,60){\line(1,0){5}}
\put(10,62.5){\line(0,1){5}}
\put(20,60){\circle*{5}}
\put(22.5,60){\line(1,0){5}}
\put(20,62.5){\line(0,1){5}}
\put(30,60){\circle{5}}
\put(32.5,60){\line(1,0){5}}
\put(30,62.5){\line(0,1){5}}
\put(40,60){\circle*{5}}
\put(42.5,60){\line(1,0){5}}
\put(40,62.5){\line(0,1){5}}
\put(50,60){\circle{5}}
\put(52.5,60){\line(1,0){5}}
\put(50,62.5){\line(0,1){5}}
\put(60,60){\circle*{5}}
\put(62.5,60){\line(1,0){5}}
\put(60,62.5){\line(0,1){5}}
\put(70,60){\circle{5}}
\put(70,62.5){\line(0,1){5}}
\put(0,70){\circle{5}}
\put(-5,66){$7$}
\put(2.5,70){\line(1,0){5}}
\put(10,70){\circle*{5}}
\put(12.5,70){\line(1,0){5}}
\put(20,70){\circle{5}}
\put(22.5,70){\line(1,0){5}}
\put(30,70){\circle*{5}}
\put(32.5,70){\line(1,0){5}}
\put(40,70){\circle{5}}
\put(42.5,70){\line(1,0){5}}
\put(50,70){\circle*{5}}
\put(52.5,70){\line(1,0){5}}
\put(60,70){\circle{5}}
\put(62.5,70){\line(1,0){5}}
\put(70,70){\circle*{5}}
\put(80,80){\circle*{5}}
\put(15,23){$\bullet$}
\put(16,24){\vector(1,0){38}}
\put(54,24){\vector(1,0){36.5}}
\put(90,23){$\bullet$}
\put(91,24){\vector(0,1){30}}
\put(91,53){\vector(0,1){3}}
\put(91,56){\vector(0,1){42}}
\put(90,98){$\bullet$}
\end{picture}
\end{center}
\vspace{0.3cm}
\begin{center}\begin{minipage}{32pc}{\small {\bf Fig.\,5:} \textbf{The Spinorial Chessboard of the second order}. Black and white circles (squares of the board) present spinorial chessboards of the first order which are shown on the Fig.\,1. These chessboards are distinguished from each other by the cycle number ($r=0,\ldots,7$) of the Brauer-Wall group $BW_{\R}\simeq\dZ_8$.}\end{minipage}\end{center}
\end{figure}

\begin{thm}
Over the field $\F=\R$ the action of the Brauer-Wall group $BW_{\R}\simeq\dZ_8$ induces modulo 4 periodic relations for the idempotent group $T_{p,q}(f)\simeq\left(\dZ_2\right)^{\otimes (k+1)}$:
\[
\left(\dZ_2\right)^{\otimes (1+q-r_{q-p}+4)}\simeq\left(\dZ_2\right)^{\otimes (1+q-r_{q-p})}\otimes\left(\dZ_2\right)^{\otimes 4}.
\]
\end{thm}
\begin{proof}
Let us consider several cycles of $BW_{\R}\simeq\dZ_8$ on the set of the groups $(\e_{\alpha_1},\e_{\alpha_2},\ldots,\e_{\alpha_k})\simeq\left(\dZ_2\right)^{\otimes
k}$. The first step $\cl_{0,0}\overset{1}{\longrightarrow}\cl_{0,1}$ of the first cycle ($h=1$, $r=1$) induces a transition $\left(\dZ_2\right)^{0}\overset{1}{\longrightarrow}\left(\dZ_2\right)^{0}$, since $k=q-r_{q-p}=0-r_0=0$ ($\cl_{0,0}$) and $k=q-r_{q-p}=1-r_1=0$ ($\cl_{0,1}$). The second step $\cl_{0,1}\overset{2}{\longrightarrow}\cl_{0,2}$ ($h=2$, $r=1$) gives also $\left(\dZ_2\right)^{0}\overset{1}{\longrightarrow}\left(\dZ_2\right)^{0}$, since $k=2-r_2=0$ for the algebra $\cl_{0,2}$. The third step $\cl_{0,2}\overset{3}{\longrightarrow}\cl_{0,3}$ of the first cycle ($h=3$, $r=1$) induces a transition $\left(\dZ_2\right)^{0}\overset{3}{\longrightarrow}\left(\dZ_2\right)^{\otimes 1}$ ($1\overset{3}{\longrightarrow}\dZ_2$), since in this case we have $k=3-r_3=1$ for $\cl_{0,3}$. The following step $\cl_{0,3}\overset{4}{\longrightarrow}\cl_{0,4}$ ($h=4$, $r=0$) of the group $BW_{\R}\simeq\dZ_8$ gives $\left(\dZ_2\right)^{\otimes 1}\overset{4}{\longrightarrow}\left(\dZ_2\right)^{\otimes 1}$ ($\dZ_2\overset{4}{\longrightarrow}\dZ_2$) in virtue of $k=4-r_4=1$ for $\cl_{0,4}$. The fifth step $\cl_{0,4}\overset{5}{\longrightarrow}\cl_{0,5}$ induces a transition $\left(\dZ_2\right)^{\otimes 1}\overset{5}{\longrightarrow}\left(\dZ_2\right)^{\otimes 2}$ ($\dZ_2\overset{5}{\longrightarrow}\dZ_2\otimes\dZ_2$), since $k=5-r_5=2$ for $\cl_{0,5}$. In turn, the sixth and seventh steps $\cl_{0,5}\overset{6}{\longrightarrow}\cl_{0,6}$ and $\cl_{0,6}\overset{7}{\longrightarrow}\cl_{0,7}$ of $BW_{\R}\simeq\dZ_8$ generate transitions $\left(\dZ_2\right)^{\otimes 2}\overset{6}{\longrightarrow}\left(\dZ_2\right)^{\otimes 3}$ ($\dZ_2\otimes\dZ_2\overset{6}{\longrightarrow}\dZ_2\otimes\dZ_2\otimes\dZ_2$) and $\left(\dZ_2\right)^{\otimes 3}\overset{7}{\longrightarrow}\left(\dZ_2\right)^{\otimes 4}$ ($\dZ_2\otimes\dZ_2\otimes\dZ_2\overset{7}{\longrightarrow}\dZ_2\otimes\dZ_2\otimes\dZ_2\otimes\dZ_4$), since $k=6-r_6=3$ ($\cl_{0,6}$) and $k=7-r_7=4$ ($\cl_{0,7}$). Finally, the eighth step $\cl_{0,7}\overset{8}{\longrightarrow}\cl_{0,8}$ induces a transition $\left(\dZ_2\right)^{\otimes 4}\overset{8}{\longrightarrow}\left(\dZ_2\right)^{\otimes 4}$, since in virtue of the modulo 8 periodic relation $r_{i+8}=r_i+4$ for the Radon-Hurwitz numbers we have $k=q-r_{q-p}=8-r_8=4-r_0=4$ ($\cl_{0,8}$).

The second cycle ($r=1$) is started by the step $\cl_{0,8}\overset{1}{\longrightarrow}\cl_{0,9}$ which leads to $\left(\dZ_2\right)^{\otimes 4}\overset{1}{\longrightarrow}\left(\dZ_2\right)^{\otimes 4}$ and so on:\\
1) $h=1$, $r=1$, $\cl_{0,8}\overset{1}{\longrightarrow}\cl_{0,9}\;\mapsto\;$ $\left(\dZ_2\right)^{\otimes 4}\overset{1}{\longrightarrow}\left(\dZ_2\right)^{\otimes 4}$, $k=9-r_9=5-r_1=4$;\\
2) $h=2$, $r=1$, $\cl_{0,9}\overset{2}{\longrightarrow}\cl_{0,10}\;\mapsto\;$ $\left(\dZ_2\right)^{\otimes 4}\overset{2}{\longrightarrow}\left(\dZ_2\right)^{\otimes 4}$, $k=10-r_{10}=6-r_2=4$;\\
3) $h=3$, $r=1$, $\cl_{0,10}\overset{3}{\longrightarrow}\cl_{0,11}\;\mapsto\;$ $\left(\dZ_2\right)^{\otimes 4}\overset{3}{\longrightarrow}\left(\dZ_2\right)^{\otimes 5}$, $k=11-r_{11}=7-r_3=5$;\\
4) $h=4$, $r=1$, $\cl_{0,11}\overset{4}{\longrightarrow}\cl_{0,12}\;\mapsto\;$ $\left(\dZ_2\right)^{\otimes 5}\overset{4}{\longrightarrow}\left(\dZ_2\right)^{\otimes 5}$, $k=12-r_{12}=8-r_4=5$;\\
5) $h=5$, $r=1$, $\cl_{0,12}\overset{5}{\longrightarrow}\cl_{0,13}\;\mapsto\;$ $\left(\dZ_2\right)^{\otimes 5}\overset{5}{\longrightarrow}\left(\dZ_2\right)^{\otimes 6}$, $k=13-r_{13}=9-r_5=6$;\\
6) $h=6$, $r=1$, $\cl_{0,13}\overset{6}{\longrightarrow}\cl_{0,14}\;\mapsto\;$ $\left(\dZ_2\right)^{\otimes 6}\overset{6}{\longrightarrow}\left(\dZ_2\right)^{\otimes 7}$, $k=14-r_{14}=10-r_6=7$;\\
7) $h=7$, $r=1$, $\cl_{0,14}\overset{7}{\longrightarrow}\cl_{0,15}\;\mapsto\;$ $\left(\dZ_2\right)^{\otimes 7}\overset{7}{\longrightarrow}\left(\dZ_2\right)^{\otimes 8}$, $k=15-r_{15}=11-r_7=8$;\\
8) $h=8$, $r=1$, $\cl_{0,15}\overset{8}{\longrightarrow}\cl_{0,16}\;\mapsto\;$ $\left(\dZ_2\right)^{\otimes 8}\overset{8}{\longrightarrow}\left(\dZ_2\right)^{\otimes 8}$, $k=16-r_{16}=8-r_0=8$.

In its turn, the third cycle of $BW_{\R}\simeq\dZ_8$ ($r=2$) is realized on the set of idempotent groups via the following eight steps:\\
1) $h=1$, $r=2$, $\cl_{0,16}\overset{1}{\longrightarrow}\cl_{0,17}\;\mapsto\;$ $\left(\dZ_2\right)^{\otimes 8}\overset{1}{\longrightarrow}\left(\dZ_2\right)^{\otimes 8}$, $k=17-r_{17}=9-r_1=8$;\\
2) $h=2$, $r=2$, $\cl_{0,17}\overset{2}{\longrightarrow}\cl_{0,18}\;\mapsto\;$ $\left(\dZ_2\right)^{\otimes 8}\overset{2}{\longrightarrow}\left(\dZ_2\right)^{\otimes 8}$, $k=18-r_{18}=10-r_2=8$;\\
3) $h=3$, $r=2$, $\cl_{0,18}\overset{3}{\longrightarrow}\cl_{0,19}\;\mapsto\;$ $\left(\dZ_2\right)^{\otimes 8}\overset{3}{\longrightarrow}\left(\dZ_2\right)^{\otimes 9}$, $k=19-r_{19}=11-r_3=9$;\\
4) $h=4$, $r=2$, $\cl_{0,19}\overset{4}{\longrightarrow}\cl_{0,20}\;\mapsto\;$ $\left(\dZ_2\right)^{\otimes 9}\overset{4}{\longrightarrow}\left(\dZ_2\right)^{\otimes 9}$, $k=20-r_{20}=12-r_4=9$;\\
5) $h=5$, $r=2$, $\cl_{0,20}\overset{5}{\longrightarrow}\cl_{0,21}\;\mapsto\;$ $\left(\dZ_2\right)^{\otimes 9}\overset{5}{\longrightarrow}\left(\dZ_2\right)^{\otimes 10}$, $k=21-r_{21}=13-r_5=10$;\\
6) $h=6$, $r=2$, $\cl_{0,21}\overset{6}{\longrightarrow}\cl_{0,22}\;\mapsto\;$ $\left(\dZ_2\right)^{\otimes 10}\overset{6}{\longrightarrow}\left(\dZ_2\right)^{\otimes 11}$, $k=22-r_{22}=14-r_6=11$;\\
7) $h=7$, $r=2$, $\cl_{0,22}\overset{7}{\longrightarrow}\cl_{0,23}\;\mapsto\;$ $\left(\dZ_2\right)^{\otimes 11}\overset{7}{\longrightarrow}\left(\dZ_2\right)^{\otimes 12}$, $k=23-r_{23}=15-r_7=12$;\\
8) $h=8$, $r=2$, $\cl_{0,23}\overset{8}{\longrightarrow}\cl_{0,24}\;\mapsto\;$ $\left(\dZ_2\right)^{\otimes 12}\overset{8}{\longrightarrow}\left(\dZ_2\right)^{\otimes 12}$, $k=24-r_{24}=12-r_0=12$.

Hence it immediately follows that the each cycle of $BW_{\R}\simeq\dZ_8$ generates on the set of $(\e_{\alpha_1},\e_{\alpha_2},\ldots,\e_{\alpha_k})\simeq\left(\dZ_2\right)^{\otimes
k}$ a modulo 4 periodic relation $\left(\dZ_2\right)^{\otimes (k+4)}\simeq\left(\dZ_2\right)^{\otimes k}\otimes\left(\dZ_2\right)^{\otimes 4}$. Therefore, for the idempotent group $T_{p,q}(f)$ we have
\[
\left(\dZ_2\right)^{\otimes (1+q-r_{q-p}+4)}\simeq\left(\dZ_2\right)^{\otimes (1+q-r_{q-p})}\otimes\left(\dZ_2\right)^{\otimes 4}.
\]
\end{proof}
\subsection{Spinor groups and quaternionic factorizations}
A universal covering $\spin_+(p,q)$ of the rotation group
$\SO_0(p,q)$ of the pseudo-Euclidean space $\R^{p,q}$ is described
in terms of even subalgebra $\cl^+_{p,q}$. In its turn, the
subalgebra $\cl^+_{p,q}$ admits the following isomorphisms:
$\cl^+_{p,q}\simeq\cl_{q,p-1}$, $\cl^+_{p,q}\simeq\cl_{p,q-1}$. At this point, for the algebras of general type $\cl_{p,q}$ ($p\neq q$) and $\cl_{p,0}$ we have
\begin{equation}\label{Isom1}
\cl^+_{p,q}\simeq\cl_{q,p-1}.
\end{equation}
The isomorphism
\begin{equation}\label{Isom2}
\cl^+_{p,q}\simeq\cl_{p,q-1}
\end{equation}
takes place for the algebras of type $\cl_{0,q}$ and $\cl_{p,p}$ ($p=q$).
In dependence on dimension of the associated space $\R^{p,q}$, the
subalgebras $\cl_{p,q-1}$ and $\cl_{q,p-1}$ have even dimensions
($p+q-1\equiv 0\pmod{2}$) or odd dimensions ($p+q-1\equiv
1\pmod{2}$).

When $p+q-1\equiv 1\pmod{2}$ we have four types $q-p+1\equiv
1,3,5,7\pmod{8}$ ($p-q+1\equiv 1,3,5,7\pmod{8}$) of the subalgebras
$\cl_{q,p-1}$ ($\cl_{p,q-1}$). In this case a center $\bZ_{q,p-1}$
($\bZ_{p,q-1}$) of $\cl_{q,p-1}$ ($\cl_{p,q-1}$) consists of the
unit $\e_0$ and the volume element
$\omega=\e_1\e_2\cdots\e_{p+q-1}$. When $\omega^2=-1$ ($q-p+1\equiv
3,7\pmod{8}$ or $p-q+1\equiv 3,7\pmod{8}$) we arrive at the
isomorphisms
$\cl_{q,p-1}\simeq\C_{q+p-2}\simeq\C\left(2^{\frac{q+p-2}{2}}\right)$, $\cl_{0,p-1}\simeq\C_{p-2}\simeq\C\left(2^{\frac{p-2}{2}}\right)$ and $\cl_{0,q-1}\simeq\C_{q-2}\simeq\C\left(2^{\frac{q-2}{2}}\right)$.
The transition from $\cl_{q,p-1}\rightarrow\C_{q+p-2}$, $\cl_{0,p-1}\rightarrow\C_{p-2}$ and $\cl_{0,q-1}\rightarrow\C_{q-2}$ can be represented as the transition from the real coordinates in
$\cl_{q,p-1}$, $\cl_{0,p-1}$ and $\cl_{0,q-1}$ to complex coordinates $a+\omega b$ in
$\C_{q+p-2}$, $\C_{p-2}$ and $\C_{q-2}$, where $\omega=\e_1\e_2\cdots\e_{p+q-1}\in\cl_{q,p-1}$, $\omega=\e_1\e_2\cdots\e_{p-1}\in\cl_{0,p-1}$ and $\omega=\e_1\e_2\cdots\e_{q-1}\in\cl_{0,q-1}$. On the other
hand, when $\omega^2=1$ ($q-p+1\equiv 1,5\pmod{8}$ or $p-q+1\equiv
1,5\pmod{8}$) we arrive at the isomorphisms
\begin{eqnarray}
\cl_{q,p-1}&\simeq&\cl_{p-1,q-1}\oplus\cl_{p-1,q-1}\simeq
\R\left(2^{\frac{p+q-2}{2}}\right)\oplus\R\left(2^{\frac{p+q-2}{2}}\right),\nonumber\\
\cl_{0,p-1}&\simeq&\cl_{0,p-2}\oplus\cl_{0,p-2}\simeq
\R\left(2^{\frac{p-2}{2}}\right)\oplus\R\left(2^{\frac{p-2}{2}}\right),\nonumber\\
\cl_{0,q-1}&\simeq&\cl_{0,q-2}\oplus\cl_{0,q-2}\simeq
\R\left(2^{\frac{q-2}{2}}\right)\oplus\R\left(2^{\frac{q-2}{2}}\right)\nonumber
\end{eqnarray}
if $q-p+1\equiv 1\pmod{8}$ (or $p-1\equiv 1\pmod{8}$ and $q-1\equiv 1\pmod{8}$). When $q-p+1\equiv 5\pmod{8}$ (or $p-1\equiv 5\pmod{8}$ and $q-1\equiv 5\pmod{8}$) we have
\begin{eqnarray}
\cl_{q,p-1}&\simeq&\cl_{p-1,q-1}\oplus\cl_{p-1,q-1}\simeq
\BH\left(2^{\frac{p+q-4}{2}}\right)\oplus\BH\left(2^{\frac{p+q-4}{2}}\right),\nonumber\\
\cl_{0,p-1}&\simeq&\cl_{0,p-2}\oplus\cl_{0,p-2}\simeq
\BH\left(2^{\frac{p-4}{2}}\right)\oplus\BH\left(2^{\frac{p-4}{2}}\right),\nonumber\\
\cl_{0,q-1}&\simeq&\cl_{0,q-2}\oplus\cl_{0,q-2}\simeq
\BH\left(2^{\frac{q-4}{2}}\right)\oplus\BH\left(2^{\frac{q-4}{2}}\right).\nonumber
\end{eqnarray}
The transitions $\cl_{q,p-1}\rightarrow{}^2\R\left(2^{\frac{p+q-2}{2}}\right)$, $\cl_{0,p-1}\rightarrow{}^2\R\left(2^{\frac{p-2}{2}}\right)$, $\cl_{0,q-1}\rightarrow{}^2\R\left(2^{\frac{q-2}{2}}\right)$ and $\cl_{q,p-1}\rightarrow{}^2\BH\left(2^{\frac{p+q-4}{2}}\right)$, $\cl_{0,p-1}\rightarrow{}^2\BH\left(2^{\frac{p-4}{2}}\right)$, $\cl_{0,q-1}\rightarrow{}^2\BH\left(2^{\frac{q-4}{2}}\right)$
can be represented as the transition from the real coordinates in
$\cl_{q,p-1}$, $\cl_{0,p-1}$, $\cl_{0,q-1}$ to double coordinates $a+\omega b$ in
${}^2\R\left(2^{\frac{p+q-2}{2}}\right)$, ${}^2\R\left(2^{\frac{p-2}{2}}\right)$, ${}^2\R\left(2^{\frac{q-2}{2}}\right)$ and
${}^2\BH\left(2^{\frac{p+q-4}{2}}\right)$, ${}^2\BH\left(2^{\frac{p-4}{2}}\right)$, ${}^2\BH\left(2^{\frac{q-4}{2}}\right)$.

Further, when $p+q-1\equiv 0\pmod{2}$ we have four types
$q-p+1\equiv 0,2,4,6\pmod{8}$ ($p-q+1\equiv 0,2,4,6\pmod{8}$) of the
subalgebras $\cl_{q,p-1}$ ($\cl_{p,q-1}$). Let $p+q-1\geq 4$ and let
$m=(p+q-1)/4$ be an integer number (this number is always integer,
since $p+q-1\equiv 0\pmod{2}$). Then, at $i\leq 2m$ we have
\begin{eqnarray}
\e_{12\ldots 2m 2m+k}\e_{i}&=&(-1)^{2m+1-i}\sigma(i-l)
\e_{12\ldots i-1 i+1\ldots 2m 2m+k},\nonumber \\
\e_{i}\e_{12\ldots 2m 2m+k}&=&(-1)^{i-1}\sigma(i-l) \e_{12\ldots i-1
i+1\ldots 2m 2m+k}\nonumber
\end{eqnarray}
and, therefore, the condition of commutativity of elements
$\e_{12\ldots 2m 2m+k}$ and $\e_{i}$ is $2m+1-i\equiv i-1\pmod{2}$.
Thus, the elements $\e_{12\ldots 2m 2m+1}$ and $\e_{12\ldots 2m
2m+2}$ commute with all the elements $\e_{i}$ whose indices do not
exceed $2m$. Therefore, a transition from the algebra $\cl_{q,p-1}$
to algebras $\cl_{q+2,p-1}$, $\cl_{q,p+1}$ and
$\cl_{q+1,p}$ can be represented as a transition from the real
coordinates in $\cl_{q,p-1}$ to coordinates of the
form $a+b\phi+c\psi+d\phi\psi$, where $\phi$ and $\psi$ are
additional basis elements $\e_{12\ldots 2m 2m+1}$ and $\e_{12\ldots
2m 2m+2}$. The elements $\e_{i_{1}i_{2}\ldots i_{k}}\phi$ contain
the index $2m+1$ and do not contain the index $2m+2$. In its turn,
the elements $\e_{i_{1}i_{2}\ldots i_{k}}\psi$ contain the index
$2m+2$ and do not contain the index $2m+1$. Correspondingly, the
elements $\e_{i_{1}i_{2}\ldots i_{k}} \phi\psi$ contain both indices
$2m+1$ and $2m+2$. Analogously, we have a transition from the algebra $\cl_{0,p-1}$ to $\cl_{2,p-1}$, $\cl_{0,p+1}$ and $\cl_{1,p}$, from the algebra $\cl_{0,q-1}$ to $\cl_{2,q-1}$, $\cl_{0,q+1}$ and $\cl_{1,q}$, and from the algebra $\cl_{p,p-1}$ to $\cl_{p+2,p-1}$, $\cl_{p,p+1}$, $\cl_{p+1,p}$. Hence it immediately follows that general
elements of these algebras can be represented as
\begin{equation}
\cA_{\cl_{q+2,p-1}}=\cl^0_{q,p-1}\e_0+\cl^1_{q,p-1}\phi+\cl^2_{q,p-1}\psi+
\cl^3_{q,p-1}\phi\psi,\label{El1}
\end{equation}
\begin{equation}
\cA_{\cl_{q,p+1}}=\cl^0_{q,p-1}\e_0+\cl^1_{q,p-1}\phi+\cl^2_{q,p-1}\psi+
\cl^3_{q,p-1}\phi\psi,\label{El2}
\end{equation}
\begin{equation}
\cA_{\cl_{q+1,p}}=\cl^0_{q,p-1}\e_0+\cl^1_{q,p-1}\phi+\cl^2_{q,p-1}\psi+
\cl^3_{q,p-1}\phi\psi,\label{El3}
\end{equation}
\begin{equation}
\cA_{\cl_{2,p-1}}=\cl^0_{0,p-1}\e_0+\cl^1_{0,p-1}\phi+\cl^2_{0,p-1}\psi+
\cl^3_{0,p-1}\phi\psi,\label{El4}
\end{equation}
\begin{equation}
\cA_{\cl_{0,p+1}}=\cl^0_{0,p-1}\e_0+\cl^1_{0,p-1}\phi+\cl^2_{0,p-1}\psi+
\cl^3_{0,p-1}\phi\psi,\label{El5}
\end{equation}
\begin{equation}
\cA_{\cl_{1,p}}=\cl^0_{0,p-1}\e_0+\cl^1_{0,p-1}\phi+\cl^2_{0,p-1}\psi+
\cl^3_{0,p-1}\phi\psi,\label{El6}
\end{equation}
\begin{equation}
\cA_{\cl_{2,q-1}}=\cl^0_{0,q-1}\e_0+\cl^1_{0,q-1}\phi+\cl^2_{0,q-1}\psi+
\cl^3_{0,q-1}\phi\psi,\label{El7}
\end{equation}
\begin{equation}
\cA_{\cl_{0,q+1}}=\cl^0_{0,q-1}\e_0+\cl^1_{0,q-1}\phi+\cl^2_{0,q-1}\psi+
\cl^3_{0,q-1}\phi\psi,\label{El8}
\end{equation}
\begin{equation}
\cA_{\cl_{1,q}}=\cl^0_{0,q-1}\e_0+\cl^1_{0,q-1}\phi+\cl^2_{0,q-1}\psi+
\cl^3_{0,q-1}\phi\psi,\label{El9}
\end{equation}
\begin{equation}
\cA_{\cl_{p+2,p-1}}=\cl^0_{p,p-1}\e_0+\cl^1_{p,p-1}\phi+\cl^2_{p,p-1}\psi+
\cl^3_{p,p-1}\phi\psi,\label{El10}
\end{equation}
\begin{equation}
\cA_{\cl_{p,p+1}}=\cl^0_{p,p-1}\e_0+\cl^1_{p,p-1}\phi+\cl^2_{p,p-1}\psi+
\cl^3_{p,p-1}\phi\psi,\label{El11}
\end{equation}
\begin{equation}
\cA_{\cl_{p+1,p}}=\cl^0_{p,p-1}\e_0+\cl^1_{p,p-1}\phi+\cl^2_{p,p-1}\psi+
\cl^3_{p,p-1}\phi\psi,\label{El12}
\end{equation}
\begin{sloppypar}\noindent
where $\cl^i_{q,p-1}$, $\cl^i_{0,p-1}$, $\cl^i_{0,q-1}$, $\cl^i_{p,p-1}$, are the algebras with a
general element ${\cal A}=\sum^{2m}_{k=0}a^{i_{1}i_{2}\ldots i_{k}}
\e_{i_{1}i_{2}\ldots i_{k}}$. When the elements $\phi=\e_{12\ldots
2m 2m+1}$ and $\psi=\e_{12\ldots 2m 2m+2}$ satisfy the condition
$\phi^{2}=\psi^{2}=-1$ we see that the basis
$\{\e_{0},\,\phi,\,\psi,\, \phi\psi\}$ is isomorphic to a basis of
the quaternion algebra $\cl_{0,2}$, therefore, the elements
(\ref{El2}), (\ref{El5}), (\ref{El8}) and (\ref{El11}) are general elements of quaternionic
algebras. In turn, when the elements $\phi$ and $\psi$ satisfy
the condition $\phi^{2}=-\psi^{2}=1$, the basis
$\{\e_{0},\,\phi,\,\psi,\, \phi\psi\}$ is isomorphic to a basis of
the anti-quaternion algebra $\cl_{1,1}$, and the elements
(\ref{El3}), (\ref{El6}), (\ref{El9}) and (\ref{El12}) are general elements of
anti-quaternionic algebras. Further, when $\phi^{2}=\psi^{2}=1$ we
have a basis of the pseudo-quaternion algebra $\cl_{2,0}$, and the
elements (\ref{El1}), (\ref{El4}), (\ref{El7}) and (\ref{El10}) are general elements of
pseudo-quaternionic algebras.\end{sloppypar}

Let us define matrix representations of the quaternion units $\phi$
and $\psi$ as follows:
\[
\phi\longmapsto\begin{bmatrix} 0 & -1\\ 1 & 0\end{bmatrix},\quad
\psi\longmapsto\begin{bmatrix} 0 & i\\ i & 0\end{bmatrix}.
\]
Using these representations and (\ref{El2}), we obtain
\[
\cl_{p,q+1}\simeq\Mat_2(\cl_{p,q-1})=\begin{bmatrix}
\cl^0_{p,q-1}-i\cl^3_{p,q-1} & -\cl^1_{p,q-1}+i\cl^2_{p,q-1}\\
\cl^1_{p,q-1}+i\cl^2_{p,q-1} & \cl^0_{p,q-1}+i\cl^3_{p,q-1}
\end{bmatrix}.
\]
The analogous expression takes place for the algebra $\cl_{q,p+1}$
with the element (\ref{El5}) and so on.

For example, let us consider the algebra $\cl_{2,4}$ associated with a six-dimensional pseudo-Euclidean space $\R^{2,4}$. A universal covering $\spin_+(2,4)$ of the rotation group $\SO_0(2,4)$ of $\R^{2,4}$ is described in terms of even subalgebra $\cl^+_{2,4}$. The algebra $\cl_{2,4}$ has the type $p-q\equiv 6\pmod{8}$, therefore, from (\ref{Isom1}) we have $\cl^+_{2,4}\simeq\cl_{4,1}$, where $\cl_{4,1}$ is a de Sitter algebra associated with the space $\R^{4,1}$. In its turn, $\cl_{4,1}$ has the type $p-q\equiv 3\pmod{8}$ and, therefore, there is an isomorphism $\cl_{4,1}\simeq\C_4$, where $\C_4$ is a Dirac algebra. The algebra $\C_4$ is a complexification of space-time algebra: $\C_4\simeq\C\otimes\cl_{1,3}$. Further, the space-time algebra $\cl_{1,3}$ admits the following factorization: $\cl_{1,3}\simeq\cl_{1,1}\otimes\cl_{0,2}$. Hence it immediately follows that $\cl_{1,3}\simeq\C\otimes\cl_{1,1}\otimes\cl_{0,2}$. Thus,
\begin{equation}\label{ConfGroup0}
\spin_+(2,4)=\left\{s\in\C\otimes\cl_{1,1}\otimes\cl_{0,2}\;|\;N(s)=1\right\}.
\end{equation}
On the other hand, in virtue of $\cl_{1,3}\simeq\cl_{1,1}\otimes\cl_{0,2}$ from (\ref{El2}) we have
\[
\cA_{\cl_{1,3}}=\cl^0_{1,1}\e_0+\cl^1_{1,1}\phi+\cl^2_{1,1}\psi+\cl^3_{1,1}\phi\psi,
\]
where $\phi=\e_{123}$, $\psi=\e_{124}$. Therefore,
\begin{equation}\label{ConfGroup}{\renewcommand{\arraystretch}{1.2}
\spin_+(2,4)=\left\{s\in\left.\begin{bmatrix} \C\otimes\cl^0_{1,1}-i\C\otimes\cl^3_{1,1} &
-\C\otimes\cl^1_{1,1}+i\C\otimes\cl^2_{1,1}\\
\C\otimes\cl^1_{1,1}+i\C\otimes\cl^2_{1,1} & \C\otimes\cl^0_{1,1}+i\C\otimes\cl^3_{1,1}\end{bmatrix}\right|\;N(s)=1
\right\}.}
\end{equation}

\section{Spinor structure and the group $\spin_+(1,3)$}
\emph{Any} irreducible finite dimensional representation $\boldsymbol{\tau}_{l\dot{l}}$ of the group $\SL(2,\C)\simeq\spin_+(1,3)$ corresponds to a \textbf{\emph{particle of the spin}} $s$, where $s=|l-\dot{l}|$ (see also \cite{Var11,Var1402}). All the values of $s$ are
\[
-s,\;\;-s+1,\;\;-s+2,\;\;\ldots,\;\;s
\]
or
\begin{equation}\label{SValues}
-|l-\dot{l}|,\;\;-|l-\dot{l}|+1,\;\;-|l-\dot{l}|+2,\;\;\ldots,\;\;|l-\dot{l}|.
\end{equation}
Here the numbers $l$ and $\dot{l}$ are
\[
l=\frac{k}{2},\quad\dot{l}=\frac{r}{2},
\]
where $k$ and $r$ are factor quantities in the tensor product
\begin{equation}\label{TenAlg}
\underbrace{\C_2\otimes\C_2\otimes\cdots\otimes\C_2}_{k\;\text{times}}\bigotimes
\underbrace{\overset{\ast}{\C}_2\otimes\overset{\ast}{\C}_2\otimes\cdots\otimes
\overset{\ast}{\C}_2}_{r\;\text{times}}
\end{equation}
associated with the representation $\boldsymbol{\tau}_{k/2,r/2}$ of $\SL(2,\C)$, where $\C_2$ and complex conjugate $\overset{\ast}{\C}_2$ are biquaternion algebras. In turn, a \emph{spinspace} $\dS_{2^{k+r}}$, associated with the tensor product (\ref{TenAlg}), is
\begin{equation}\label{SpinSpace}
\underbrace{\dS_2\otimes\dS_2\otimes\cdots\otimes\dS_2}_{k\;\text{times}}\bigotimes
\underbrace{\dot{\dS}_2\otimes\dot{\dS}_2\otimes\cdots\otimes\dot{\dS}_2}_{r\;\text{times}}.
\end{equation}
Usual definition of the spin we obtain at the restriction $\boldsymbol{\tau}_{l\dot{l}}\rightarrow\boldsymbol{\tau}_{l,0}$ (or $\boldsymbol{\tau}_{l\dot{l}}\rightarrow\boldsymbol{\tau}_{0,\dot{l}}$), that is, at the restriction of $\SL(2,\C)$ to its subgroup $\SU(2)$. In this case the sequence of spin values (\ref{SValues}) is reduced to $-l$, $-l+1$, $-l+2$, $\ldots$, $l$ (or $-\dot{l}$, $-\dot{l}+1$, $-\dot{l}+2$, $\ldots$, $\dot{l}$).

Let
\[
\boldsymbol{S}=\boldsymbol{s}^{\alpha_1\alpha_2\ldots\alpha_k\dot{\alpha}_1\dot{\alpha}_2\ldots
\dot{\alpha}_r}=\sum \boldsymbol{s}^{\alpha_1}\otimes
\boldsymbol{s}^{\alpha_2}\otimes\cdots\otimes
\boldsymbol{s}^{\alpha_k}\otimes
\boldsymbol{s}^{\dot{\alpha}_1}\otimes
\boldsymbol{s}^{\dot{\alpha}_2}\otimes\cdots\otimes
\boldsymbol{s}^{\dot{\alpha}_r}
\]
be a spintensor polynomial, then any pair of substitutions
\[
\alpha=\begin{pmatrix} 1 & 2 & \ldots & k\\
\alpha_1 & \alpha_2 & \ldots & \alpha_k\end{pmatrix},\quad \beta=\begin{pmatrix} 1 & 2 & \ldots & r\\
\dot{\alpha}_1 & \dot{\alpha}_2 & \ldots & \dot{\alpha}_r\end{pmatrix}
\]
defines a transformation $(\alpha,\beta)$ mapping $\boldsymbol{S}$ to the following polynomial:
\[
P_{\alpha\beta}\boldsymbol{S}=\boldsymbol{s}^{\alpha\left(\alpha_1\right)\alpha\left(\alpha_2\right)\ldots
\alpha\left(\alpha_k\right)\beta\left(\dot{\alpha}_1\right)\beta\left(\dot{\alpha}_2\right)\ldots
\beta\left(\dot{\alpha}_r\right)}.
\]
The spintensor $\boldsymbol{S}$ is called a \emph{symmetric spintensor} if at any $\alpha$, $\beta$ the equality
\[
P_{\alpha\beta}\boldsymbol{S}=\boldsymbol{S}
\]
holds. The space $\Sym_{(k,r)}$ of symmetric spintensors has the dimensionality
\begin{equation}\label{Degree}
\dim\Sym_{(k,r)}=(k+1)(r+1).
\end{equation}
The dimensionality of $\Sym_{(k,r)}$ is called a \emph{degree of the representation} $\boldsymbol{\tau}_{l\dot{l}}$ of the group $\SL(2,\C)$. It is easy to see that $\SL(2,\C)$ has representations of \textbf{\emph{any degree}} (in contrast to $\SU(3)$, $\SU(6)$ and other groups of internal symmetries, see \cite{Var1402}).

For the each $A\in\SL(2,\C)$ we define a linear transformation of the spintensor $\boldsymbol{s}$ via the formula
\[
\boldsymbol{s}^{\alpha_1\alpha_2\ldots\alpha_k\dot{\alpha}_1\dot{\alpha}_2\ldots
\dot{\alpha}_r}\longrightarrow\sum_{\left(\beta\right)\left(\dot{\beta}\right)}A^{\alpha_1\beta_1}A^{\alpha_2\beta_2}
\cdots A^{\alpha_k\beta_k}\overline{A}^{\dot{\alpha}_1\dot{\beta}_1}\overline{A}^{\dot{\alpha}_2\dot{\beta}_2}
\cdots\overline{A}^{\dot{\alpha}_r\dot{\beta}_r}\boldsymbol{s}^{\beta_1\beta_2\ldots\beta_k\dot{\beta}_1\dot{\beta}_2\ldots
\dot{\beta}_r},
\]
where the symbols $\left(\beta\right)$ and $\left(\dot{\beta}\right)$ mean $\beta_1$, $\beta_2$, $\ldots$, $\beta_k$ and $\dot{\beta}_1$, $\dot{\beta}_2$, $\ldots$, $\dot{\beta}_r$. This representation of $\SL(2,\C)$ we denote as $\boldsymbol{\tau}_{\frac{k}{2},\frac{r}{2}}=\boldsymbol{\tau}_{l\dot{l}}$. The each \emph{irreducible} finite dimensional representation of $\SL(2,\C)$ is equivalent to one from $\boldsymbol{\tau}_{k/2,r/2}$.

All the representations $\boldsymbol{\tau}_{l\dot{l}}$ can be grouped into spin multiplets in the Hilbert space $\bsH^S_{2s+1}\otimes\bsH_\infty$ (see Fig.\,6). $\bsH^S_{2s+1}\otimes\bsH_\infty$ is a subspace of the more general spin-charge Hilbert space $\bsH^S\otimes\bsH^Q\otimes\bsH_\infty$ \cite{Var1402}. The vertical lines (\emph{spin lines}) correspond to particles of the same spin (but different masses). The horizontal lines (\emph{spin chains} or spin multiplets) correspond to particles of the same mass (but different spins). Along the each spin chain the numbers $l$ and $\dot{l}$ are changed as
\[
l,\;l+\frac{1}{2},\;l+1,\;l+\frac{3}{2},\;\ldots,\;\dot{l},
\]
\[
\dot{l},\;\dot{l}-\frac{1}{2},\;\dot{l}-1,\;\dot{l}-\frac{3}{2},\;\ldots,\;l.
\]
Therefore, along the each spin chain we have the following representations:
\[
\boldsymbol{\tau}_{l\dot{l}},\;\boldsymbol{\tau}_{l+\frac{1}{2},\dot{l}-\frac{1}{2}},\;
\boldsymbol{\tau}_{l+1,\dot{l}-1},\;\boldsymbol{\tau}_{l+\frac{3}{2},\dot{l}-\frac{3}{2}},\;\ldots,\;
\boldsymbol{\tau}_{\dot{l}l},
\]
where the spin $s=l-\dot{l}$ is changed as
\[
l-\dot{l},\;l-\dot{l}+1,\;l-\dot{l}+2,\;l-\dot{l}+3,\,\ldots,\;\dot{l}-l.
\]
\begin{figure}[ht]
\[
\dgARROWPARTS=28
\dgARROWLENGTH=0.5em
\dgHORIZPAD=1.2em 
\dgVERTPAD=6.2ex 
\begin{diagram}
\node{\overset{(0,\frac{7}{2})}{\bullet}}\arrow[2]{e,..,-}\arrow[14]{s,l,20,..,-}{-\tfrac{7}{2}}
\node[2]{\overset{(\frac{1}{2},3)}{\bullet}}\arrow[2]{e,..,-}\arrow[2]{s,..,-}
\node[2]{\overset{(1,\frac{5}{2})}{\bullet}}\arrow[2]{e,..,-}\arrow[2]{s,..,-}
\node[2]{\overset{(\frac{3}{2},2)}{\bullet}}\arrow[2]{e,..,-}\arrow[2]{s,..,-}
\node[2]{\overset{(2,\frac{3}{2})}{\bullet}}\arrow[2]{e,..,-}\arrow[2]{s,..,-}
\node[2]{\overset{(\frac{5}{2},1)}{\bullet}}\arrow[2]{e,..,-}\arrow[2]{s,..,-}
\node[2]{\overset{(3,\frac{1}{2})}{\bullet}}\arrow[2]{e,..,-}\arrow[2]{s,..,-}
\node[2]{\overset{(\frac{7}{2},0)}{\bullet}}\arrow[14]{s,r,..,-}{\tfrac{7}{2}}\\
\node[2]{\overset{(0,3)}{\bullet}}\arrow[2]{e,..,-}\arrow[12]{s,l,..,-}{-3}\arrow[2]{n,..,-}
\node[2]{\overset{(\frac{1}{2},\frac{5}{2})}{\bullet}}\arrow[2]{e,..,-}\arrow[2]{s,..,-}\arrow[2]{n,..,-}
\node[2]{\overset{(1,2)}{\bullet}}\arrow[2]{e,..,-}\arrow[2]{s,..,-}\arrow[2]{n,..,-}
\node[2]{\overset{(\frac{3}{2},\frac{3}{2})}{\bullet}}\arrow[2]{s,-}\arrow[2]{e,..,-}\arrow[2]{n,-}
\node[2]{\overset{(2,1)}{\bullet}}\arrow[2]{e,..,-}\arrow[2]{s,..,-}\arrow[2]{n,..,-}
\node[2]{\overset{(\frac{5}{2},\frac{1}{2})}{\bullet}}\arrow[2]{e,..,-}\arrow[2]{s,..,-}\arrow[2]{n,..,-}
\node[2]{\overset{(3,0)}{\bullet}}\arrow[2]{n,..,-}\arrow[12]{s,r,..,-}{3}\\
\node[3]{\overset{(0,\frac{5}{2})}{\bullet}}\arrow[2]{e,..,-}\arrow[10]{s,l,..,-}{-\tfrac{5}{2}}
\node[2]{\overset{(\frac{1}{2},2)}{\bullet}}\arrow[2]{e,..,-}\arrow[2]{s,..,-}
\node[2]{\overset{(1,\frac{3}{2})}{\bullet}}\arrow[2]{e,..,-}\arrow[2]{s,..,-}
\node[2]{\overset{(\frac{3}{2},1)}{\bullet}}\arrow[2]{e,..,-}\arrow[2]{s,..,-}
\node[2]{\overset{(2,\frac{1}{2})}{\bullet}}\arrow[2]{e,..,-}\arrow[2]{s,..,-}
\node[2]{\overset{(\frac{5}{2},0)}{\bullet}}\arrow[10]{s,r,..,-}{\tfrac{5}{2}}\\
\node[4]{\overset{(0,2)}{\bullet}}\arrow[2]{e,..,-}\arrow[8]{s,l,..,-}{-2}
\node[2]{\overset{(\frac{1}{2},\frac{3}{2})}{\bullet}}\arrow[2]{e,..,-}\arrow[2]{s,..,-}
\node[2]{\overset{(1,1)}{\bullet}}\arrow[2]{s,-}\arrow[2]{e,..,-}
\node[2]{\overset{(\frac{3}{2},\frac{1}{2})}{\bullet}}\arrow[2]{e,..,-}\arrow[2]{s,..,-}
\node[2]{\overset{(2,0)}{\bullet}}\arrow[8]{s,r,..,-}{2}\\
\node[5]{\overset{(0,\frac{3}{2})}{\bullet}}\arrow[2]{e,..,-}\arrow[6]{s,l,..,-}{-\tfrac{3}{2}}
\node[2]{\overset{(\frac{1}{2},1)}{\bullet}}\arrow[2]{e,..,-}\arrow[2]{s,..,-}
\node[2]{\overset{(1,\frac{1}{2})}{\bullet}}\arrow[2]{e,..,-}\arrow[2]{s,..,-}
\node[2]{\overset{(\frac{3}{2},0)}{\bullet}}\arrow[6]{s,r,..,-}{\tfrac{3}{2}}\\
\node[6]{\overset{(0,1)}{\bullet}}\arrow[2]{e,..,-}\arrow[4]{s,l,..,-}{-1}
\node[2]{\overset{(\frac{1}{2},\frac{1}{2})}{\bullet}}\arrow[2]{s,-}\arrow[2]{e,..,-}
\node[2]{\overset{(1,0)}{\bullet}}\arrow[4]{s,r,..,-}{1}\\
\node[7]{\overset{(0,\frac{1}{2})}{\bullet}}\arrow[2]{e,..,-}\arrow[2]{s,l,..,-}{-\tfrac{1}{2}}
\node[2]{\overset{(\frac{1}{2},0)}{\bullet}}\arrow[2]{s,r,..,-}{\tfrac{1}{2}}\\
\node[8]{\overset{(0,0)}{\bullet}}\arrow[2]{s,-}\arrow[9]{w,-}\arrow[9]{e}\\
\node[7]{\overset{(0,\frac{1}{2})}{\bullet}}\arrow[2]{e,..,-}\arrow[2]{s,..,-}
\node[2]{\overset{(\frac{1}{2},0)}{\bullet}}\arrow[2]{s,..,-}\\
\node[6]{\overset{(0,1)}{\bullet}}\arrow[2]{e,..,-}\arrow[2]{s,..,-}
\node[2]{\overset{(\frac{1}{2},\frac{1}{2})}{\bullet}}\arrow[2]{s,-}\arrow[2]{e,..,-}
\node[2]{\overset{(1,0)}{\bullet}}\arrow[2]{s,..,-}\\
\node[5]{\overset{(0,\frac{3}{2})}{\bullet}}\arrow[2]{e,..,-}\arrow[2]{s,..,-}
\node[2]{\overset{(\frac{1}{2},1)}{\bullet}}\arrow[2]{e,..,-}\arrow[2]{s,..,-}
\node[2]{\overset{(1,\frac{1}{2})}{\bullet}}\arrow[2]{e,..,-}\arrow[2]{s,..,-}
\node[2]{\overset{(\frac{3}{2},0)}{\bullet}}\arrow[2]{s,..,-}\\
\node[4]{\overset{(0,2)}{\bullet}}\arrow[2]{e,..,-}\arrow[2]{s,..,-}
\node[2]{\overset{(\frac{1}{2},\frac{3}{2})}{\bullet}}\arrow[2]{e,..,-}\arrow[2]{s,..,-}
\node[2]{\overset{(1,1)}{\bullet}}\arrow[2]{s,-}\arrow[2]{e,..,-}
\node[2]{\overset{(\frac{3}{2},\frac{1}{2})}{\bullet}}\arrow[2]{e,..,-}\arrow[2]{s,..,-}
\node[2]{\overset{(2,0)}{\bullet}}\arrow[2]{s,..,-}\\
\node[3]{\overset{(0,\frac{5}{2})}{\bullet}}\arrow[2]{e,..,-}\arrow[2]{s,..,-}
\node[2]{\overset{(\frac{1}{2},2)}{\bullet}}\arrow[2]{e,..,-}\arrow[2]{s,..,-}
\node[2]{\overset{(1,\frac{3}{2})}{\bullet}}\arrow[2]{e,..,-}\arrow[2]{s,..,-}
\node[2]{\overset{(\frac{3}{2},1)}{\bullet}}\arrow[2]{e,..,-}\arrow[2]{s,..,-}
\node[2]{\overset{(2,\frac{1}{2})}{\bullet}}\arrow[2]{e,..,-}\arrow[2]{s,..,-}
\node[2]{\overset{(\frac{5}{2},0)}{\bullet}}\arrow[2]{s,..,-}\\
\node[2]{\overset{(0,3)}{\bullet}}\arrow[2]{e,..,-}\arrow[2]{s,..,-}
\node[2]{\overset{(\frac{1}{2},\frac{5}{2})}{\bullet}}\arrow[2]{e,..,-}\arrow[2]{s,..,-}
\node[2]{\overset{(1,2)}{\bullet}}\arrow[2]{e,..,-}\arrow[2]{s,..,-}
\node[2]{\overset{(\frac{3}{2},\frac{3}{2})}{\bullet}}\arrow[2]{s,-}\arrow[2]{e,..,-}
\node[2]{\overset{(2,1)}{\bullet}}\arrow[2]{e,..,-}\arrow[2]{s,..,-}
\node[2]{\overset{(\frac{5}{2},\frac{1}{2})}{\bullet}}\arrow[2]{e,..,-}\arrow[2]{s,..,-}
\node[2]{\overset{(3,0)}{\bullet}}\arrow[2]{s,..,-}\\
\node{\overset{(0,\frac{7}{2})}{\bullet}}\arrow[2]{e,..,-}
\node[2]{\overset{(\frac{1}{2},3)}{\bullet}}\arrow[2]{e,..,-}
\node[2]{\overset{(1,\frac{5}{2})}{\bullet}}\arrow[2]{e,..,-}
\node[2]{\overset{(\frac{3}{2},2)}{\bullet}}\arrow[2]{e,..,-}
\node[2]{\overset{(2,\frac{3}{2})}{\bullet}}\arrow[2]{e,..,-}
\node[2]{\overset{(\frac{5}{2},1)}{\bullet}}\arrow[2]{e,..,-}
\node[2]{\overset{(3,\frac{1}{2})}{\bullet}}\arrow[2]{e,..,-}
\node[2]{\overset{(\frac{7}{2},0)}{\bullet}}
\end{diagram}
\]
\begin{center}{\small {\bf Fig.\,6:} Matter and antimatter spin multiplets in $\bsH^S_{2s+1}\otimes\bsH_\infty$.}\end{center}
\end{figure}

For example, let us consider the following spin chain (7-plet):
\[
\begin{diagram}
\node{\overset{(0,3)}{\bullet}}\arrow{e,..,-}\node{\overset{(\frac{1}{2},\frac{5}{2})}{\bullet}}\arrow{e,..,-}
\node{\overset{(1,2)}{\bullet}}\arrow{e,..,-}\node{\overset{(\frac{3}{2},\frac{3}{2})}{\bullet}}\arrow{e,..,-}
\node{\overset{(2,1)}{\bullet}}\arrow{e,..,-}\node{\overset{(\frac{5}{2},\frac{1}{2})}{\bullet}}
\arrow{e,..,-}\node{\overset{(3,0)}{\bullet}}\\
\node{-3}\arrow{e,-}\node{-2}\arrow{e,-}\node{-1}\arrow{e,-}\node{0}\arrow{e,-}
\node{1}\arrow{e,-}\node{2}\arrow{e,-}\node{3}
\end{diagram}
\]
In the underlying spinor structure we have the following sequence of algebras associated with this 7-plet:
\begin{multline}
\overset{\ast}{\C}_2\otimes\overset{\ast}{\C}_2\otimes
\overset{\ast}{\C}_2\otimes\overset{\ast}{\C}_2\otimes\overset{\ast}{\C}_2
\otimes\overset{\ast}{\C}_2\;\longleftrightarrow\;\C_2\bigotimes\overset{\ast}{\C}_2\otimes\overset{\ast}{\C}_2\otimes
\overset{\ast}{\C}_2\otimes\overset{\ast}{\C}_2\otimes\overset{\ast}{\C}_2\;\longleftrightarrow\\
\C_2\otimes\C_2\bigotimes\overset{\ast}{\C}_2\otimes
\overset{\ast}{\C}_2\otimes\overset{\ast}{\C}_2\otimes\overset{\ast}{\C}_2\;\longleftrightarrow\;
\C_2\otimes\C_3\otimes\C_2\bigotimes\overset{\ast}{\C}_2\otimes\overset{\ast}{\C}_2\otimes\overset{\ast}{\C}_2
\;\longleftrightarrow\\
\C_2\otimes\C_2\otimes\C_2\otimes\C_2\bigotimes\overset{\ast}{\C}_2\otimes\overset{\ast}{\C}_2
\;\longleftrightarrow\;\C_2\otimes\C_2\otimes\C_2\otimes\C_2\otimes\C_2\bigotimes\overset{\ast}{\C}_2
\longleftrightarrow\;\C_2\otimes\C_2\otimes\C_2\otimes\C_2\otimes\C_2\otimes\C_2.\nonumber
\end{multline}
Wave equations for the fields of type $(l,0)\oplus(0,\dot{l})$ and their solutions in the form of series in hyperspherical functions were given in \cite{Var03a}-\cite{Var07}. It should be noted that $(l,0)\oplus(0,\dot{l})$ type wave equations correspond to the usual definition of the spin. In turn, wave equations for the fields of type $(l,\dot{l})\oplus(\dot{l},l)$ (arbitrary spin chains) and their solutions in the form of series in generalized hyperspherical functions were studied in \cite{Var07b}. Wave equations for arbitrary spin chains (spin multiplets) correspond to the generalized spin $s=|l-\dot{l}|$.
\subsection{Spinor and twistor structures}
The products (\ref{TenAlg}) and (\ref{SpinSpace}) define an \emph{algebraic (spinor) structure} associated with the representation $\boldsymbol{\tau}_{k/2,r/2}$ of the group $\SL(2,\C)$. Usually, spinor structures are understood as double (universal) coverings of the orthogonal groups $\SO(p,q)$. For that reason it seems that the spinor structure presents itself a derivative construction. However, in accordance with Penrose twistor programme \cite{Pen77,Pen68,Pen} the spinor (twistor) structure presents a more fundamental level of reality rather then  space-time continuum. Moreover, space-time continuum is generated by the twistor structure. This is a natural consequence of the well known fact of the van der Waerden 2-spinor formalism \cite{Wae32}, in which any vector of the Minkowski space-time can be constructed via the pair of mutually conjugated 2-spinors. For that reason it is more adequate to consider spinors as the \emph{underlying structure}. We choose $\spin_+(1,3)$ as a \emph{generating kernel} of the underlying spinor structure. In this context space-time discrete symmetries $P$, $T$ and their combination $PT$ should be considered as projections of the fundamental automorphisms belonging to the background spinor structure \cite{Var01,Var05c,Var05}. However, the group $\spin_+(2,4)\simeq\SU(2,2)$ (a universal covering of the conformal group $\SO_0(2,4)$) can be chosen as such a kernel. The choice $\spin_+(2,4)\simeq\SU(2,2)$ takes place in the Penrose twistor programme \cite{Pen} and also in the Paneitz-Segal approach \cite{PS1}--\cite{Lev11}.
\subsubsection{Spinors}
Let us consider in brief the basic facts concerning the theory of spinor
representations of the Lorentz group. The initial point of this theory
is a correspondence between transformations of the proper Lorentz group
and complex matrices of the second order. Indeed, following to
\cite{GMS} let us compare the Hermitian matrix of the second order
\begin{equation}\label{6.1}
X=\ar\begin{bmatrix}
x_0+x_3 & x_1-ix_2\\
x_1+ix_2 & x_0-x_3
\end{bmatrix}
\end{equation}
to the vector $v$ of the Minkowski space-time $\R^{1,3}$ with coordinates
$x_0,x_1,x_2,x_3$. At this point, $\det X=x^2_0-x^2_1-x^2_2-x^2_3=S^2(x)$.
The correspondence between matrices $X$ and vectors $v$ is one-to-one and
linear. Any linear transformation $X^\prime=aXa^\ast$ in a space of the
matrices $X$ may be considered as a linear transformation $g_a$ in
$\R^{1,3}$, where $a$ is a complex matrix of the second order with
$\det a=1$. The correspondence $a\sim g_a$ possesses following properties:
1) $\ar\begin{bmatrix}
1 & 0 \\ 0 & 1\end{bmatrix}\sim e$ (identity element); 2)
$g_{a_1}g_{a_2}=g_{a_1a_2}$ (composition); 3) two different matrices
$a_1$ and $a_2$ correspond to one and the same transformation $g_{a_1}=g_{a_2}$
only in the case $a_1=-a_2$. Since the each complex matrix is defined by
eight real numbers, then from the requirement $\det a=1$ it follow two
conditions $\re\det a=1$ and $\im\det a=0$. These conditions leave six
independent parameters, that coincides with parameter number of the
proper Lorentz group.

Further, a set of all complex matrices of the second order forms a full
matrix algebra $\Mat_2(\C)$ that is isomorphic to a biquaternion algebra
$\C_2$. In turn, Pauli matrices
\begin{equation}\label{6.2}
\sigma_0=\ar\begin{bmatrix}
1 & 0 \\
0 & 1
\end{bmatrix},\quad
\sigma_1=\begin{bmatrix}
0 & 1\\
1 & 0
\end{bmatrix},\quad
\sigma_2=\begin{bmatrix}
0 & -i\\
i & 0
\end{bmatrix},\quad
\sigma_3=\begin{bmatrix}
1 & 0\\
0 &-1
\end{bmatrix}
\end{equation}
form a spinbasis of the algebra $\C_2$
(by this reason in physics the algebra $\C_2\simeq\cl^+_{1,3}\simeq\cl_{3,0}$
is called Pauli algebra). Using the basis (\ref{6.2}), we can write the
matrix (\ref{6.1}) in the form
\begin{equation}\label{6.3}
X=x^\mu\sigma_\mu.
\end{equation}
The Hermitian matrix (\ref{6.3}) corresponds to a spintensor $X^{\lambda\dot{\nu}}$ with the following coordinates:
\begin{eqnarray}
x^0=+(1/\sqrt{2})(\xi^1\xi^{\dot{1}}+\xi^2\xi^{\dot{2}}),&&
x^1=+(1/\sqrt{2})(\xi^1\xi^{\dot{2}}+\xi^2\xi^{\dot{1}}),\nonumber\\
x^2=-(i/\sqrt{2})(\xi^1\xi^{\dot{2}}-\xi^2\xi^{\dot{1}}),&&
x^3=+(1/\sqrt{2})(\xi^1\xi^{\dot{1}}-\xi^2\xi^{\dot{2}}),\label{6.4}
\end{eqnarray}
where $\xi^\mu$ and $\xi^{\dot{\mu}}$ are correspondingly coordinates of
spinors and cospinors of spinspaces $\dS_2$ and $\dot{\dS}_2$. Linear
transformations of `vectors' (spinors and cospinors) of the spinspaces
$\dS_2$ and $\dot{\dS}_2$ have the form
\begin{equation}\label{6.5}\ar
\begin{array}{ccc}
\begin{array}{ccc}
{}^\prime\xi^1&=&\alpha\xi^1+\beta\xi^2,\\
{}^\prime\xi^2&=&\gamma\xi^1+\delta\xi^2,
\end{array} & \phantom{ccc} &
\begin{array}{ccc}
{}^\prime\xi^{\dot{1}}&=&\dot{\alpha}\xi^{\dot{1}}+\dot{\beta}\xi^{\dot{2}},\\
{}^\prime\xi^{\dot{2}}&=&\dot{\gamma}\xi^{\dot{1}}+\dot{\delta}\xi^{\dot{2}},
\end{array}\\
\sigma=\ar\begin{bmatrix}
\alpha & \beta\\
\gamma & \delta
\end{bmatrix} & \phantom{ccc} &
\dot{\sigma}=\begin{bmatrix}
\dot{\alpha} & \dot{\beta}\\
\dot{\gamma} & \dot{\delta}
\end{bmatrix}.
\end{array}
\end{equation}
Transformations (\ref{6.5}) form the group $\SL(2,\C)$, since
$\sigma\in\Mat_2(\C)$ and
\[
\SL(2,\C)=\left\{\ar\begin{bmatrix} \alpha & \beta\\ \gamma &\delta\end{bmatrix}
\in\C_2:\;\det\begin{bmatrix} \alpha & \beta \\ \gamma & \delta\end{bmatrix}=1
\right\}\simeq\spin_+(1,3).
\]
The expressions (\ref{6.4}) and (\ref{6.5}) compose a base of the 2-spinor
van der Waerden formalism \cite{Wae29,Rum36}, in which the spaces
$\dS_2$ and $\dot{\dS}_2$ are called correspondingly spaces of
{\it undotted and dotted spinors}. The each of the spaces
$\dS_2$ and $\dot{\dS}_2$ is homeomorphic to an extended complex plane
$\C\cup\infty$ representing an absolute (a set of infinitely distant points)
of a Lobatchevskii space $S^{1,2}$. At this point, a group of fractional
linear transformations of the plane $\C\cup\infty$ is isomorphic to a motion
group of $S^{1,2}$ \cite{Roz55}. Besides, the Lobatchevskii space $S^{1,2}$ is an absolute of the
Minkowski world $\R^{1,3}$ and, therefore, the group of fractional linear
transformations of the plane $\C\cup\infty$ (motion group of
$S^{1,2}$) twice covers a `rotation group' of the space-time $\R^{1,3}$,
that is, the proper Lorentz group.

\subsubsection{Twistors}
The main idea of the Penrose twistor programme lies in the understanding of classical space-time as a some secondary construction which should be derived from the more primary notions. In capacity of the more primary notions we have here 2-component (complex) spinors, moreover, the pairs of 2-component spinors. In Penrose programme they called \emph{twistors}. It is interesting to note that twistor theory gives a mathematical description of physics which based totally on the complex structure. At this point, space-time geometry and quantum mechanical superposition principle arise as closely related aspects of this complex twistor structure.

Twistor $\bsZ^\alpha$ is constructed by the pair of 2-component quantities: spinor $\boldsymbol{\omega}^s$ and covariant spinor $\boldsymbol{\pi}_{\dot{s}}$ from conjugated space, that is, $\bsZ^\alpha=\left(\boldsymbol{\omega}^s,\boldsymbol{\pi}_{\dot{s}}\right)$ (or $\bsZ^\alpha=\left(\boldsymbol{\xi}^\mu,\boldsymbol{\xi}_{\dot{\mu}}\right)$). In twistor theory  momentum ($\vec{\omega}$) and impulse ($\vec{\pi}$) of the particle are constructed from the quantities $\boldsymbol{\omega}^s$ and $\boldsymbol{\pi}_{\dot{s}}$. One of the most important moments of this theory is \emph{a transition from twistors to coordinate space-time}. Penrose described this transition with the help of so-called \emph{basic relation of twistor theory}
\begin{equation}\label{Twistor}
\boldsymbol{\omega}^s=ix^{s\dot{r}}\boldsymbol{\pi}_{\dot{s}},
\end{equation}
where $x^{s\dot{r}}$ is a mixed spintensor of the second rank (see (\ref{6.1})). In more details we have
\[
\begin{bmatrix}
\omega_1\\
\omega_2
\end{bmatrix}=\frac{i}{\sqrt{2}}\begin{bmatrix}
x^0+x^3 & x^1+ix^2\\
x^1+ix^2 & x^0-x^3
\end{bmatrix}\begin{bmatrix}
\pi_{\dot{1}}\\
\pi_{\dot{2}}
\end{bmatrix}.
\]
From the basic relation (\ref{Twistor}) it immediately follows that space-time points are re-established over twistor space (they correspond to definite linear subspaces), but these points are secondary notion with respect to twistors.

In fact, twistors can be considered as `reduced spinors'\footnote{These reduced spinors are understood as follows. General spinors are elements of the minimal left ideal of the conformal algebra $\cl_{2,4}$,
\[
I_{2,4}=\cl_{2,4}f_{24}=\cl_{2,4}\frac{1}{2}(1+\e_{15})\frac{1}{2}(1+\e_{26}).
\]
The reduced spinors (twistors) are formulated within the even subalgebra $\cl^+_{2,4}\simeq\cl_{4,1}$ (the de Sitter algebra). The minimal left ideal of $\cl_{4,1}\simeq\C_4$ is (see also (\ref{IdealS}))
\[
I_{4,1}=\cl_{4,1}f_{4,1}=\cl_{4,1}\frac{1}{2}(1+\e_0)\frac{1}{2}
(1+i\e_{12}).
\]
Therefore, after reduction $I_{2,4}\rightarrow I_{4,1}$, generated by the isomorphism $\cl^+_{2,4}\simeq\cl_{4,1}$, we see that twistors $\bsZ^\alpha$ are elements of the ideal $I_{4,1}$ which leads to $\SU(2,2)\simeq\spin_+(2,4)\in\cl^+_{2,4}$ (see (\ref{ConfGroup0}) and (\ref{ConfGroup})). Moreover, from (\ref{2e1}) we have a relation between twistors and Dirac-Hestenes spinors.} for a pseudo-unitary group $\SO_0(2,4)$ which acts in six-dimensional space. This group is isomorphic locally to a 15-parameter conformal group of the Minkowski space $\R^{1,3}$ (the group of point-to-point mappings of $\R^{1,3}$ onto itself with preservation of the conformal structure of this space). Such mappings induce linear transformations of the twistor space which preserve the form $\bsZ^\alpha\overline{\bsZ}_\alpha$. The signature of $\bsZ^\alpha\overline{\bsZ}_\alpha$ has the form $(+,+,-,-)$, it means that the corresponding group in the twistor space is $\SU(2,2)$ (the group of pseudo-unitary $(+,+,-,-)$ unimodular $4\times 4$ matrices, see also (\ref{ConfGroup})):
\[
\SU(2,2)=\left\{\ar\begin{bmatrix} A & B\\ C & D\end{bmatrix}
\in\C_4:\;\det\begin{bmatrix} A & B \\ C & D\end{bmatrix}=1
\right\}\simeq\spin_+(2,4).
\]
\subsubsection{Qubits}
As is known, the qubit is a state vector of the two-level system. Thus, the qubit is a minimally possible (elementary) state vector. Any state vector can be represented as a set of such elementary vectors, for that reason the qubit is an original `building block' for the all other state vectors of any dimension. The vector (qubit) of the two-level system can be written in the form
\begin{equation}\label{1}
\left|\boldsymbol{\psi}\right\rangle=a\left|\boldsymbol{0}\right\rangle+b\left|\boldsymbol{1}\right\rangle,
\end{equation}
where $a,b\in\C$. A space of the two states, when the system can transits from one state to another (two-level system), is a simplest Hilbert space. The quantum state of $N$ qubits can be expressed as a vector in a space of dimension $2^N$. It is obvious that this space coincides with the spinspace $\dS_{2^N}$. We can choose as an orthonormal basis for this space the states in which each qubit has a definite value, either $\left|\boldsymbol{0}\right\rangle$ or $\left|\boldsymbol{1}\right\rangle$. These can be labeled by binary strings such as
\[
\left|\boldsymbol{0}\boldsymbol{1}\boldsymbol{1}\boldsymbol{1}\boldsymbol{0}\boldsymbol{0}\boldsymbol{1}
\boldsymbol{0}\,\cdots\,\boldsymbol{1}\boldsymbol{0}\boldsymbol{0}\boldsymbol{1}\right\rangle.
\]
A general normalized vector can be expressed in this basis as $\sum^{2^N-1}_{x=0}a_x\left|x\right\rangle$, where $a_x$ are complex numbers satisfying $\sum_x|a_x|^2=1$. Here we have a deep analogy between qubits and 2-component spinors. Just like the qubits, 2-component spinors are `building blocks' of the underlying spinor structure (via the tensor products of $\C_2$ and $\overset{\ast}{\C}_2$, see (\ref{TenAlg}) and (\ref{SpinSpace})). Moreover, vectors of the Hilbert space $\bsH^S\otimes\bsH^Q\otimes\bsH_\infty$ are constructed via the same way \cite{Var1402}, see also \cite{Wang}.

The density matrix of the qubit has $2\times 2$ size and for pure state (\ref{1}) can be written as
\begin{equation}\label{2}
\boldsymbol{r}=\left|\boldsymbol{\psi}\right\rangle\left\langle\boldsymbol{\psi}\right|=\begin{bmatrix} |a|^2 & ab^\ast\\
ba^\ast & |b|^2
\end{bmatrix}.
\end{equation}

There exists a more general expression for the density matrix of the qubit which includes both pure and mixed states:
\begin{equation}\label{3}
\boldsymbol{r}=\frac{1}{2}(\sigma_0+\boldsymbol{P}\cdot\boldsymbol{\sigma})=\frac{1}{2}\begin{bmatrix} 1+P_3 & P_1-iP_2\\
P_1+iP_2 & 1-P_3
\end{bmatrix},
\end{equation}
where $\boldsymbol{P}=(P_1,P_2,P_3)$ is a Bloch vector (polarization vector). Components of the Bloch vector are defined as average values of the Pauli matrices via the rule $P_j=\langle\sigma_j\rangle=\Tr(P_j\sigma_j)$, $j=1,2,3$. In accordance with (\ref{3}), three projections $P_1$, $P_2$, $P_3$ of the polarization vector define the density matrix of the qubit. In the case of pure state the length of $\boldsymbol{P}$ is equal to 1 ($|\boldsymbol{P}|^2=1$) and this vector describes a sphere of the unit radius which called a \emph{Bloch sphere} (see Fig.\,7)
\begin{figure}[ht]
\[
\unitlength=0.3mm
\begin{picture}(100.00,100.00)(0,25)
\put(50,50){\vector(0,1){51}}
\put(50,50){\line(0,-1){2}}
\put(50,46){\line(0,-1){2}}
\put(50,42){\line(0,-1){2}}
\put(50,38){\line(0,-1){2}}
\put(50,34){\line(0,-1){2}}
\put(50,30){\line(0,-1){2}}
\put(50,26){\line(0,-1){2}}
\put(50,22){\line(0,-1){2}}
\put(50,18){\line(0,-1){2}}
\put(50,14){\line(0,-1){2}}
\put(50,10){\line(0,-1){2}}
\put(50,6){\vector(0,-1){2}}
\put(50,51.5){\vector(1,0){70}}
\put(50,51.5){\vector(-2,-1){65}}
\put(50,51.5){\vector(2,1){38}}
\put(-30,30){$\frac{\left|\boldsymbol{0}\right\rangle+\left|\boldsymbol{1}\right\rangle}{\sqrt{2}}$}
\put(50,100){\vector(0,1){20}}
\put(53,110){$\left|\boldsymbol{0}\right\rangle$}
\put(53,85){$\left|\boldsymbol{\Psi}\right\rangle$}
\put(50,2){\vector(0,-1){20}}
\put(53,-10){$\left|\boldsymbol{1}\right\rangle$}
\put(46.5,49.5){$\bullet$}
\put(76,72){$\boldsymbol{P}$}
\put(86,67.75){$\bullet$}
\put(50,34){$\cdot$}
\put(49.5,34){$\cdot$}
\put(49,34){$\cdot$}
\put(48.5,34){$\cdot$}
\put(48,34.01){$\cdot$}
\put(47.5,34.02){$\cdot$}
\put(47,34.03){$\cdot$}
\put(46.5,34.04){$\cdot$}
\put(46,34.05){$\cdot$}
\put(45.5,34.06){$\cdot$}
\put(45,34.08){$\cdot$}
\put(44.5,34.10){$\cdot$}
\put(44,34.11){$\cdot$}
\put(43.5,34.13){$\cdot$}
\put(43,34.16){$\cdot$}
\put(42.5,34.18){$\cdot$}
\put(42,34.20){$\cdot$}
\put(41.5,34.23){$\cdot$}
\put(41,34.26){$\cdot$}
\put(40.5,34.29){$\cdot$}
\put(40,34.32){$\cdot$}
\put(39.5,34.36){$\cdot$}
\put(39,34.39){$\cdot$}
\put(38.5,34.43){$\cdot$}
\put(38,34.47){$\cdot$}
\put(37.5,34.51){$\cdot$}
\put(37,34.55){$\cdot$}
\put(36.5,34.59){$\cdot$}
\put(36,34.64){$\cdot$}
\put(35.5,34.69){$\cdot$}
\put(35,34.74){$\cdot$}
\put(34.5,34.79){$\cdot$}
\put(34,34.84){$\cdot$}
\put(33.5,34.90){$\cdot$}
\put(33,34.95){$\cdot$}
\put(32.5,35.01){$\cdot$}
\put(32,35.07){$\cdot$}
\put(31.5,35.13){$\cdot$}
\put(31,35.20){$\cdot$}
\put(30.5,35.27){$\cdot$}
\put(30,35.33){$\cdot$}
\put(29.5,35.40){$\cdot$}
\put(29,35.48){$\cdot$}
\put(28.5,35.55){$\cdot$}
\put(28,35.63){$\cdot$}
\put(27.5,35.71){$\cdot$}
\put(27,35.79){$\cdot$}
\put(26.5,35.88){$\cdot$}
\put(26,35.96){$\cdot$}
\put(25.5,36.05){$\cdot$}
\put(25,36.14){$\cdot$}
\put(24.5,36.24){$\cdot$}
\put(24,36.33){$\cdot$}
\put(23.5,36.43){$\cdot$}
\put(23,36.53){$\cdot$}
\put(22.5,36.64){$\cdot$}
\put(22,36.74){$\cdot$}
\put(21.5,36.85){$\cdot$}
\put(21,36.97){$\cdot$}
\put(20.5,37.08){$\cdot$}
\put(20,37.20){$\cdot$}
\put(19.5,37.32){$\cdot$}
\put(19,37.45){$\cdot$}
\put(18.5,37.57){$\cdot$}
\put(18,37.70){$\cdot$}
\put(17.5,37.84){$\cdot$}
\put(17,37.98){$\cdot$}
\put(16.5,38.12){$\cdot$}
\put(16,38.27){$\cdot$}
\put(15.5,38.42){$\cdot$}
\put(15,38.57){$\cdot$}
\put(14.5,38.73){$\cdot$}
\put(14,38.90){$\cdot$}
\put(13.5,39.06){$\cdot$}
\put(13,39.24){$\cdot$}
\put(12.5,39.42){$\cdot$}
\put(12,39.60){$\cdot$}
\put(11.5,39.79){$\cdot$}
\put(11,39.99){$\cdot$}
\put(10.5,40.19){$\cdot$}
\put(10,40.40){$\cdot$}
\put(9.5,40.62){$\cdot$}
\put(9,40.84){$\cdot$}
\put(8.5,41.07){$\cdot$}
\put(8,41.32){$\cdot$}
\put(7.5,41.57){$\cdot$}
\put(7,41.83){$\cdot$}
\put(6.5,42.11){$\cdot$}
\put(6,42.40){$\cdot$}
\put(5.5,42.70){$\cdot$}
\put(5,43.02){$\cdot$}
\put(4.5,43.36){$\cdot$}
\put(4,43.73){$\cdot$}
\put(3.5,44.12){$\cdot$}
\put(3,44.54){$\cdot$}
\put(2.5,45.00){$\cdot$}
\put(2,45.52){$\cdot$}
\put(1.5,46.11){$\cdot$}
\put(1,46.82){$\cdot$}
\put(0.5,47.74){$\cdot$}
\put(0,50){$\cdot$}
\put(50,34){$\cdot$}
\put(50.5,34){$\cdot$}
\put(51,34){$\cdot$}
\put(51.5,34){$\cdot$}
\put(52,34.01){$\cdot$}
\put(52.5,34.02){$\cdot$}
\put(53,34.03){$\cdot$}
\put(53.5,34.04){$\cdot$}
\put(54,34.05){$\cdot$}
\put(54.5,34.06){$\cdot$}
\put(55,34.08){$\cdot$}
\put(55.5,34.10){$\cdot$}
\put(56,34.11){$\cdot$}
\put(56.5,34.13){$\cdot$}
\put(57,34.16){$\cdot$}
\put(57.5,34.18){$\cdot$}
\put(58,34.20){$\cdot$}
\put(58.5,34.23){$\cdot$}
\put(59,34.26){$\cdot$}
\put(59.5,34.29){$\cdot$}
\put(60,34.32){$\cdot$}
\put(60.5,34.36){$\cdot$}
\put(61,34.39){$\cdot$}
\put(61.5,34.43){$\cdot$}
\put(62,34.47){$\cdot$}
\put(62.5,34.51){$\cdot$}
\put(63,34.55){$\cdot$}
\put(63.5,34.59){$\cdot$}
\put(64,34.64){$\cdot$}
\put(64.5,34.69){$\cdot$}
\put(65,34.74){$\cdot$}
\put(65.5,34.79){$\cdot$}
\put(66,34.84){$\cdot$}
\put(66.5,34.90){$\cdot$}
\put(67,34.95){$\cdot$}
\put(67.5,35.01){$\cdot$}
\put(68,35.07){$\cdot$}
\put(68.5,35.13){$\cdot$}
\put(69,35.20){$\cdot$}
\put(69.5,35.27){$\cdot$}
\put(70,35.33){$\cdot$}
\put(70.5,35.40){$\cdot$}
\put(71,35.48){$\cdot$}
\put(71.5,35.55){$\cdot$}
\put(72,35.63){$\cdot$}
\put(72.5,35.71){$\cdot$}
\put(73,35.79){$\cdot$}
\put(73.5,35.88){$\cdot$}
\put(74,35.96){$\cdot$}
\put(74.5,36.05){$\cdot$}
\put(75,36.14){$\cdot$}
\put(75.5,36.24){$\cdot$}
\put(76,36.33){$\cdot$}
\put(76.5,36.43){$\cdot$}
\put(77,36.53){$\cdot$}
\put(77.5,36.64){$\cdot$}
\put(78,36.74){$\cdot$}
\put(78.5,36.85){$\cdot$}
\put(79,36.97){$\cdot$}
\put(79.5,37.08){$\cdot$}
\put(80,37.20){$\cdot$}
\put(80.5,37.32){$\cdot$}
\put(81,37.45){$\cdot$}
\put(81.5,37.57){$\cdot$}
\put(82,37.70){$\cdot$}
\put(82.5,37.84){$\cdot$}
\put(83,37.98){$\cdot$}
\put(83.5,38.12){$\cdot$}
\put(84,38.27){$\cdot$}
\put(84.5,38.42){$\cdot$}
\put(85,38.57){$\cdot$}
\put(85.5,38.73){$\cdot$}
\put(86,38.90){$\cdot$}
\put(86.5,39.06){$\cdot$}
\put(87,39.24){$\cdot$}
\put(87.5,39.42){$\cdot$}
\put(88,39.60){$\cdot$}
\put(88.5,39.79){$\cdot$}
\put(89,39.99){$\cdot$}
\put(89.5,40.19){$\cdot$}
\put(90,40.40){$\cdot$}
\put(90.5,40.62){$\cdot$}
\put(91,40.84){$\cdot$}
\put(91.5,41.07){$\cdot$}
\put(92,41.32){$\cdot$}
\put(92.5,41.57){$\cdot$}
\put(93,41.83){$\cdot$}
\put(93.5,42.11){$\cdot$}
\put(94,42.40){$\cdot$}
\put(94.5,42.70){$\cdot$}
\put(95,43.02){$\cdot$}
\put(95.5,43.36){$\cdot$}
\put(96,43.73){$\cdot$}
\put(96.5,44.12){$\cdot$}
\put(97,44.54){$\cdot$}
\put(97.5,45.00){$\cdot$}
\put(98,45.52){$\cdot$}
\put(98.5,46.11){$\cdot$}
\put(99,46.82){$\cdot$}
\put(99.5,47.74){$\cdot$}
\put(100,50){$\cdot$}

\put(50,66){$\cdot$}
\put(49.5,65.99){$\cdot$}
\put(49,65.99){$\cdot$}
\put(48.5,65.99){$\cdot$}
\put(48,65.98){$\cdot$}
\put(45,65.91){$\cdot$}
\put(44.5,65.90){$\cdot$}
\put(44,65.88){$\cdot$}
\put(43.5,65.86){$\cdot$}
\put(43,65.84){$\cdot$}
\put(40,65.68){$\cdot$}
\put(39.5,65.64){$\cdot$}
\put(39,65.60){$\cdot$}
\put(38.5,65.57){$\cdot$}
\put(38,65.53){$\cdot$}
\put(35,65.26){$\cdot$}
\put(34.5,65.21){$\cdot$}
\put(34,65.16){$\cdot$}
\put(33.5,65.10){$\cdot$}
\put(33,65.04){$\cdot$}
\put(30,64.66){$\cdot$}
\put(29.5,64.59){$\cdot$}
\put(29,64.52){$\cdot$}
\put(28.5,64.44){$\cdot$}
\put(28,64.37){$\cdot$}
\put(25,63.86){$\cdot$}
\put(24.5,63.76){$\cdot$}
\put(24,63.66){$\cdot$}
\put(23.5,63.57){$\cdot$}
\put(23,63.47){$\cdot$}
\put(20,62.80){$\cdot$}
\put(19.5,62.68){$\cdot$}
\put(19,62.55){$\cdot$}
\put(18.5,62.42){$\cdot$}
\put(18,62.29){$\cdot$}
\put(15,61.43){$\cdot$}
\put(14.5,61.27){$\cdot$}
\put(14,61.10){$\cdot$}
\put(13.5,60.93){$\cdot$}
\put(13,60.76){$\cdot$}
\put(10,59.60){$\cdot$}
\put(9.5,59.38){$\cdot$}
\put(9,59.16){$\cdot$}
\put(8.5,58.92){$\cdot$}
\put(8,58.68){$\cdot$}
\put(5,56.97){$\cdot$}
\put(4.5,56.63){$\cdot$}
\put(4,56.27){$\cdot$}
\put(3.5,55.88){$\cdot$}
\put(3,55.46){$\cdot$}
\put(53,65.98){$\cdot$}
\put(53.5,65.97){$\cdot$}
\put(54,65.96){$\cdot$}
\put(54.5,65.95){$\cdot$}
\put(55,65.93){$\cdot$}
\put(58,65.82){$\cdot$}
\put(58.5,65.79){$\cdot$}
\put(59,65.77){$\cdot$}
\put(59.5,65.74){$\cdot$}
\put(60,65.70){$\cdot$}
\put(63,65.49){$\cdot$}
\put(63.5,65.45){$\cdot$}
\put(64,65.40){$\cdot$}
\put(64.5,65.36){$\cdot$}
\put(65,65.31){$\cdot$}
\put(68,64.99){$\cdot$}
\put(68.5,64.93){$\cdot$}
\put(69,64.86){$\cdot$}
\put(69.5,64.79){$\cdot$}
\put(70,64.73){$\cdot$}
\put(73,64.29){$\cdot$}
\put(73.5,64.20){$\cdot$}
\put(74,64.12){$\cdot$}
\put(74.5,64.04){$\cdot$}
\put(75,63.95){$\cdot$}
\put(78,63.36){$\cdot$}
\put(78.5,63.25){$\cdot$}
\put(79,63.15){$\cdot$}
\put(79.5,63.03){$\cdot$}
\put(80,62.92){$\cdot$}
\put(83,62.16){$\cdot$}
\put(83.5,62.02){$\cdot$}
\put(84,61.88){$\cdot$}
\put(84.5,61.73){$\cdot$}
\put(85,61.58){$\cdot$}
\put(88,60.58){$\cdot$}
\put(88.5,60.40){$\cdot$}
\put(89,60.20){$\cdot$}
\put(89.5,60.01){$\cdot$}
\put(90,59.81){$\cdot$}
\put(93,58.43){$\cdot$}
\put(93.5,58.16){$\cdot$}
\put(94,57.88){$\cdot$}
\put(94.5,57.59){$\cdot$}
\put(95,57.29){$\cdot$}
\put(98,54.99){$\cdot$}
\put(98.5,54.48){$\cdot$}
\put(99,53.89){$\cdot$}
\put(99.5,53.18){$\cdot$}
\put(100,52.26){$\cdot$}

\put(50,100){$\cdot$}
\put(49.5,99.99){$\cdot$}
\put(49,99.98){$\cdot$}
\put(48.5,99.97){$\cdot$}
\put(48,99.96){$\cdot$}
\put(47.5,99.94){$\cdot$}
\put(47,99.90){$\cdot$}
\put(46.5,99.88){$\cdot$}
\put(46,99.84){$\cdot$}
\put(45.5,99.80){$\cdot$}
\put(45,99.75){$\cdot$}
\put(44.5,99.70){$\cdot$}
\put(44,99.64){$\cdot$}
\put(43.5,99.58){$\cdot$}
\put(43,99.51){$\cdot$}
\put(42.5,99.43){$\cdot$}
\put(42,99.35){$\cdot$}
\put(41.5,99.27){$\cdot$}
\put(41,99.18){$\cdot$}
\put(40.5,99.09){$\cdot$}
\put(40,98.99){$\cdot$}
\put(39.5,98.88){$\cdot$}
\put(39,98.77){$\cdot$}
\put(38.5,98.66){$\cdot$}
\put(38,98.54){$\cdot$}
\put(37.5,98.41){$\cdot$}
\put(37,98.28){$\cdot$}
\put(36.5,98.14){$\cdot$}
\put(36,98.00){$\cdot$}
\put(35.5,97.85){$\cdot$}
\put(35,97.70){$\cdot$}
\put(34.5,97.54){$\cdot$}
\put(34,97.37){$\cdot$}
\put(33.5,97.20){$\cdot$}
\put(33,97.02){$\cdot$}
\put(32.5,96.84){$\cdot$}
\put(32,96.65){$\cdot$}
\put(31.5,96.45){$\cdot$}
\put(31,96.25){$\cdot$}
\put(30.5,96.04){$\cdot$}
\put(30,95.82){$\cdot$}
\put(29.5,95.60){$\cdot$}
\put(29,95.38){$\cdot$}
\put(28.5,95.14){$\cdot$}
\put(28,94.90){$\cdot$}
\put(27.5,94.65){$\cdot$}
\put(27,94.39){$\cdot$}
\put(26.5,94.13){$\cdot$}
\put(26,93.86){$\cdot$}
\put(25.5,93.59){$\cdot$}
\put(25,93.30){$\cdot$}
\put(24.5,93.01){$\cdot$}
\put(24,92.71){$\cdot$}
\put(23.5,92.40){$\cdot$}
\put(23,92.08){$\cdot$}
\put(22.5,91.76){$\cdot$}
\put(22,91.42){$\cdot$}
\put(21.5,91.08){$\cdot$}
\put(21,90.73){$\cdot$}
\put(20.5,90.37){$\cdot$}
\put(20,90.00){$\cdot$}
\put(19.5,89.62){$\cdot$}
\put(19,89.23){$\cdot$}
\put(18.5,88.83){$\cdot$}
\put(18,88.42){$\cdot$}
\put(17.5,87.99){$\cdot$}
\put(17,87.56){$\cdot$}
\put(16.5,87.12){$\cdot$}
\put(16,86.66){$\cdot$}
\put(15.5,86.19){$\cdot$}
\put(15,85.70){$\cdot$}
\put(14.5,85.21){$\cdot$}
\put(14,84.70){$\cdot$}
\put(13.5,84.17){$\cdot$}
\put(13,83.63){$\cdot$}
\put(12.5,83.07){$\cdot$}
\put(12,82.49){$\cdot$}
\put(11.5,81.90){$\cdot$}
\put(11,81.29){$\cdot$}
\put(10.5,80.65){$\cdot$}
\put(10,80.00){$\cdot$}
\put(9.5,79.32){$\cdot$}
\put(9,78.62){$\cdot$}
\put(8.5,77.88){$\cdot$}
\put(8,77.13){$\cdot$}
\put(7.5,76.34){$\cdot$}
\put(7,75.51){$\cdot$}
\put(6.5,74.65){$\cdot$}
\put(6,73.75){$\cdot$}
\put(5.5,72.80){$\cdot$}
\put(5,71.79){$\cdot$}
\put(4.5,70.73){$\cdot$}
\put(4,69.59){$\cdot$}
\put(3.5,68.38){$\cdot$}
\put(3,67.06){$\cdot$}
\put(2.5,65.61){$\cdot$}
\put(2,64.00){$\cdot$}
\put(1.5,62.15){$\cdot$}
\put(1,59.95){$\cdot$}
\put(0.5,57.05){$\cdot$}
\put(0,50){$\cdot$}
\put(50,100){$\cdot$}
\put(50.5,99.99){$\cdot$}
\put(51,99.98){$\cdot$}
\put(51.5,99.97){$\cdot$}
\put(52,99.96){$\cdot$}
\put(52.5,99.94){$\cdot$}
\put(53,99.90){$\cdot$}
\put(53.5,99.88){$\cdot$}
\put(54,99.84){$\cdot$}
\put(54.5,99.80){$\cdot$}
\put(55,99.75){$\cdot$}
\put(55.5,99.70){$\cdot$}
\put(56,99.64){$\cdot$}
\put(56.5,99.58){$\cdot$}
\put(57,99.51){$\cdot$}
\put(57.5,99.43){$\cdot$}
\put(58,99.35){$\cdot$}
\put(58.5,99.27){$\cdot$}
\put(59,99.18){$\cdot$}
\put(59.5,99.09){$\cdot$}
\put(60,98.99){$\cdot$}
\put(60.5,98.88){$\cdot$}
\put(61,98.77){$\cdot$}
\put(61.5,98.66){$\cdot$}
\put(62,98.54){$\cdot$}
\put(62.5,98.41){$\cdot$}
\put(63,98.28){$\cdot$}
\put(63.5,98.14){$\cdot$}
\put(64,98.00){$\cdot$}
\put(64.5,97.85){$\cdot$}
\put(65,97.70){$\cdot$}
\put(65.5,97.54){$\cdot$}
\put(66,97.37){$\cdot$}
\put(66.5,97.20){$\cdot$}
\put(67,97.02){$\cdot$}
\put(67.5,96.84){$\cdot$}
\put(68,96.65){$\cdot$}
\put(68.5,96.45){$\cdot$}
\put(69,96.25){$\cdot$}
\put(69.5,96.04){$\cdot$}
\put(70,95.82){$\cdot$}
\put(70.5,95.60){$\cdot$}
\put(71,95.38){$\cdot$}
\put(71.5,95.14){$\cdot$}
\put(72,94.90){$\cdot$}
\put(72.5,94.65){$\cdot$}
\put(73,94.39){$\cdot$}
\put(73.5,94.13){$\cdot$}
\put(74,93.86){$\cdot$}
\put(74.5,93.59){$\cdot$}
\put(75,93.30){$\cdot$}
\put(75.5,93.01){$\cdot$}
\put(76,92.71){$\cdot$}
\put(76.5,92.40){$\cdot$}
\put(77,92.08){$\cdot$}
\put(77.5,91.76){$\cdot$}
\put(78,91.42){$\cdot$}
\put(78.5,91.08){$\cdot$}
\put(79,90.73){$\cdot$}
\put(79.5,90.37){$\cdot$}
\put(80,90.00){$\cdot$}
\put(80.5,89.62){$\cdot$}
\put(81,89.23){$\cdot$}
\put(81.5,88.83){$\cdot$}
\put(82,88.42){$\cdot$}
\put(82.5,87.99){$\cdot$}
\put(83,87.56){$\cdot$}
\put(83.5,87.12){$\cdot$}
\put(84,86.66){$\cdot$}
\put(84.5,86.19){$\cdot$}
\put(85,85.70){$\cdot$}
\put(85.5,85.21){$\cdot$}
\put(86,84.70){$\cdot$}
\put(86.5,84.17){$\cdot$}
\put(87,83.63){$\cdot$}
\put(87.5,83.07){$\cdot$}
\put(88,82.49){$\cdot$}
\put(88.5,81.90){$\cdot$}
\put(89,81.29){$\cdot$}
\put(89.5,80.65){$\cdot$}
\put(90,80.00){$\cdot$}
\put(90.5,79.32){$\cdot$}
\put(91,78.62){$\cdot$}
\put(91.5,77.88){$\cdot$}
\put(92,77.13){$\cdot$}
\put(92.5,76.34){$\cdot$}
\put(93,75.51){$\cdot$}
\put(93.5,74.65){$\cdot$}
\put(94,73.75){$\cdot$}
\put(94.5,72.80){$\cdot$}
\put(95,71.79){$\cdot$}
\put(95.5,70.73){$\cdot$}
\put(96,69.59){$\cdot$}
\put(96.5,68.38){$\cdot$}
\put(97,67.06){$\cdot$}
\put(97.5,65.61){$\cdot$}
\put(98,64.00){$\cdot$}
\put(98.5,62.15){$\cdot$}
\put(99,59.95){$\cdot$}
\put(99.5,57.05){$\cdot$}
\put(100,50){$\cdot$}

\put(50,0){$\cdot$}
\put(49.5,0){$\cdot$}
\put(49,0.01){$\cdot$}
\put(48.5,0.02){$\cdot$}
\put(48,0.04){$\cdot$}
\put(47.5,0.06){$\cdot$}
\put(47,0.09){$\cdot$}
\put(46.5,0.12){$\cdot$}
\put(46,0.16){$\cdot$}
\put(45.5,0.2){$\cdot$}
\put(45,0.25){$\cdot$}
\put(44.5,0.3){$\cdot$}
\put(44,0.36){$\cdot$}
\put(43.5,0.42){$\cdot$}
\put(43,0.49){$\cdot$}
\put(42.5,0.56){$\cdot$}
\put(42,0.64){$\cdot$}
\put(41.5,0.73){$\cdot$}
\put(41,0.82){$\cdot$}
\put(40.5,0.91){$\cdot$}
\put(40,1.01){$\cdot$}
\put(39.5,1.11){$\cdot$}
\put(39,1.22){$\cdot$}
\put(38.5,1.34){$\cdot$}
\put(38,1.46){$\cdot$}
\put(37.5,1.59){$\cdot$}
\put(37,1.72){$\cdot$}
\put(36.5,1.86){$\cdot$}
\put(36,2.0){$\cdot$}
\put(35.5,2.15){$\cdot$}
\put(35,2.3){$\cdot$}
\put(34.5,2.46){$\cdot$}
\put(34,2.63){$\cdot$}
\put(33.5,2.8){$\cdot$}
\put(33,2.98){$\cdot$}
\put(32.5,3.16){$\cdot$}
\put(32,3.35){$\cdot$}
\put(31.5,3.55){$\cdot$}
\put(31,3.75){$\cdot$}
\put(30.5,3.96){$\cdot$}
\put(30,4.17){$\cdot$}
\put(29.5,4.39){$\cdot$}
\put(29,4.62){$\cdot$}
\put(28.5,4.86){$\cdot$}
\put(28,5.1){$\cdot$}
\put(27.5,5.35){$\cdot$}
\put(27,5.6){$\cdot$}
\put(26.5,5.87){$\cdot$}
\put(26,6.14){$\cdot$}
\put(25.5,6.41){$\cdot$}
\put(25,6.7){$\cdot$}
\put(24.5,6.99){$\cdot$}
\put(24,7.29){$\cdot$}
\put(23.5,7.6){$\cdot$}
\put(23,7.9){$\cdot$}
\put(22.5,8.24){$\cdot$}
\put(22,8.57){$\cdot$}
\put(21.5,8.92){$\cdot$}
\put(21,9.27){$\cdot$}
\put(20.5,9.63){$\cdot$}
\put(20,10.0){$\cdot$}
\put(19.5,10.38){$\cdot$}
\put(19,10.77){$\cdot$}
\put(18.5,11.17){$\cdot$}
\put(18,11.58){$\cdot$}
\put(17.5,12.0){$\cdot$}
\put(17,12.44){$\cdot$}
\put(16.5,12.88){$\cdot$}
\put(16,13.34){$\cdot$}
\put(15.5,13.81){$\cdot$}
\put(15,14.29){$\cdot$}
\put(14.5,14.79){$\cdot$}
\put(14,15.3){$\cdot$}
\put(13.5,15.83){$\cdot$}
\put(13,16.37){$\cdot$}
\put(12.5,16.93){$\cdot$}
\put(12,17.5){$\cdot$}
\put(11.5,18.1){$\cdot$}
\put(11,18.71){$\cdot$}
\put(10.5,19.34){$\cdot$}
\put(10,20.00){$\cdot$}
\put(9.5,20.68){$\cdot$}
\put(9,21.38){$\cdot$}
\put(8.5,22.11){$\cdot$}
\put(8,22.87){$\cdot$}
\put(7.5,23.66){$\cdot$}
\put(7,24.48){$\cdot$}
\put(6.5,25.35){$\cdot$}
\put(6,26.25){$\cdot$}
\put(5.5,27.2){$\cdot$}
\put(5,28.2){$\cdot$}
\put(4.5,29.27){$\cdot$}
\put(4,30.4){$\cdot$}
\put(3.5,31.62){$\cdot$}
\put(3,32.94){$\cdot$}
\put(2.5,34.39){$\cdot$}
\put(2,36.00){$\cdot$}
\put(1.5,37.84){$\cdot$}
\put(1,40.05){$\cdot$}
\put(0.5,42.95){$\cdot$}
\put(0,50){$\cdot$}
\put(50,0){$\cdot$}
\put(50.5,0){$\cdot$}
\put(51,0.01){$\cdot$}
\put(51.5,0.02){$\cdot$}
\put(52,0.04){$\cdot$}
\put(52.5,0.06){$\cdot$}
\put(53,0.09){$\cdot$}
\put(53.5,0.12){$\cdot$}
\put(54,0.16){$\cdot$}
\put(54.5,0.2){$\cdot$}
\put(55,0.25){$\cdot$}
\put(55.5,0.3){$\cdot$}
\put(56,0.36){$\cdot$}
\put(56.5,0.42){$\cdot$}
\put(57,0.49){$\cdot$}
\put(57.5,0.56){$\cdot$}
\put(58,0.64){$\cdot$}
\put(58.5,0.73){$\cdot$}
\put(59,0.82){$\cdot$}
\put(59.5,0.91){$\cdot$}
\put(60,1.01){$\cdot$}
\put(60.5,1.11){$\cdot$}
\put(61,1.22){$\cdot$}
\put(61.5,1.34){$\cdot$}
\put(62,1.46){$\cdot$}
\put(62.5,1.59){$\cdot$}
\put(63,1.72){$\cdot$}
\put(63.5,1.86){$\cdot$}
\put(64,2.0){$\cdot$}
\put(64.5,2.15){$\cdot$}
\put(65,2.3){$\cdot$}
\put(65.5,2.46){$\cdot$}
\put(66,2.63){$\cdot$}
\put(66.5,2.8){$\cdot$}
\put(67,2.98){$\cdot$}
\put(67.5,3.16){$\cdot$}
\put(68,3.35){$\cdot$}
\put(68.5,3.55){$\cdot$}
\put(69,3.75){$\cdot$}
\put(69.5,3.96){$\cdot$}
\put(70,4.17){$\cdot$}
\put(70.5,4.39){$\cdot$}
\put(71,4.62){$\cdot$}
\put(71.5,4.86){$\cdot$}
\put(72,5.1){$\cdot$}
\put(72.5,5.35){$\cdot$}
\put(73,5.6){$\cdot$}
\put(73.5,5.87){$\cdot$}
\put(74,6.14){$\cdot$}
\put(74.5,6.41){$\cdot$}
\put(75,6.7){$\cdot$}
\put(75.5,6.99){$\cdot$}
\put(76,7.29){$\cdot$}
\put(76.5,7.6){$\cdot$}
\put(77,7.9){$\cdot$}
\put(77.5,8.24){$\cdot$}
\put(78,8.57){$\cdot$}
\put(78.5,8.92){$\cdot$}
\put(79,9.27){$\cdot$}
\put(79.5,9.63){$\cdot$}
\put(80,10.0){$\cdot$}
\put(80.5,10.38){$\cdot$}
\put(81,10.77){$\cdot$}
\put(81.5,11.17){$\cdot$}
\put(82,11.58){$\cdot$}
\put(82.5,12.0){$\cdot$}
\put(83,12.44){$\cdot$}
\put(83.5,12.88){$\cdot$}
\put(84,13.34){$\cdot$}
\put(84.5,13.81){$\cdot$}
\put(85,14.29){$\cdot$}
\put(85.5,14.79){$\cdot$}
\put(86,15.3){$\cdot$}
\put(86.5,15.83){$\cdot$}
\put(87,16.37){$\cdot$}
\put(87.5,16.93){$\cdot$}
\put(88,17.5){$\cdot$}
\put(88.5,18.1){$\cdot$}
\put(89,18.71){$\cdot$}
\put(89.5,19.34){$\cdot$}
\put(90,20.00){$\cdot$}
\put(90.5,20.68){$\cdot$}
\put(91,21.38){$\cdot$}
\put(91.5,22.11){$\cdot$}
\put(92,22.87){$\cdot$}
\put(92.5,23.66){$\cdot$}
\put(93,24.48){$\cdot$}
\put(93.5,25.35){$\cdot$}
\put(94,26.25){$\cdot$}
\put(94.5,27.2){$\cdot$}
\put(95,28.2){$\cdot$}
\put(95.5,29.27){$\cdot$}
\put(96,30.4){$\cdot$}
\put(96.5,31.62){$\cdot$}
\put(97,32.94){$\cdot$}
\put(97.5,34.39){$\cdot$}
\put(98,36.00){$\cdot$}
\put(98.5,37.84){$\cdot$}
\put(99,40.05){$\cdot$}
\put(99.5,42.95){$\cdot$}
\put(100,50){$\cdot$}

\put(0,50){$\cdot$}
\put(0.01,50.5){$\cdot$}
\put(0.01,51){$\cdot$}
\put(0.02,51.5){$\cdot$}
\put(0.04,52){$\cdot$}
\put(0.06,52.5){$\cdot$}
\put(0.09,53){$\cdot$}
\put(0.12,53.5){$\cdot$}
\put(0.16,54){$\cdot$}
\put(0.2,54.5){$\cdot$}
\put(0.25,55){$\cdot$}
\put(0.3,55.5){$\cdot$}
\put(0.36,56){$\cdot$}
\put(0.42,56.5){$\cdot$}
\put(0.49,57){$\cdot$}
\put(0.56,57.5){$\cdot$}
\put(0.64,58){$\cdot$}
\put(0.78,58.5){$\cdot$}
\put(0.82,59){$\cdot$}
\put(0.91,59.5){$\cdot$}
\put(1.01,60){$\cdot$}
\put(1.11,60.5){$\cdot$}
\put(1.22,61){$\cdot$}
\put(1.34,61.5){$\cdot$}
\put(1.46,62){$\cdot$}
\put(1.59,62.5){$\cdot$}
\put(1.72,63){$\cdot$}
\put(1.86,63.5){$\cdot$}
\put(2,64){$\cdot$}
\put(2.15,64.5){$\cdot$}
\put(2.3,65){$\cdot$}
\put(2.46,65.5){$\cdot$}
\put(2.63,66){$\cdot$}
\put(2.6,66.5){$\cdot$}
\put(2.98,67){$\cdot$}
\put(3.16,67.5){$\cdot$}
\put(3.35,68){$\cdot$}
\put(3.54,68.5){$\cdot$}
\put(3.75,69){$\cdot$}
\put(3.96,69.5){$\cdot$}
\put(4.17,70){$\cdot$}
\put(4.39,70.5){$\cdot$}
\put(4.62,71){$\cdot$}
\put(4.86,71.5){$\cdot$}
\put(5.1,72){$\cdot$}
\put(5.35,72.5){$\cdot$}
\put(5.6,73){$\cdot$}
\put(5.87,73.5){$\cdot$}
\put(6.14,74){$\cdot$}
\put(6.41,74.5){$\cdot$}
\put(6.7,75){$\cdot$}

\put(0,50){$\cdot$}
\put(0.01,49.5){$\cdot$}
\put(0.01,49){$\cdot$}
\put(0.02,48.5){$\cdot$}
\put(0.04,48){$\cdot$}
\put(0.06,47.5){$\cdot$}
\put(0.09,47){$\cdot$}
\put(0.12,46.5){$\cdot$}
\put(0.16,46){$\cdot$}
\put(0.2,45.5){$\cdot$}
\put(0.25,45){$\cdot$}
\put(0.3,44.5){$\cdot$}
\put(0.36,44){$\cdot$}
\put(0.42,43.5){$\cdot$}
\put(0.49,43){$\cdot$}
\put(0.56,42.5){$\cdot$}
\put(0.64,42){$\cdot$}
\put(0.78,41.5){$\cdot$}
\put(0.82,41){$\cdot$}
\put(0.91,40.5){$\cdot$}
\put(1.01,40){$\cdot$}
\put(1.11,39.5){$\cdot$}
\put(1.22,39){$\cdot$}
\put(1.34,38.5){$\cdot$}
\put(1.46,38){$\cdot$}
\put(1.59,37.5){$\cdot$}
\put(1.72,37){$\cdot$}
\put(1.86,36.5){$\cdot$}
\put(2,36){$\cdot$}
\put(2.15,35.5){$\cdot$}
\put(2.3,35){$\cdot$}
\put(2.46,34.5){$\cdot$}
\put(2.63,34){$\cdot$}
\put(2.6,33.5){$\cdot$}
\put(2.98,33){$\cdot$}
\put(3.16,32.5){$\cdot$}
\put(3.35,32){$\cdot$}
\put(3.54,31.5){$\cdot$}
\put(3.75,31){$\cdot$}
\put(3.96,30.5){$\cdot$}
\put(4.17,30){$\cdot$}
\put(4.39,29.5){$\cdot$}
\put(4.62,29){$\cdot$}
\put(4.86,28.5){$\cdot$}
\put(5.1,28){$\cdot$}
\put(5.35,27.5){$\cdot$}
\put(5.6,27){$\cdot$}
\put(5.87,26.5){$\cdot$}
\put(6.14,26){$\cdot$}
\put(6.41,25.5){$\cdot$}
\put(6.7,25){$\cdot$}

\put(100,50){$\cdot$}
\put(99.99,50.5){$\cdot$}
\put(99.99,51){$\cdot$}
\put(99.98,51.5){$\cdot$}
\put(99.96,52){$\cdot$}
\put(99.94,52.5){$\cdot$}
\put(99.91,53){$\cdot$}
\put(99.88,53.5){$\cdot$}
\put(99.84,54){$\cdot$}
\put(99.8,54.5){$\cdot$}
\put(99.75,55){$\cdot$}
\put(99.7,55.5){$\cdot$}
\put(99.64,56){$\cdot$}
\put(99.58,56.5){$\cdot$}
\put(99.51,57){$\cdot$}
\put(99.44,57.5){$\cdot$}
\put(99.36,58){$\cdot$}
\put(99.22,58.5){$\cdot$}
\put(99.18,59){$\cdot$}
\put(99.09,59.5){$\cdot$}
\put(98.99,60){$\cdot$}
\put(98.89,60.5){$\cdot$}
\put(98.78,61){$\cdot$}
\put(98.66,61.5){$\cdot$}
\put(98.54,62){$\cdot$}
\put(98.41,62.5){$\cdot$}
\put(98.28,63){$\cdot$}
\put(98.14,63.5){$\cdot$}
\put(98,64){$\cdot$}
\put(97.85,64.5){$\cdot$}
\put(97.7,65){$\cdot$}
\put(97.54,65.5){$\cdot$}
\put(97.37,66){$\cdot$}
\put(97.4,66.5){$\cdot$}
\put(97.02,67){$\cdot$}
\put(96.84,67.5){$\cdot$}
\put(96.65,68){$\cdot$}
\put(96.46,68.5){$\cdot$}
\put(96.25,69){$\cdot$}
\put(96.04,69.5){$\cdot$}
\put(95.83,70){$\cdot$}
\put(95.61,70.5){$\cdot$}
\put(95.38,71){$\cdot$}
\put(95.14,71.5){$\cdot$}
\put(94.9,72){$\cdot$}
\put(94.65,72.5){$\cdot$}
\put(94.4,73){$\cdot$}
\put(94.13,73.5){$\cdot$}
\put(93.86,74){$\cdot$}
\put(93.59,74.5){$\cdot$}
\put(93.3,75){$\cdot$}

\put(100,50){$\cdot$}
\put(99.99,49.5){$\cdot$}
\put(99.99,49){$\cdot$}
\put(99.98,48.5){$\cdot$}
\put(99.96,48){$\cdot$}
\put(99.94,47.5){$\cdot$}
\put(99.91,47){$\cdot$}
\put(99.88,46.5){$\cdot$}
\put(99.84,46){$\cdot$}
\put(99.8,45.5){$\cdot$}
\put(99.75,45){$\cdot$}
\put(99.7,44.5){$\cdot$}
\put(99.64,44){$\cdot$}
\put(99.58,43.5){$\cdot$}
\put(99.51,43){$\cdot$}
\put(99.44,42.5){$\cdot$}
\put(99.36,42){$\cdot$}
\put(99.22,41.5){$\cdot$}
\put(99.18,41){$\cdot$}
\put(99.09,40.5){$\cdot$}
\put(98.99,40){$\cdot$}
\put(98.89,39.5){$\cdot$}
\put(98.78,39){$\cdot$}
\put(98.66,38.5){$\cdot$}
\put(98.54,38){$\cdot$}
\put(98.41,37.5){$\cdot$}
\put(98.28,37){$\cdot$}
\put(98.14,36.5){$\cdot$}
\put(98,36){$\cdot$}
\put(97.85,35.5){$\cdot$}
\put(97.7,35){$\cdot$}
\put(97.54,34.5){$\cdot$}
\put(97.37,34){$\cdot$}
\put(97.4,33.5){$\cdot$}
\put(97.02,33){$\cdot$}
\put(96.84,32.5){$\cdot$}
\put(96.65,32){$\cdot$}
\put(96.46,31.5){$\cdot$}
\put(96.25,31){$\cdot$}
\put(96.04,30.5){$\cdot$}
\put(95.83,30){$\cdot$}
\put(95.61,29.5){$\cdot$}
\put(95.38,29){$\cdot$}
\put(95.14,28.5){$\cdot$}
\put(94.9,28){$\cdot$}
\put(94.65,27.5){$\cdot$}
\put(94.4,27){$\cdot$}
\put(94.13,26.5){$\cdot$}
\put(93.86,26){$\cdot$}
\put(93.59,25.5){$\cdot$}
\put(93.3,25){$\cdot$}
\end{picture}
\]
\vspace{0.3cm}
\begin{center}
\begin{minipage}{25pc}
{\small {\bf Fig.\,7:} \textbf{Bloch sphere}. The point $P$ on the surface presents a pure state, and the points inside the ball present mixed states.}
\end{minipage}
\end{center}
\end{figure}
In the case of mixed state for the length of the polarization vector we have $0<|\boldsymbol{P}|^2<1$. Therefore, the density matrix of the qubit can be represented by a point in three-dimensional space. That is, there exists one-to-one correspondence between the density matrix and the points of the unit ball. Pure states correspond to the points on the surface of this ball, and mixed states are described by the points inside the ball. Bloch sphere allows one to illustrate as a \emph{classical domain} arises from a \emph{quantum domain} in the process of decoherence \cite{Zur03}. When the system takes two possible values (positions) `up' and `down' along $Z$-axis, then in the limits of sphere the points on this axis present a totality of classical states, which can be appeared in the result of decoherence.

Carrying on the analogy between qubits and spinors (twistors), we see that there is a deep relationship between spinor structure (twistor programme) and theory of quantum information on the one hand and decoherence theory on the other hand. The underlying spinor structure presents by itself the mathematics that works on the level of nonlocal quantum substrate.

\section{Modulo 8 periodicity and particle representations of $\spin_+(1,3)$}
\begin{thm}
The action of the group $BW_{\R}\simeq\dZ_8$ induces modulo 2 periodic relations on the system of real representations of the group $\spin_+(1,3)\simeq\SL(2,\C)$.
\end{thm}
\begin{proof}
First of all, for the algebras of type $\cl_{0,q}$ ($q\equiv 1\pmod{2}$) there exists a decomposition $\cl_{0,q}\simeq\cl^+_{0,q}\oplus\cl^+_{0,q}$, where $\cl^+_{0,q}$ is an even subalgebra of $\cl_{0,q}$. In virtue of an isomorphism $\cl^+_{0,q}\simeq\cl_{0,q-1}$ we have
$\cl_{0,q}\simeq\cl_{0,q-1}\oplus\cl_{0,q-1}$. This decomposition
can be represented by the following scheme:
\[
\unitlength=0.5mm
\begin{picture}(70,50)
\put(35,40){\vector(2,-3){15}} \put(35,40){\vector(-2,-3){15}}
\put(28.25,42){$\cl_{0,q}$} \put(16,28){$\lambda_{+}$}
\put(49.5,28){$\lambda_{-}$} \put(9,9.20){$\cl_{0,q-1}$}
\put(47,9){$\cl_{0,q-1}$} \put(32.5,10){$\oplus$}
\end{picture}
\]
Here \emph{central idempotents}
\[
\lambda^+=\frac{1+\e_1\e_2\cdots\e_{q}}{2},\quad
\lambda^-=\frac{1-\e_1\e_2\cdots\e_{q}}{2}
\]
satisfy the relations $(\lambda^+)^2=\lambda^+$,
$(\lambda^-)^2=\lambda^-$, $\lambda^+\lambda^-=0$.
Further, there is a homomorphic mapping
\begin{equation}\label{Hmap}
\epsilon:\;\cl_{0,q}\longrightarrow{}^\epsilon\cl_{0,q-1},
\end{equation}
where
\[
{}^\epsilon\cl_{0,q-1}\simeq\cl_{0,q}/\Ker\,\epsilon
\]
is a quotient algebra, $\Ker\,\epsilon=\{\cA^1-\omega\cA^1\}$ is a
kernel of the homomorphism $\epsilon$, $\cA^1\in\cl_{0,q-1}$ is an arbitrary element of $\cl_{0,q-1}$, and
$\omega=\e_1\e_2\cdots\e_q\in\cl_{0,q}$ is a volume element of $\cl_{0,q}$.
Therefore, in virtue of the homomorphic mapping (\ref{Hmap}) we can replace the double representations of $\spin_+(1,3)$ by quotient representations ${}^\epsilon\boldsymbol{\tau}^r_{l\dot{l}}$ and ${}^\epsilon\boldsymbol{\tau}^q_{l\dot{l}}$, where ${}^\epsilon\boldsymbol{\tau}^r_{l\dot{l}}$ is a real quotient representation, and ${}^\epsilon\boldsymbol{\tau}^q_{l\dot{l}}$ is a quaternionic quotient representation. About detailed structure of the quotient representations of $\spin_+(1,3)$ see \cite{Var04,Var1402}.

On the underlying spinor structure the first step $\cl^+_{0,1}\overset{1}{\longrightarrow}\cl_{0,1}$ ($\cl_{0,0}\overset{1}{\longrightarrow}\cl_{0,1}$) of the Brauer-Wall group $BW_{\R}\simeq\dZ_8$ generates a transition $\boldsymbol{\tau}^r_{0,0}\overset{1}{\longrightarrow}{}^\epsilon\boldsymbol{\tau}^r_{0,0}$, where $\boldsymbol{\tau}^r_{0,0}$ is a real representation of $\spin_+(1,3)$ associated with the algebra $\cl_{0,0}$ ($p-q\equiv 0\pmod{8}$, $\K\simeq\R$), ${}^\epsilon\boldsymbol{\tau}^r_{0,0}$ is a real quotient representation of $\spin_+(1,3)$ associated with the quotient algebra ${}^\epsilon\cl_{0,0}\simeq\cl_{0,1}/\Ker\,\epsilon$, since in virtue of $\cl^+_{0,1}\simeq\cl_{0,0}$ we have $\cl_{0,1}\simeq\cl_{0,0}\oplus i\cl_{0,0}$. The second step $\cl^+_{0,2}\overset{2}{\longrightarrow}\cl_{0,2}$ ($\cl_{0,1}\overset{2}{\longrightarrow}\cl_{0,2}$) generates ${}^\epsilon\boldsymbol{\tau}^r_{0,0}\overset{2}{\longrightarrow}\boldsymbol{\tau}^q_{0,\frac{1}{2}}$, where $\boldsymbol{\tau}^q_{0,\frac{1}{2}}$ is a quaternionic representation of $\spin_+(1,3)$ associated with the algebra $\cl_{0,2}$ ($p-q\equiv 6\pmod{8}$, $\K\simeq\BH$). The third step $\cl_{0,2}\overset{3}{\longrightarrow}\cl_{0,3}$ of $BW_{\R}\simeq\dZ_8$ induces a transition $\boldsymbol{\tau}^q_{0,\frac{1}{2}}\overset{3}{\longrightarrow}{}^\epsilon\boldsymbol{\tau}^q_{0,\frac{1}{2}}$, where ${}^\epsilon\boldsymbol{\tau}^q_{0,\frac{1}{2}}$ is a quaternionic quotient representation associated with the quotient algebra ${}^\epsilon\cl_{0,2}\simeq\cl_{0,3}/\Ker\,\epsilon$. The following step $\cl_{0,3}\overset{4}{\longrightarrow}\cl_{0,4}$ generates ${}^\epsilon\boldsymbol{\tau}^q_{0,\frac{1}{2}}\overset{4}{\longrightarrow}\boldsymbol{\tau}^q_{0,1}$, where $\boldsymbol{\tau}^q_{0,1}$ is a quaternionic representation associated with $\cl_{0,4}$ ($p-q\equiv 4\pmod{8}$, $\K\simeq\BH$) in the underlying spinor structure. The fifth step $\cl_{0,4}\overset{5}{\longrightarrow}\cl_{0,5}$ of the first cycle of $BW_{\R}\simeq\dZ_8$ leads to $\boldsymbol{\tau}^q_{0,1}\overset{5}{\longrightarrow}{}^\epsilon\boldsymbol{\tau}^q_{0,1}$, where ${}^\epsilon\boldsymbol{\tau}^q_{0,1}$ is a quaternionic quotient representation associated with ${}^\epsilon\cl_{0,4}\simeq\cl_{0,5}/\Ker\,\epsilon$. In turn, the sixth $\cl_{0,5}\overset{6}{\longrightarrow}\cl_{0,6}$ and seventh $\cl_{0,6}\overset{7}{\longrightarrow}\cl_{0,7}$ steps generate transitions ${}^\epsilon\boldsymbol{\tau}^q_{0,1}\overset{6}{\longrightarrow}\boldsymbol{\tau}^r_{0,\frac{3}{2}}$ and $\boldsymbol{\tau}^r_{0,\frac{3}{2}}\overset{7}{\longrightarrow}{}^\epsilon\boldsymbol{\tau}^r_{0,\frac{3}{2}}$, where $\boldsymbol{\tau}^r_{0,\frac{3}{2}}$ is a real representation of $\spin_+(1,3)$ associated with the algebra $\cl_{0,6}$ ($p-q\equiv 2\pmod{8}$, $\K\simeq\R$) and ${}^\epsilon\boldsymbol{\tau}^r_{0,\frac{3}{2}}$ is a real quotient representation associated with ${}^\epsilon\cl_{0,6}\simeq\cl_{0,7}/\Ker\,\epsilon$. The eighth step $\cl_{0,7}\overset{8}{\longrightarrow}\cl_{0,8}$ finishes the first cycle ($r=0$) of $BW_{\R}\simeq\dZ_8$ and induces a transition ${}^\epsilon\boldsymbol{\tau}^r_{0,\frac{3}{2}}\overset{8}{\longrightarrow}\boldsymbol{\tau}^r_{0,2}$, where $\boldsymbol{\tau}^r_{0,2}$ is a real representation associated with the algebra $\cl_{0,8}$. The first cycle generates the first eight representations ($\boldsymbol{\tau}^r_{0,0}$, ${}^\epsilon\boldsymbol{\tau}^r_{0,0}$, $\boldsymbol{\tau}^q_{0,\frac{1}{2}}$, ${}^\epsilon\boldsymbol{\tau}^q_{0,\frac{1}{2}}$, $\boldsymbol{\tau}^q_{0,1}$, ${}^\epsilon\boldsymbol{\tau}^q_{0,1}$, $\boldsymbol{\tau}^r_{0,\frac{3}{2}}$,  ${}^\epsilon\boldsymbol{\tau}^r_{0,\frac{3}{2}}$) associated with the first eight squares ($\cl_{0,q}$, $q=0,\ldots,7$) of the spinorial chessboard (see Fig.\,1). However, the pairs ($\boldsymbol{\tau}^r_{0,0}$, ${}^\epsilon\boldsymbol{\tau}^r_{0,0}$), ($\boldsymbol{\tau}^q_{0,\frac{1}{2}}$, ${}^\epsilon\boldsymbol{\tau}^q_{0,\frac{1}{2}}$), ($\boldsymbol{\tau}^q_{0,1}$, ${}^\epsilon\boldsymbol{\tau}^q_{0,1}$), ($\boldsymbol{\tau}^r_{0,\frac{3}{2}}$,  ${}^\epsilon\boldsymbol{\tau}^r_{0,\frac{3}{2}}$) present particles of the same spin $s$, respectively, $s=0,\frac{1}{2},1,\frac{3}{2}$. Particles within the pair exist in a quantum superposition. Therefore, the pair ($\boldsymbol{\tau}^r_{0,0}$, ${}^\epsilon\boldsymbol{\tau}^r_{0,0}$) belongs to spin-0 line, the pair ($\boldsymbol{\tau}^q_{0,\frac{1}{2}}$, ${}^\epsilon\boldsymbol{\tau}^q_{0,\frac{1}{2}}$) belongs to the dual spin-1/2 line and so on. The following eight representations $\boldsymbol{\tau}^r_{\frac{1}{2},0}$, ${}^\epsilon\boldsymbol{\tau}^r_{\frac{1}{2},0}$, $\boldsymbol{\tau}^r_{\frac{1}{2},\frac{1}{2}}$, ${}^\epsilon\boldsymbol{\tau}^r_{\frac{1}{2},\frac{1}{2}}$, $\boldsymbol{\tau}^q_{\frac{1}{2},1}$, ${}^\epsilon\boldsymbol{\tau}^q_{\frac{1}{2},1}$, $\boldsymbol{\tau}^q_{\frac{1}{2},\frac{3}{2}}$,  ${}^\epsilon\boldsymbol{\tau}^q_{\frac{1}{2},\frac{3}{2}}$ are generated by the first cycle also via the rule $\cl_{2,q}\simeq\cl_{2,0}\otimes\cl_{0,q}$, $q=0,\ldots,7$, here $\cl_{0,2}\rightsquigarrow\boldsymbol{\tau}^r_{\frac{1}{2},0}$. In like manner (via the action of the first cycle of $BW_{\R}\simeq\dZ_8$) we obtain the first representation block of $\spin_+(1,3)$ (see Fig.\,8) associated with the spinorial chessboard.
\begin{figure}[ht]
\[
\dgARROWPARTS=28
\dgARROWLENGTH=0.5em
\dgHORIZPAD=1.9em 
\dgVERTPAD=2.3ex 
\begin{diagram}
\node[5]{\vdots}\\
\node[5]{\overset{\boldsymbol{\tau}^r_{2,2}}{\bullet}}\\
\node[4]{\overset{\boldsymbol{\tau}^q_{\frac{3}{2},2}}{\bullet}}
\arrow{n,..,-}
\node[2]{\overset{\boldsymbol{\tau}^r_{2,\frac{3}{2}}}{\bullet}}
\arrow{n,..,-}\\
\node[3]{\overset{\boldsymbol{\tau}^q_{1,2}}{\bullet}}
\arrow[2]{n,..,-}
\node[2]{\overset{\boldsymbol{\tau}^r_{\frac{3}{2},\frac{3}{2}}}{\bullet}}
\node[2]{\overset{\boldsymbol{\tau}^q_{2,1}}{\bullet}}
\arrow[2]{n,..,-}\\
\node[2]{\overset{\boldsymbol{\tau}^r_{\frac{1}{2},2}}{\bullet}}
\arrow[3]{n,..,-}
\node[2]{\overset{\boldsymbol{\tau}^q_{1,\frac{3}{2}}}{\bullet}}
\node[2]{\overset{\boldsymbol{\tau}^r_{\frac{3}{2},1}}{\bullet}}
\node[2]{\overset{\boldsymbol{\tau}^q_{2,\frac{1}{2}}}{\bullet}}
\arrow[3]{n,..,-}\\
\node{\overset{\boldsymbol{\tau}^r_{0,2}}{\bullet}}
\arrow[4]{n,..,-}\arrow[4]{s,..,-}
\node[2]{\overset{\boldsymbol{\tau}^q_{\frac{1}{2},\frac{3}{2}}}{\bullet}}
\node[2]{\overset{\boldsymbol{\tau}^r_{1,1}}{\bullet}}
\node[2]{\overset{\boldsymbol{\tau}^q_{\frac{3}{2},\frac{1}{2}}}{\bullet}}
\node[2]{\overset{\boldsymbol{\tau}^r_{2,0}}{\bullet}}
\arrow[4]{n,..,-}\arrow[4]{s,..,-}\\
\node[2]{\overset{\boldsymbol{\tau}^r_{0,\frac{3}{2}}}{\bullet}}
\arrow[3]{s,..,-}
\node[2]{\overset{\boldsymbol{\tau}^q_{\frac{1}{2},1}}{\bullet}}
\node[2]{\overset{\boldsymbol{\tau}^r_{1,\frac{1}{2}}}{\bullet}}
\node[2]{\overset{\boldsymbol{\tau}^q_{\frac{3}{2},0}}{\bullet}}
\arrow[3]{s,..,-}\\
\node[3]{\overset{\boldsymbol{\tau}^q_{0,1}}{\bullet}}
\arrow[2]{s,..,-}
\node[2]{\overset{\boldsymbol{\tau}^r_{\frac{1}{2},\frac{1}{2}}}{\bullet}}
\node[2]{\overset{\boldsymbol{\tau}^q_{1,0}}{\bullet}}
\arrow[2]{s,..,-}\\
\node[4]{\overset{\boldsymbol{\tau}^q_{0,\frac{1}{2}}}{\bullet}}
\arrow{s,..,-}
\node[2]{\overset{\boldsymbol{\tau}^r_{\frac{1}{2},0}}{\bullet}}
\arrow{s,..,-}\\
\node[5]{\overset{\boldsymbol{\tau}^r_{0,0}}{\bullet}}
\arrow[5]{e}
\arrow[5]{w,-}
\end{diagram}
\]
\begin{center}\begin{minipage}{32pc}{\small {\bf Fig.\,8:} The first representation block of the group $\spin_+(1,3)$ associated with the spinorial chessboard of the first order (see Fig.\,1). This block is generated by the first cycle of the group $BW_{\R}\simeq\dZ_8$.}\end{minipage}\end{center}
\end{figure}
From Fig.\,8 it follows that spin-lines within the block are divided into real and quaternionic spin-lines.

The second cycle ($r=1$) of the group $BW_{\R}\simeq\dZ_8$ consists of the following eight steps:\\
1) $h=1$, $r=1$, $\cl_{0,8}\overset{1}{\longrightarrow}\cl_{0,9}\;\rightsquigarrow\;$ $\boldsymbol{\tau}^r_{0,2}\overset{1}{\longrightarrow}{}^\epsilon\boldsymbol{\tau}^r_{0,2}$;\\
2) $h=2$, $r=1$, $\cl_{0,9}\overset{2}{\longrightarrow}\cl_{0,10}\;\rightsquigarrow\;$ ${}^\epsilon\boldsymbol{\tau}^r_{0,2}\overset{2}{\longrightarrow}\boldsymbol{\tau}^q_{0,\frac{5}{2}}$;\\
3) $h=3$, $r=1$, $\cl_{0,10}\overset{3}{\longrightarrow}\cl_{0,11}\;\rightsquigarrow\;$ $\boldsymbol{\tau}^q_{0,\frac{5}{2}}\overset{3}{\longrightarrow}{}^\epsilon\boldsymbol{\tau}^q_{0,\frac{5}{2}}$;\\
4) $h=4$, $r=1$, $\cl_{0,11}\overset{4}{\longrightarrow}\cl_{0,12}\;\rightsquigarrow\;$ ${}^\epsilon\boldsymbol{\tau}^q_{0,\frac{5}{2}}\overset{4}{\longrightarrow}\boldsymbol{\tau}^q_{0,3}$;\\
5) $h=5$, $r=1$, $\cl_{0,12}\overset{5}{\longrightarrow}\cl_{0,13}\;\rightsquigarrow\;$ $\boldsymbol{\tau}^q_{0,3}\overset{5}{\longrightarrow}{}^\epsilon\boldsymbol{\tau}^q_{0,3}$;\\
6) $h=6$, $r=1$, $\cl_{0,13}\overset{6}{\longrightarrow}\cl_{0,14}\;\rightsquigarrow\;$ ${}^\epsilon\boldsymbol{\tau}^q_{0,3}\overset{6}{\longrightarrow}\boldsymbol{\tau}^r_{0,\frac{7}{2}}$;\\
7) $h=7$, $r=1$, $\cl_{0,14}\overset{7}{\longrightarrow}\cl_{0,15}\;\rightsquigarrow\;$ $\boldsymbol{\tau}^r_{0,\frac{7}{2}}\overset{7}{\longrightarrow}{}^\epsilon\boldsymbol{\tau}^r_{0,\frac{7}{2}}$;\\
8) $h=8$, $r=1$, $\cl_{0,15}\overset{8}{\longrightarrow}\cl_{0,16}\;\rightsquigarrow\;$ ${}^\epsilon\boldsymbol{\tau}^r_{0,\frac{7}{2}}\overset{8}{\longrightarrow}\boldsymbol{\tau}^r_{0,4}$.

Further, after the eighth cycle ($r=7$), generated by the steps\\
1) $h=1$, $r=7$, $\cl_{0,56}\overset{1}{\longrightarrow}\cl_{0,57}\;\rightsquigarrow\;$ $\boldsymbol{\tau}^r_{0,14}\overset{1}{\longrightarrow}{}^\epsilon\boldsymbol{\tau}^r_{0,14}$;\\
2) $h=2$, $r=7$, $\cl_{0,57}\overset{2}{\longrightarrow}\cl_{0,58}\;\rightsquigarrow\;$ ${}^\epsilon\boldsymbol{\tau}^r_{0,14}\overset{2}{\longrightarrow}\boldsymbol{\tau}^q_{0,\frac{29}{2}}$;\\
3) $h=3$, $r=7$, $\cl_{0,58}\overset{3}{\longrightarrow}\cl_{0,59}\;\rightsquigarrow\;$ $\boldsymbol{\tau}^q_{0,\frac{29}{2}}\overset{3}{\longrightarrow}{}^\epsilon\boldsymbol{\tau}^q_{0,\frac{29}{2}}$;\\
4) $h=4$, $r=7$, $\cl_{0,59}\overset{4}{\longrightarrow}\cl_{0,60}\;\rightsquigarrow\;$ ${}^\epsilon\boldsymbol{\tau}^q_{0,\frac{29}{2}}\overset{4}{\longrightarrow}\boldsymbol{\tau}^q_{0,15}$;\\
5) $h=5$, $r=7$, $\cl_{0,60}\overset{5}{\longrightarrow}\cl_{0,61}\;\rightsquigarrow\;$ $\boldsymbol{\tau}^q_{0,15}\overset{5}{\longrightarrow}{}^\epsilon\boldsymbol{\tau}^q_{0,15}$;\\
6) $h=6$, $r=7$, $\cl_{0,61}\overset{6}{\longrightarrow}\cl_{0,62}\;\rightsquigarrow\;$ ${}^\epsilon\boldsymbol{\tau}^q_{0,15}\overset{6}{\longrightarrow}\boldsymbol{\tau}^r_{0,\frac{31}{2}}$;\\
7) $h=7$, $r=7$, $\cl_{0,62}\overset{7}{\longrightarrow}\cl_{0,63}\;\rightsquigarrow\;$ $\boldsymbol{\tau}^r_{0,\frac{31}{2}}\overset{7}{\longrightarrow}{}^\epsilon\boldsymbol{\tau}^r_{0,\frac{31}{2}}$;\\
8) $h=8$, $r=7$, $\cl_{0,63}\overset{8}{\longrightarrow}\cl_{0,64}\;\rightsquigarrow\;$ ${}^\epsilon\boldsymbol{\tau}^r_{0,\frac{31}{2}}\overset{8}{\longrightarrow}\boldsymbol{\tau}^r_{0,16}$,\\
we come to a fractal self-similar algebraic structure of the second order which induces on the system of real representations of $\spin_+(1,3)$ modulo 2 periodic structure shown on the Fig.\,9.
\begin{figure}[ht]
\[
\dgARROWPARTS=28
\dgARROWLENGTH=0.5em
\dgHORIZPAD=1.9em 
\dgVERTPAD=2.3ex 
\begin{diagram}
\node[5]{\vdots}\\
\node[5]{\overset{\boldsymbol{\tau}^r_{16,16}}{\bullet}}\\
\node[4]{\overset{\boldsymbol{\tau}^q_{\frac{31}{2},16}}{\bullet}}
\arrow[2]{n,..,-}
\node[2]{\overset{\boldsymbol{\tau}^r_{16,\frac{31}{2}}}{\bullet}}
\arrow[2]{n,..,-}\\
\node[3]{\overset{\boldsymbol{\tau}^q_{15,16}}{\bullet}}
\arrow[3]{n,..,-}
\node[2]{\overset{\boldsymbol{\tau}^r_{\frac{31}{2},\frac{31}{2}}}{\bullet}}
\node[2]{\overset{\boldsymbol{\tau}^q_{16,15}}{\bullet}}
\arrow[3]{n,..,-}\\
\node[2]{\overset{\boldsymbol{\tau}^r_{\frac{29}{2},16}}{\bullet}}
\arrow[4]{n,..,-}
\node[2]{\overset{\boldsymbol{\tau}^q_{15,\frac{31}{2}}}{\bullet}}
\node[2]{\overset{\boldsymbol{\tau}^r_{\frac{31}{2},15}}{\bullet}}
\node[2]{\overset{\boldsymbol{\tau}^q_{16,\frac{29}{2}}}{\bullet}}
\arrow[4]{n,..,-}\\
\node{\overset{\boldsymbol{\tau}^r_{14,16}}{\bullet}}
\arrow[5]{n,..,-}\arrow[5]{s,..,-}
\node[2]{\overset{\boldsymbol{\tau}^q_{\frac{29}{2},\frac{31}{2}}}{\bullet}}
\node[2]{\overset{\boldsymbol{\tau}^r_{15,15}}{\bullet}}
\node[2]{\overset{\boldsymbol{\tau}^q_{\frac{31}{2},\frac{29}{2}}}{\bullet}}
\node[2]{\overset{\boldsymbol{\tau}^r_{16,14}}{\bullet}}
\arrow[5]{n,..,-}\arrow[5]{s,..,-}\\
\node[2]{\overset{\boldsymbol{\tau}^r_{14,\frac{31}{2}}}{\bullet}}
\arrow[4]{s,..,-}
\node[2]{\overset{\boldsymbol{\tau}^q_{\frac{29}{2},15}}{\bullet}}
\node[2]{\overset{\boldsymbol{\tau}^r_{15,\frac{29}{2}}}{\bullet}}
\node[2]{\overset{\boldsymbol{\tau}^q_{\frac{31}{2},14}}{\bullet}}
\arrow[4]{s,..,-}\\
\node[3]{\overset{\boldsymbol{\tau}^q_{14,15}}{\bullet}}
\arrow[3]{s,..,-}
\node[2]{\overset{\boldsymbol{\tau}^r_{\frac{29}{2},\frac{29}{2}}}{\bullet}}
\node[2]{\overset{\boldsymbol{\tau}^q_{15,14}}{\bullet}}
\arrow[3]{s,..,-}\\
\node[4]{\overset{\boldsymbol{\tau}^q_{14,\frac{29}{2}}}{\bullet}}
\arrow[2]{s,..,-}
\node[2]{\overset{\boldsymbol{\tau}^r_{\frac{29}{2},14}}{\bullet}}
\arrow[2]{s,..,-}\\
\node[5]{\overset{\boldsymbol{\tau}^r_{14,14}}{\bullet}}\\
\node[5]{\vdots}\\
\node[5]{\overset{\boldsymbol{\tau}^r_{4,4}}{\bullet}}\\
\node[4]{\overset{\boldsymbol{\tau}^q_{\frac{7}{2},4}}{\bullet}}
\arrow[2]{n,..,-}
\node[2]{\overset{\boldsymbol{\tau}^r_{4,\frac{7}{2}}}{\bullet}}
\arrow[2]{n,..,-}\\
\node[3]{\overset{\boldsymbol{\tau}^q_{3,4}}{\bullet}}
\arrow[3]{n,..,-}
\node[2]{\overset{\boldsymbol{\tau}^r_{\frac{7}{2},\frac{7}{2}}}{\bullet}}
\node[2]{\overset{\boldsymbol{\tau}^q_{4,3}}{\bullet}}
\arrow[3]{n,..,-}\\
\node[2]{\overset{\boldsymbol{\tau}^r_{\frac{5}{2},4}}{\bullet}}
\arrow[4]{n,..,-}
\node[2]{\overset{\boldsymbol{\tau}^q_{3,\frac{7}{2}}}{\bullet}}
\node[2]{\overset{\boldsymbol{\tau}^r_{\frac{7}{2},3}}{\bullet}}
\node[2]{\overset{\boldsymbol{\tau}^q_{4,\frac{5}{2}}}{\bullet}}
\arrow[4]{n,..,-}\\
\node{\overset{\boldsymbol{\tau}^r_{2,4}}{\bullet}}
\arrow[5]{n,..,-}\arrow[4]{s,..,-}
\node[2]{\overset{\boldsymbol{\tau}^q_{\frac{5}{2},\frac{7}{2}}}{\bullet}}
\node[2]{\overset{\boldsymbol{\tau}^r_{3,3}}{\bullet}}
\node[2]{\overset{\boldsymbol{\tau}^q_{\frac{7}{2},\frac{5}{2}}}{\bullet}}
\node[2]{\overset{\boldsymbol{\tau}^r_{4,2}}{\bullet}}
\arrow[5]{n,..,-}\arrow[4]{s,..,-}\\
\node[2]{\overset{\boldsymbol{\tau}^r_{2,\frac{7}{2}}}{\bullet}}
\arrow[3]{s,..,-}
\node[2]{\overset{\boldsymbol{\tau}^q_{\frac{5}{2},3}}{\bullet}}
\node[2]{\overset{\boldsymbol{\tau}^r_{3,\frac{5}{2}}}{\bullet}}
\node[2]{\overset{\boldsymbol{\tau}^q_{\frac{7}{2},2}}{\bullet}}
\arrow[3]{s,..,-}\\
\node[3]{\overset{\boldsymbol{\tau}^q_{2,3}}{\bullet}}
\arrow[2]{s,..,-}
\node[2]{\overset{\boldsymbol{\tau}^r_{\frac{5}{2},\frac{5}{2}}}{\bullet}}
\node[2]{\overset{\boldsymbol{\tau}^q_{3,2}}{\bullet}}
\arrow[2]{s,..,-}\\
\node[4]{\overset{\boldsymbol{\tau}^q_{2,\frac{5}{2}}}{\bullet}}
\arrow{s,..,-}
\node[2]{\overset{\boldsymbol{\tau}^r_{\frac{5}{2},2}}{\bullet}}
\arrow{s,..,-}\\
\node[5]{\overset{\boldsymbol{\tau}^r_{2,2}}{\bullet}}\\
\node[4]{\overset{\boldsymbol{\tau}^q_{\frac{3}{2},2}}{\bullet}}
\arrow{n,..,-}
\node[2]{\overset{\boldsymbol{\tau}^r_{2,\frac{3}{2}}}{\bullet}}
\arrow{n,..,-}\\
\node[3]{\overset{\boldsymbol{\tau}^q_{1,2}}{\bullet}}
\arrow[2]{n,..,-}
\node[2]{\overset{\boldsymbol{\tau}^r_{\frac{3}{2},\frac{3}{2}}}{\bullet}}
\node[2]{\overset{\boldsymbol{\tau}^q_{2,1}}{\bullet}}
\arrow[2]{n,..,-}\\
\node[2]{\overset{\boldsymbol{\tau}^r_{\frac{1}{2},2}}{\bullet}}
\arrow[3]{n,..,-}
\node[2]{\overset{\boldsymbol{\tau}^q_{1,\frac{3}{2}}}{\bullet}}
\node[2]{\overset{\boldsymbol{\tau}^r_{\frac{3}{2},1}}{\bullet}}
\node[2]{\overset{\boldsymbol{\tau}^q_{2,\frac{1}{2}}}{\bullet}}
\arrow[3]{n,..,-}\\
\node{\overset{\boldsymbol{\tau}^r_{0,2}}{\bullet}}
\arrow[4]{n,..,-}\arrow[4]{s,..,-}
\node[2]{\overset{\boldsymbol{\tau}^q_{\frac{1}{2},\frac{3}{2}}}{\bullet}}
\node[2]{\overset{\boldsymbol{\tau}^r_{1,1}}{\bullet}}
\node[2]{\overset{\boldsymbol{\tau}^q_{\frac{3}{2},\frac{1}{2}}}{\bullet}}
\node[2]{\overset{\boldsymbol{\tau}^r_{2,0}}{\bullet}}
\arrow[4]{n,..,-}\arrow[4]{s,..,-}\\
\node[2]{\overset{\boldsymbol{\tau}^r_{0,\frac{3}{2}}}{\bullet}}
\arrow[3]{s,..,-}
\node[2]{\overset{\boldsymbol{\tau}^q_{\frac{1}{2},1}}{\bullet}}
\node[2]{\overset{\boldsymbol{\tau}^r_{1,\frac{1}{2}}}{\bullet}}
\node[2]{\overset{\boldsymbol{\tau}^q_{\frac{3}{2},0}}{\bullet}}
\arrow[3]{s,..,-}\\
\node[3]{\overset{\boldsymbol{\tau}^q_{0,1}}{\bullet}}
\arrow[2]{s,..,-}
\node[2]{\overset{\boldsymbol{\tau}^r_{\frac{1}{2},\frac{1}{2}}}{\bullet}}
\node[2]{\overset{\boldsymbol{\tau}^q_{1,0}}{\bullet}}
\arrow[2]{s,..,-}\\
\node[4]{\overset{\boldsymbol{\tau}^q_{0,\frac{1}{2}}}{\bullet}}
\arrow{s,..,-}
\node[2]{\overset{\boldsymbol{\tau}^r_{\frac{1}{2},0}}{\bullet}}
\arrow{s,..,-}\\
\node[5]{\overset{\boldsymbol{\tau}^r_{0,0}}{\bullet}}
\arrow[5]{e}
\arrow[5]{w,-}
\end{diagram}
\]
\begin{center}\begin{minipage}{32pc}{\small {\bf Fig.\,9:} The representation block of the second order of the group $\spin_+(1,3)$ (the main diagonal of this block). This block is generated by the eight cycles of the group $BW_{\R}\simeq\dZ_8$.}\end{minipage}\end{center}
\end{figure}
Thus, we have here the representation block of the second order, which, obviously, can be extended to infinity (to the blocks of any order) via the consecutive cycles of $BW_{\R}\simeq\dZ_8$.
\end{proof}


\begin{thebibliography}{00}
\bibitem{Dir28} P.\,A.\,M. Dirac, {\it The Quantum Theory of the Electron I
{\rm\&} II}. Proc. Roy. Soc. A{\bf 117} (1928) 610--624 \& A{\bf 118} (1928) 351-361.
%
\bibitem{LU31} O. Laport, G.\,E. Uhlenbeck, {\it Application of spinor analysis
to the Maxwell and Dirac equations}. Phys. Rev. {\bf 37} (1931) 1380--1397.
%
\bibitem{Wae29} B.\,L. van der Waerden, {\it Spinoranalyse}. Nachr. d. Ces. d.
Wiss. G\"{o}ttingen (1929) 100--109.
%
\bibitem{Var022} V.\,V. Varlamov, {\it About Algebraic Foundations of
Majorana-Oppenheimer Quantum Electrodynamics and de Broglie-Jordan Neutrino
Theory of Light}. Annales de la Fondation Louis de Broglie. \textbf{27} (2002) 273--286; arXiv:math-ph/0109024 (2001).
%
\bibitem{Pen77} R. Penrose, \textit{The twistor programme}. Rep. Math. Phys. \textbf{12} (1977) 65--76.
%
\bibitem{PM72} R. Penrose, M.\,A.\,H. MacCallum, \textit{Twistor theory: an approach to the quantization of fields and space-time}. Physics Reports. \textbf{6} (1972) 241--316.
%
\bibitem{JZKGKS} E. Joos, H.\,D. Zeh, C. Kiefer, D.\,J.\,W. Giulini, J. Kupsch, I.-O. Stamatescu, \textit{Decoherence and Appearence of a Classical World in Quantum Theory}. Springer-Verlag,
Berlin, 2003.
%
\bibitem{Var12} V.\,V. Varlamov, \textit{Cyclic structures of Cliffordian supergroups and particle
representations of $\spin_+(1,3)$}. Adv. Appl. Clifford Algebras. \textbf{24} (2014) 849--874; arXiv: 1207.6162 [math-ph] (2012).
%
\bibitem{Wig39}  E.\,P. Wigner, {\it On unitary representations of the
inhomogeneous Lorentz group}. Ann. Math. {\bf 40} (1939) 149--204.
%
\bibitem{Var1402} V.\,V. Varlamov, \textit{Spinor structure and internal symmetries}. arXiv: 1409.1400 [math-ph] (2014).
%
\bibitem{Var04} V.\,V. Varlamov, \emph{Universal Coverings of Orthogonal
Groups}. Adv. Appl. Clifford Algebras. \textbf{14} (2004) 81--168;
arXiv:math-ph/0405040 (2004).
%
\bibitem{Var1401} V.\,V. Varlamov, \emph{$CPT$ groups of spinor fields in de Sitter and anti-de Sitter spaces}. arXiv: 1401.7723 [math-ph]
(2014), to appear in Adv. Appl. Clifford Algebras.
%
\bibitem{Var01} V.\,V. Varlamov, {\it Discrete Symmetries and Clifford
Algebras}. Int. J. Theor. Phys. {\bf 40} (2001) 769--805;
arXiv:math-ph/0009026 (2000).
%
\bibitem{Var11} V.\,V. Varlamov, \textit{CPT Groups of Higher Spin Fields}.
Int. J. Theor. Phys. \textbf{51} (2012) 1453--1481; arXiv: 1107.4156
[math-ph] (2011).
%
\bibitem{Rad22} J. Radon, {\it Lineare Scharen orthogonaler
Matrizen}. Abh. Math. Seminar Hamburg. {\bf 1} (1922) 1--24.
%
\bibitem{Hur23} A. Hurwitz, {\it Uber die Komposition der quadratischen
Formen}. Math. Ann. {\bf 88} (1923) 1--25.
%
\bibitem{AF02} R. Ab\l amowicz, B. Fauser, \textit{On the transposition anti-involution
in real Clifford algebras II: Stabilizer groups of primitive
idempotents}. Linear and Multilinear Algebra \textbf{59}(12) (2011)
1359--1381; arXiv: 1005.3558 [math-ph] (2010).
%
\bibitem{FRO90a} V.\,L. Figueiredo, W.\,A. Rodrigues, Jr., E.\,C. Oliveira,
{\it Covariant, algebraic, and operator spinors}. Int. J. Theor. Phys. {\bf 29} (1990) 371--395.
%
\bibitem{HS84} D. Hestenes, G. Sobczyk, {\it Clifford Algebra to Geometric
Calculus}. Dordrecht, Reidel, 1984.
%
\bibitem{RF90} W.\,A. Rodrigues, Jr., V.\,L. Figueiredo, {\it Real
spin-Clifford bundle and the spinor structure of the spacetime}.
Int. J. Theor. Phys. {\bf 29} (1990) 413--424.
%
\bibitem{RSVL} W.\,A. Rodrigues, Jr., Q.\,A.\,G. de Souza, J. Vaz, Jr.,
P. Lounesto, {\it Dirac-Hestenes spinor fields in Riemann-Cartan spacetime}.
Int. J. Theor. Phys. {\bf 35} (1996) 1849--1900.
%
\bibitem{Keller} J. Keller, {\it The geometric content of the electron
theory}, Adv. Appl. Clifford Algebras. \textbf{3}(2) (1993) 147--200;
Adv. Appl. Clifford Algebras. \textbf{9}(2) (1999) 309--395.
%
\bibitem{Var99a} V.\,V. Varlamov, {\it Generalized Weierstrass
representation for surfaces in terms of Dirac-Hestenes spinor field}.
J. Geometry and Physics {\bf 32}(3) (2000) 241--251; arXiv:math/9807152 (1998).
%
\bibitem{BT88} P. Budinich, A. Trautman, {\it The Spinorial
Chessboard}. Springer, Berlin, 1988.
%
\bibitem{Che55} C. Chevalley, {\it The construction and study of certain
important algebras}. Publications of Mathematical Society of Japan
No 1, Herald Printing, Tokyo, 1955.
%
\bibitem{Kar79} M. Karoubi, {\it K-Theory. An Introduction}. Springer-Verlag,
Berlin, 1979.
%
\bibitem{Wal64} C.\,T.\,C. Wall, {\it Graded Brauer Groups}. J. Reine Angew. Math. {\bf 213} (1964) 187--199.
%
\bibitem{Rash} P. K. Rashevskii, {\it The Theory of Spinors}. (in Russian)
Uspekhi Mat. Nauk {\bf 10} (1955), 3--110; English translation in
Amer. Math. Soc. Transl. (Ser.~2) {\bf 6} (1957) 1.
%
\bibitem{Port} I.\,R. Porteous, {\it Topological Geometry}.
van Nostrand, London, 1969.
%
\bibitem{Lou81} P. Lounesto, {\it Scalar Products of Spinors and an
Extension of Brauer-Wall Groups}. Found. Phys. {\bf 11} (1981)
721--740.
%
\bibitem{AtBSh} M.\,F. Atiyah, R. Bott, A. Shapiro, {\it Clifford
modules}. Topology {\bf 3}, Suppl. 1, (1964) 3--38.
%
\bibitem{Mandel} B.\,B. Mandelbrot, \textit{The Fractal Geometry of Nature}. Freeman, New York, 1977.
%
\bibitem{Var03a} V.\,V. Varlamov, \emph{General Solutions of Relativistic
Wave Equations}. Int. J. Theor. Phys. \textbf{42} (2003) 583--633; arXiv:math-ph/0209036 (2002).
%
\bibitem{Var04e} V.\,V. Varlamov, \emph{Relativistic wavefunctions on the Poincar\'{e}
group}. J. Phys. A: Math. Gen. \textbf{37} (2004) 5467--5476; arXiv:math-ph/0308038 (2003).
%
\bibitem{Var05b} V.\,V. Varlamov, \emph{Maxwell field on the Poincar\'{e}
group}. Int. J. Mod. Phys. A. {\bf 20} (2005) 4095--4112;
arXiv:math-ph/0310051 (2003).
%
\bibitem{Var06} V.\,V. Varlamov, \emph{Relativistic spherical functions
on the Lorentz group}. J. Phys. A: Math. Gen. \textbf{39} (2006) 805--822; arXiv:math-ph/0507056 (2005).
%
\bibitem{Var07} V.\,V. Varlamov, \emph{Spherical functions
on the de Sitter group}. J. Phys. A: Math. Theor. \textbf{40} (2007) 163--201; arXiv:math-ph/0604026 (2006).
%
\bibitem{Var07b} V.\,V. Varlamov, \emph{General Solutions of Relativistic
Wave Equations II: Arbitrary Spin Chains}. Int. J. Theor. Phys.
\textbf{46} (2007) 741--805; arXiv:math-ph/0503058 (2005).
%
\bibitem{Pen68} R. Penrose, \textit{Structure of space-time}. Benjamin, New York-Amsterdam, 1968.
%
\bibitem{Pen} R. Penrose, W. Rindler, \textit{Spinors and space-time. Vol. 2. Spinor and twistor methods in
space-time geometry}. Cambridge Monographs on Mathematical Physics, 2nd edn. Cambridge University Press, Cambridge and New York, 1988.
%
\bibitem{Wae32} B.\,L. van der Waerden, {\it Die Gruppentheoretische Methode in der
Quantenmechanik}. Springer, Berlin, 1932.
%
\bibitem{Var05c} V.\,V. Varlamov, \emph{The CPT Group  in  the de Sitter
Space}. Annales de la Fondation Louis de Broglie. \textbf{29} (2004) 969--987; arXiv:math-ph/0406060 (2004).
%
\bibitem{Var05} V.\,V. Varlamov, \emph{$CPT$ groups for spinor field in
de Sitter space}. Phys. Lett. B \textbf{631} (2005) 187--191; arXiv:math-ph/0508050 (2005).
%
\bibitem{PS1} S.\,M. Paneitz, I.\,E. Segal, \emph{Analysis in space-time bundles. I. General considerations and the scalar bundle}. J. Funct. Anal. \textbf{47} (1982) 78--142.
%
\bibitem{PS2} S.\,M. Paneitz, I.\,E. Segal, \emph{Analysis in space-time bundles. II. The spinor and form bundles}. J. Funct. Anal. \textbf{49} (1982) 335--414.
%
\bibitem{PS3} S.\,M. Paneitz, \emph{Analysis in space-time bundles. III. Higher spin bundles}. J. Funct. Anal. \textbf{54} (1983) 18--112.
%
\bibitem{Pan} S.\,M. Paneitz, I.\,E. Segal, D.\,A. Vogan, Jr., \emph{Analysis in Space-Time Bundles. IV. Natural Bundles Deforming into and Composed of the Same Invariant Factors as the Spin and Form Bundles}. J. Funct. Anal. \textbf{75} (1987) 1--57.
%
\bibitem{Lev11} A.\,V. Levichev, {\it Pseudo-Hermitian Realization of the Minkowski World through the DLF-Theory}. Physica Scripta {\bf 83} (2011) 1--9.
%
\bibitem{GMS} I.\, M. Gel'fand, R.\,A. Minlos, Z.\,Ya. Shapiro, {\it Representations
of the Rotation and Lorentz Groups and their Applications}. Pergamon
Press, Oxford, 1963.
%
\bibitem{Rum36} Yu.\,B. Rumer, {\it Spinorial Analysis}. URSS, Moscow, 2010 [in Russian].
%
\bibitem{Roz55} B.\,A. Rozenfeld, {\it Geometry of Lie groups}. Dordrecht-Boston-London, 1997.
%
\bibitem{Wang} Z.\,S. Wang, Qian Liu, \textit{Geometric phase and spinorial representation of mixed state}. Phys. Lett. A \textbf{377} (2013) 3272--3278.
%
\bibitem{Zur03} W.\,H. Zurek, \textit{Decoherence, Einselection, and the Quantum Origins of the Classical}. Rev. Mod. Phys. {\bf 75} (2003) 715; arXiv:quant-ph/0105127 (2003).
\end{thebibliography}
\end{document}